%% file: StronglyFocusedGravitationalWaves.tex
\documentclass[cmp,numbook,envcountsame]{svjour} 
\usepackage{graphics}
\usepackage{times}
\usepackage{amsmath,amssymb}
\usepackage[dvips]{graphicx, color}
\usepackage{verbatim}
\usepackage{bbm}
\usepackage{longtable}

\DeclareMathOperator{\tr}{tr}
\DeclareMathOperator{\supp}{supp}

\newcommand{\dd}{\mathrm{d}}

\newcommand{\R}{\mathbb{R}}
\newcommand{\C}{\mathbb{C}} 

\newcommand{\N}{\mathbb{N}}
\newcommand{\p}{\partial}
\renewcommand{\b}[1]{\mathbf{#1}}
\renewcommand{\u}[1]{\underline{#1}}
\newcommand{\s}[1]{\displaystyle{\not}#1}

\newcommand{\secs}[1]{$\widehat{\text{\bf \small (S#1)}}$}

\newcommand{\iv}{q}
\newcommand{\smallSHS}{\b{c}_2}
\newcommand{\smallSHSprime}{\b{c}_2'}
\newcommand{\smallRho}{\b{c}_1}

\newcommand{\smallEE}{\b{c}_3}
\newcommand{\bigEE}{\b{c}_4}
\newcommand{\smallEEProof}{\b{c}_{\ast}}

\newcommand{\smallRxx}{\b{c}_6}
\newcommand{\smallYxx}{\b{c}_{7}}
\newcommand{\bigSobolev}{\b{c}_{8}}
\newcommand{\smallFinal}{\b{b}}
\newcommand{\smallFinalPrime}{\b{c}}

\newcommand{\anisotropic}{{\mathfrak A}}

\newcommand{\asympt}[1]{\s{#1}}

\newcommand{\shs}[1]{{\mathfrak #1 }}

\newcommand{\Minkowski}{\mathcal{M}}

\newcommand{\PhiOne}{e}
\newcommand{\PhiTwo}{\gamma}
\newcommand{\PhiThree}{w}
\newcommand{\PsiOne}{f}
\newcommand{\PsiTwo}{\omega}
\newcommand{\PsiThree}{z}
\newcommand{\PsiSharpOne}{s}
\newcommand{\PsiSharpTwo}{p}
\newcommand{\PsiSharpThree}{y}
\newcommand{\SpecPsiOne}{h}
\newcommand{\SpecPsiTwo}{\sigma}
\newcommand{\SpecPsiThree}{\ell}
\newcommand{\datafunc}{\eta}

\renewcommand{\labelitemi}{$\bullet$}

\newcommand{\range}{\mathcal{R}}
\newcommand{\ranges}{\widehat{\mathcal{R}}}

\spnewtheorem{convention}[theorem]{Convention}{\it}{\rm}

\newcommand{\FLIP}{\text{\bf \small Flip}}
\newcommand{\STRIP}{\text{\bf \small Strip}}
\newcommand{\source}{\text{\bf \small Src}}
\newcommand{\SUP}{\text{\bf \small Sup}}

\begin{document}

\rule{0pt}{0pt}
\vskip 40mm 
{\bf \Large
\noindent Strongly Focused Gravitational Waves
}
\vskip 4mm
\noindent {\bf Michael Reiterer, Eugene Trubowitz}
\vskip 2mm
\noindent {Department of Mathematics, ETH Zurich, Switzerland}
\vskip 4mm
\noindent {\bf Abstract:}
This paper contains a new proof of the formation of trapped spheres, in vacuum spacetimes, by the focusing of gravitational waves, from generic data. The first such result was obtained by Christodoulou \cite{Chr}. We exploit the same physical mechanism, but give a logically independent construction of these spacetimes.

%
%
%
%
%
\input{SectionIntroduction.tex}
\input{SectionReformulation.tex}
\input{SectionSymmetries.tex}
\input{SectionMinkowski.tex}
\input{SectionAnsatz.tex}
\input{SectionFormalSolution.tex}
\input{SectionEnergyEstimates.tex}
\input{SectionClassicalSolution.tex}

\appendix 

\section{Program for Proposition \ref{prop:ri1}} \label{app:appprog}
The Mathematica\footnote{http://www.wolfram.com/mathematica/} code below generates the equations in Proposition \ref{prop:ri1}. Notation:
{\small \verb+F[1,2]+} is ${F_1}^2$,
{\small \verb+\[CapitalGamma][1,1,2]+} is $\Gamma_{112}$,
{\small \verb+W[3,2,3,4]+} is $W_{3234}$,
{\small \verb+T[3,4,1]+} is ${T_{34}}^1$,
{\small \verb+Fr[2,W[1,2,3,4]]+} is $F_2(W_{1234})$,
{\small \verb+ctr[W,{2,3},{4,3,4},1]+} is ${\Gamma_{23}}^mW_{m434}$,
{\small \verb+ctr[W,{2,3},{4,3,4},3]+} is ${\Gamma_{23}}^mW_{43m4}$,
{\small \verb+C+} is conjugation,
{\small \verb+\[GothicCapitalA]+} is $\anisotropic$.
\vskip 6mm

{\small
\begin{verbatim}
Module[{C,\[CapitalPhi],F,\[CapitalGamma],W,T,U,V,Fr,M,E,u,
\[CapitalPsi],ebold,\[Lambda]bold,\[Rho],\[GothicCapitalA],S,\[Lambda]},

I1={1,2,3,4}; I1T[n_]:=Tuples[I1,n];
I2={{1,2},{3,1},{3,2},{4,1},{4,2},{3,4}};

Module[{P=\[CapitalPhi]},FMatrix={{P[1],P[2],0,0},{C[P[1]],C[P[2]],0,0},
{P[4],P[5],0,1},{0,0,P[3],0}};\[CapitalGamma]Matrix={{P[8]+C[P[9]],
C[P[12]],P[11],P[6],P[7],P[8]-C[P[9]]},{-P[9]-C[P[8]],P[11],P[12],P[7],
C[P[6]],-P[9]+C[P[8]]},{P[13]-C[P[13]],0,0,-P[8]+C[P[9]],P[9]-C[P[8]],
P[13]+C[P[13]]},{0,C[P[10]],P[10],0,0,0}};WMatrix={{P[16]+C[P[16]],
C[P[17]],-P[17],P[15],-C[P[15]],P[16]-C[P[16]]},{C[P[17]],C[P[18]],0,0,
-C[P[16]],-C[P[17]]},{-P[17],0,P[18],-P[16],0,-P[17]},{P[15],0,-P[16],
P[14],0,P[15]},{-C[P[15]],-C[P[16]],0,0,C[P[14]],C[P[15]]},
{P[16]-C[P[16]],-C[P[17]],-P[17],P[15],C[P[15]],P[16]+C[P[16]]}};];

Map[Set@@#&,Solve[Flatten[{Table[F[i,j],{i,I1},{j,I1}]==FMatrix,
Table[\[CapitalGamma]@@Join[{i},j],{i,I1},{j,I2}]==
\[CapitalGamma]Matrix, Table[W@@Join[i,j],{i,I2},{j,I2}]==WMatrix,
Table[\[CapitalGamma]@@a==-\[CapitalGamma]@@a[[{1,3,2}]],{a,I1T[3]}],
Table[W@@a==-W@@a[[{2,1,3,4}]]==-W@@a[[{1,2,4,3}]],{a,I1T[4]}]}],
Flatten[{Table[F@@a,{a,I1T[2]}],Table[\[CapitalGamma]@@a,{a,I1T[3]}],
Table[W@@a,{a,I1T[4]}]}]][[1]]];

ctr[b_,m_,n_,i_]:=Sum[s[[1]]*\[CapitalGamma]@@Insert[m,s[[2]],3]
*b@@Insert[n,s[[3]],i],{s,{{1,1,2},{1,2,1},{-1,3,4},{-1,4,3}}}];

T[a_,b_,m_]:=
Fr[b,F[a,m]]-Fr[a,F[b,m]]+ctr[F,{a,b},{m},1]-ctr[F,{b,a},{m},1];
U[k_,l_,a_,b_]:=
-W[k,l,a,b]+Fr[a,\[CapitalGamma][b,l,k]]-Fr[b,\[CapitalGamma][a,l,k]]
+ctr[\[CapitalGamma],{b,l},{a,k},2]-ctr[\[CapitalGamma],{a,l},{b,k},2]
-ctr[\[CapitalGamma],{a,b},{l,k},1]+ctr[\[CapitalGamma],{b,a},{l,k},1];
V[b_,j_,k_]:=Sum[s[[1]]*Module[{a=s[[2]],i=s[[3]]},Fr[i,W[a,b,j,k]]
-ctr[W,{i,a},{b,j,k},1]-ctr[W,{i,b},{a,j,k},2]-ctr[W,{i,j},{a,b,k},3]
-ctr[W,{i,k},{a,b,j},4]],{s,{{1,1,2},{1,2,1},{-1,3,4},{-1,4,3}}}];

M={2,2,2,3,3,1,2,2,2,2,2,2,3,1,2,3,4,4};
MapThread[(\[CapitalPhi][#1]=#2+Power[u,-#3]*\[CapitalPsi][#1])&,
{Range[1,18],{ebold,I*ebold,\[Rho],0,0,0,Power[\[GothicCapitalA],2],
\[Lambda]bold,C[\[Lambda]bold],0,-1,0,0,0,0,0,0,0}/\[Rho],M}];

C[a_+b_]:=C[a]+C[b]; C[a_*b_]:=C[a]*C[b];
C[Power[a_,n_Integer]]:=Power[C[a],n]; C[C[x_]]:=x; 
C[x_?NumericQ]:=Conjugate[x]; C[x:\[Rho]|u|ebold]:=x;
Fr[a_,b_*c_]:=b*Fr[a,c]+c*Fr[a,b]; Fr[_,_?NumericQ|\[GothicCapitalA]]=0;
Fr[a_,Power[b_,n_Integer]]:=n*Power[b,n-1]*Fr[a,b];
Fr[a_,b_+c_]:=Fr[a,b]+Fr[a,c]; Fr[1|2|4,u]=0; Fr[3,u]=1;
Fr[n:1|2|3|4,C[b_]]:=C[Fr[{2,1,3,4}[[n]],b]]; 
Fr[4,ebold|\[Lambda]bold]=0; Fr[1|2,\[Rho]]=0; Fr[3,\[Rho]]=-1;
Fr[4,\[Rho]]=\[CapitalPhi][3]*Power[\[GothicCapitalA],2];
\[Lambda][i_]:=Power[u,2*i];

Format[Fr[n:1|2|3|4,x_]]:=({"D",OverBar["D"],"N","L"}[[n]])[x];
Format[\[GothicCapitalA]]="\[GothicCapitalA]"; Format[u]="u";
Format[ebold]=Style["e",Bold]; Format[C[x_]]:=OverBar[x];
Format[\[Lambda]bold]=Style["\[Lambda]",Bold]; Format[S]="S";
Format[\[CapitalPsi][i_]]:=Piecewise[{{Subscript["f",i],1<=i<=5},
{Subscript["\[Omega]",i-5],6<=i<=13},{Subscript["z",i-13],14<=i<=18}}];

E={4,4,4,6,6,2,4,4,4,4,4,4,6,0,0,0,0,0};
Collect[Expand[Expand[Power[u,E-M]*{T[1,4,1],T[1,4,2],T[4,3,3],T[3,4,1],
T[3,4,2],U[1,4,4,1],U[2,4,4,1],(U[2,1,4,1]+U[4,3,4,1])/2,(U[1,2,4,2]
+U[3,4,4,2])/2,U[2,3,3,4],U[1,3,3,2],U[2,3,3,2],(U[1,2,3,4]
+U[3,4,3,4])/2,\[Lambda][1]*V[1,1,4],\[Lambda][1]*V[4,1,4]+
\[Lambda][2]*V[3,4,1],\[Lambda][2]*V[2,4,1]+\[Lambda][3]*V[1,3,2],
\[Lambda][3]*V[4,3,2]+\[Lambda][4]*V[3,2,3],\[Lambda][4]*V[2,2,3]}]
/.{\[Rho]->1/(-1/u+S/Power[u,2])}
][[{1,2,4,5,6,7,8,9,13,3,10,11,12,14,15,16,17,18}]],u]//MatrixForm
]
\end{verbatim}

\end{document}

%% file: SectionIntroduction.tex
\section{Introduction}
The vacuum spacetimes $(M,g)$ constructed in this paper have the following properties:
\vskip 1mm
\begin{center}
\input{vaidya.pstex_t}
\end{center}
\begin{enumerate}
\item $M = \big\{ (x^{\mu})_{\mu=1,2,3,4} = (\xi^1,\xi^2,\u{u},u) \in \R^2 \times (0,1) \times (-\infty,u_0)\big\}$ with $u_0 < 0$.
The plane $\R^2$ is identified with a punctured $S^2$ by stereographic projection.
The metric $g$ extends smoothly to $S^2 \times (0,1) \times (-\infty,u_0)$.
\item $g^{-1} = g^{ab}F_a\otimes F_b$, where $F_a = {F_a}^{\mu} \tfrac{\p}{\p x^{\mu}}$ are complex vector fields and
$$
({F_a}^{\mu}) = \begin{pmatrix}
\ast & \ast & 0 & 0 \\
\ast & \ast & 0 & 0 \\
\ast & \ast & 0 & 1 \\
0 & 0 & \ast & 0 
\end{pmatrix}
\qquad (g^{ab}) = \begin{pmatrix}
0 & 1 & 0 & 0 \\
1 & 0 & 0 & 0 \\
0 & 0 & 0 & -1\\
0 & 0 & -1 & 0
\end{pmatrix}
$$
The four tuple $F = (F_a)$ is a complex frame. The indices $a,b = 1,2,3,4$. Asterisk entries $\ast$ can be nonzero. Some are nonzero, because $F$ is a frame.
We require: $F_3+F_4$ is future-directed, ${F_4}^3 > 0$ and $$\overline{(F_1,F_2,F_3,F_4)} = (F_2,F_1,F_3,F_4)$$ Thus, $g^{-1} = g^{ab} F_a \otimes F_b$ is real and has signature $(-,+,+,+)$.
\item
For all $0 < \u{u} <1/2$,
\begin{equation*}
({F_a}^{\mu}) = \begin{pmatrix}
\rho^{-1} \mathbf{e} &\;&+i\rho^{-1} \mathbf{e} &\;& 0 && 0 \\
\rho^{-1} \mathbf{e} && -i\rho^{-1} \mathbf{e} && 0 && 0 \\
0 && 0 && 0 && 1 \\
0 && 0 && 1 && 0 
\end{pmatrix}\hskip 14mm \begin{aligned}\rho & = \u{u} - u \\ \mathbf{e} & = \tfrac{1}{2}\big(1+(\xi^1)^2 + (\xi^2)^2\big)\end{aligned}
\end{equation*}
The subset $0 < \u{u} < 1/2$ of $M$ is isometric to a subset of Minkowski spacetime.
The isometry is given by $\u{u} = 2^{-1/2}(t+r)$, $u = 2^{-1/2}(t-r)$, with standard Minkowski time $t$, radius $r$, and standard stereographic coordinates $(\xi^1,\xi^2)$. Thus, the past light cone $t + r \leq 0$ in Minkowski spacetime can be smoothly attached to $M$ along $\u{u}=0$.
\item There is a map $\text{\small \bf DATA}: \R^2 \times (0,1) \to \C$ with
$$
\lim_{u \to -\infty} u^2({F_1}^1 - \rho^{-1}\mathbf{e})
= 
\lim_{u \to -\infty} i u^2({F_1}^2 - i\rho^{-1}\mathbf{e})
= 
\mathbf{e} \textstyle\int_0^{\u{u}} \dd \u{u}'\, \text{\small \bf DATA}(\xi^1,\xi^2,\u{u}')
$$
The limits are taken at constant $\xi^1$, $\xi^2$, $\u{u}$. Here, $u \to -\infty$ is interpreted as \emph{past null infinity}, and $\text{\small \bf DATA}$ as initial data at past null infinity. It vanishes when $\u{u} < \tfrac{1}{2}$. Informally, $\text{\small \bf DATA}$ describes the incoming radiation. $\text{\small \bf DATA}$  uniquely determines the metric on $M$,
if $\lim_{u\to-\infty} {F_4}^3=1$ and we make technical assumptions about the decay  as $u\to -\infty$. These assumptions are not discussed in this introduction.
\item 
The uniqueness statement in the last item is possible because the gauge has already been completely fixed\footnote{Actually, there is still the local $U(1)$ gauge degree of freedom $(F_1,F_2) \to (e^{+i\theta} F_1, e^{-i\theta} F_2)$. It is fixed by a transport equation, but this is not discussed in this introduction.}, by the asterisk pattern of $({F_a}^{\mu})$. The three zeros in the third (resp. fourth) column imply that $\u{u}$ (resp. $u$) is a solution to the eikonal equation
$g^{ab}F_a(\u{u})F_b(\u{u}) = 0$ (resp. $g^{ab}F_a(u)F_b(u) = 0$)
and that its gradient\footnote{The gradient of $f: M \to \R$ is the vector field $g^{ab}F_a(f) F_b$.} is $-{F_4}^3F_3$ (resp. $-F_4$). The two zeros in the lower-left corner imply that $\xi^1,\xi^2$ are transported along $F_4$. \emph{Thus, two eikonal equations, two transport equations}, together with $\lim_{u\to-\infty} {F_4}^3=1$ and the Minkowskian data for $0 < \u{u} < \tfrac{1}{2}$, fix the coordinates.
\item There is an explicit expansion in powers of $\tfrac{1}{u}$ for the frame,
for $u$ large negative, with rigorous error estimates.
\item The intersections of level sets of $u$ and $\u{u}$ are spacelike 2-dimensional spheres. For a generic class of initial data, some of them are trapped, in the sense that ``the two systems of null geodesics which meet [the surface] orthogonally converge locally in future directions at [the surface]'' \cite{Pen}. The existence of trapped spheres is read off from the lowest order term of the $\tfrac{1}{u}$ expansion.
\item Let $\|\text{\small \bf DATA}\|$ be a $C^{10}$ norm\footnote{This norm takes the two patches on $S^2$ into account.} of $\text{\small \bf DATA}$. The analysis is tailored to large $\|\text{\small \bf DATA}\|$.
However, the bigger $\|\text{\small \bf DATA}\|$, the smaller the interval $u \in (-\infty,u_0)$, with $u_0<0$, that we control. More precisely,
 $|u_0| \sim \|\text{\small \bf DATA}\|$ as 
$\|\text{\small \bf DATA}\| \to \infty$.
\end{enumerate}
The discussion above is framed in what we call the \emph{High Amplitude Picture}. Another picture, the \emph{Regularized Picture}, is better suited for giving proofs, and is used in almost the entire paper, including the statement of the main Theorem \ref{mainenergy}. The two pictures are related by a symmetry transformation\footnote{A symmetry of the vacuum Einstein equations in the frame formalism of Newman, Penrose, Friedrich.}, see Section \ref{subsec:tpovXXXXX}. Modulo this transformation, Theorem \ref{mainenergy} is precise about the technical issues that were glossed over in the above discussion. A third picture is introduced to compare the present paper with \cite{Chr}, here referred to as the \emph{Finite Mass Picture}.

Christodoulou \cite{Chr} implements the focusing of gravitational waves by a dedicated geometric optics argument, that he calls the `short pulse method'. It consists of an appropriate geometric setup, not unlike Vaidya's \cite{Vai}, and a self-consistent scheme of bounds for all the unknown quantities in terms of a small parameter $\delta > 0$. Many bounds are of the singular form $\mathcal{O}(\delta^{-\alpha})$, with $\alpha > 0$.

Our results are stronger than \cite{Chr} in the sense that we construct semi-global solutions (all the way to past null infinity)\footnote{Christodoulou underscores the importance of the semi-global limit \cite[p. 4]{Chr}: ``[...] the physically interesting problem is the problem where the initial conditions are of arbitrarily \emph{low compactness}, that is, arbitrarily far from already containing closed trapped surfaces [...]''.}, obtain an explicit asymptotic expansion, and control the solution for a longer time. Our results are slightly weaker from the point of view of regularity\footnote{The amplitude $\|\text{\small \bf DATA}\|$ can be taken to be $C^7$ in \cite{Chr}, as opposed to $C^{10}$ in this paper.}. However, the main point of this paper is its methodology. Here, it yields a new and very short construction of semi-global solutions to the vacuum Einstein equations with trapped spheres that form in evolution.

We use the orthonormal frame formalism of Newman and Penrose \cite{NP}, and the key idea of hyperbolic reduction to  symmetric hyperbolic systems due to Friedrich \cite{Fr}.
Our hyperbolic reduction is new, directly in a double null gauge (Section \ref{sec:geomety}).
The equations are brought into a relevant\,/\,irrelevant form that exhibits the essential constituents that have to be treated carefully, and sweeps everything else into `generic terms' that one doesn't need to know much about (Section \ref{sec:ansatz} and the program in Appendix \ref{app:appprog}). Formal power series solutions in $\tfrac{1}{u}$ (Section \ref{sec:formalsolution}) and
energy estimates (Section \ref{sec:energyestimates}) are combined to construct classical solutions (Section \ref{sec:classicalvacuumfields}). Under certain generic conditions, they contain trapped spheres (Section \ref{sec:conclusionsXX}). 

The initial value problem with data along the asymptotic characteristic surfaces $u = -\infty$, see above, is solved by taking the limit of a sequence of solutions to initial value problems with data on the spacelike level sets of $u + \u{u}$. The elements of the sequence have no direct physical meaning, because they do not necessarily satisfy the constraint equations. However, the semiglobal limit does.

A longer version of this paper, in which even standard proofs are written out in all detail, can be found on the \emph{arxiv} preprint server\footnote{http://arxiv.org/abs/0906.3812 (version v1, June 2009)}. The \emph{arxiv} version in addition includes a proof of a `Minkowski to Schwarzschild transition', and introduces an expansion, more powerful than the $\tfrac{1}{u}$ expansion, that can be used to continue the solutions further.

We appreciate the great effort that Demetrios Christodoulou invested over many years to nurture the mathematical study of general relativity at ETH Zurich.

We thank Lydia Bieri, Joel Feldman, Horst Kn\"orrer and Martin Lohmann for encouragement and helpful conversations.

%% file: vaidya.pstex_t
\begin{picture}(0,0)%
\includegraphics{vaidya.pstex}%
\end{picture}%
\setlength{\unitlength}{3158sp}%
\begingroup\makeatletter\ifx\SetFigFont\undefined%
\gdef\SetFigFont#1#2#3#4#5{%
  \reset@font\fontsize{#1}{#2pt}%
  \fontfamily{#3}\fontseries{#4}\fontshape{#5}%
  \selectfont}%
\fi\endgroup%
\begin{picture}(3112,3204)(2229,-7069)
\put(2401,-5686){\makebox(0,0)[lb]{\smash{{\SetFigFont{10}{12.0}{\familydefault}{\mddefault}{\updefault}{\color[rgb]{0,0,0}$\text{Minkowski}$}%
}}}}
\put(3639,-5716){\rotatebox{315.0}{\makebox(0,0)[lb]{\smash{{\SetFigFont{10}{12.0}{\rmdefault}{\mddefault}{\updefault}{\color[rgb]{0,0,0}$\u{u}=0$}%
}}}}}
\put(4088,-5251){\rotatebox{315.0}{\makebox(0,0)[lb]{\smash{{\SetFigFont{10}{12.0}{\rmdefault}{\mddefault}{\updefault}{\color[rgb]{0,0,0}$\u{u}=1$}%
}}}}}
\put(4231,-5057){\rotatebox{45.0}{\makebox(0,0)[lb]{\smash{{\SetFigFont{10}{12.0}{\rmdefault}{\mddefault}{\updefault}{\color[rgb]{0,0,0}$u=-\infty$}%
}}}}}
\put(3204,-4157){\rotatebox{45.0}{\makebox(0,0)[lb]{\smash{{\SetFigFont{10}{12.0}{\rmdefault}{\mddefault}{\updefault}{\color[rgb]{0,0,0}$u=u_0$}%
}}}}}
\put(3191,-4957){\makebox(0,0)[lb]{\smash{{\SetFigFont{10}{12.0}{\rmdefault}{\mddefault}{\updefault}{\color[rgb]{0,0,0}$\text{M}$}%
}}}}
\put(2529,-7013){\rotatebox{45.0}{\makebox(0,0)[lb]{\smash{{\SetFigFont{10}{12.0}{\familydefault}{\mddefault}{\updefault}{\color[rgb]{0,0,0}$\text{past null infinity}$}%
}}}}}
\put(4426,-4036){\makebox(0,0)[lb]{\smash{{\SetFigFont{10}{12.0}{\familydefault}{\mddefault}{\updefault}{\color[rgb]{0,0,0}$\text{may contain trapped spheres}$}%
}}}}
\put(5326,-6061){\makebox(0,0)[lb]{\smash{{\SetFigFont{10}{12.0}{\familydefault}{\mddefault}{\updefault}{\color[rgb]{0,0,0}$\text{initial data (= {\bf \small DATA})}$}%
}}}}
\end{picture}%

%% file: SectionReformulation.tex
\section{Gauge fixing and hyperbolic reduction} \label{sec:geomety}
This section uses the Newman-Penrose-Friedrich formalism, see \cite{NP} and \cite{Fr}.
\begin{convention} Small Latin frame indices and small Greek coordinate indices run from 1 to 4. Pairs such as $(ab)$ run over the ordered sequence
$(12),\hskip-1pt(31),\hskip-1pt(32),\hskip-1pt(41),\hskip-1pt(42),\hskip-1pt(34)$.
\end{convention}
\begin{definition}[Parametrization] \label{def:paramparam} Introduce the vector space of real dimension 31:
\begin{equation*} 
\range = \big\{(\PhiOne,\PhiTwo,\PhiThree) \in \C^5 \oplus \C^8 \oplus \C^5\; \big| \; \PhiOne_3,\, \PhiOne_4,\, \PhiOne_5,\, \PhiTwo_2,\,\PhiTwo_6 \in \R \,\big\}
\end{equation*}
Let $\mathrm{Param}: \range \ni (e,\gamma,w) \mapsto (F,\Gamma,W)$ be given by
\begin{align*}
\big( {F_a}^\mu\big)
 &= 
\begin{pmatrix} e_1 & e_2 & 0 & 0 \\ 
\overline{e}_1 & \overline{e}_2 & 0 & 0 \\
e_4 & e_5 & 0 & 1 \\
0 & 0 & e_3 & 0 
\end{pmatrix}\\ \displaybreak[0]
\big(\Gamma_{a(bc)}\big)
 &=   \input{matrixG.tex}\\ \displaybreak[0]
\big( W_{(ab)(cd)}\big)
 &= \input{matrixW.tex}
\end{align*}
and $\Gamma_{abc} = - \Gamma_{acb}$ and $W_{abcd} = - W_{bacd} = - W_{abdc}$.
\end{definition}
\begin{definition}
\begin{equation}\label{gabgab} 
\big(\b{g}_{ab}\big)  = 
\begin{pmatrix}
0 & 1 & 0 & 0\\
1 & 0 & 0 & 0\\
0 & 0 & 0 & - 1\\
0 & 0 & - 1 & 0
\end{pmatrix}
\qquad 
\big(\b{g}^{ab}\big)  =
\begin{pmatrix}
0 & 1 & 0 & 0\\
1 & 0 & 0 & 0\\
0 & 0 & 0 & - 1\\
0 & 0 & - 1 & 0
\end{pmatrix}
\end{equation}
\end{definition}
\begin{convention} Indices are raised and lowered with $(\b{g}_{ab})$ and its matrix inverse $(\b{g}^{ab})$.
\end{convention}
\begin{remark} \label{remehekhee}
$W_{abcd} = W_{cdab}$ and $W_{ijka} + W_{jkia} + W_{kija} = 0$ and
${W^a}_{iaj} = 0$.
Equivalently, $W$ has the algebraic symmetries of a Weyl tensor with respect to $(\b{g}_{ab})$.
And
${F_{a'}}^{\mu} = \overline{{F_a}^{\mu}}$,
$\Gamma_{a'b'c'} = \overline{\Gamma_{abc}}$,
$W_{a'b'c'd'} = \overline{W_{abcd}}$ with $(1',2',3',4') = (2,1,3,4)$.
\end{remark}
\begin{definition}[Vacuum field] \label{def:vacvacvac}
Let $\mathcal{U}\subset \R^4$ be open and $(x^{\mu})$ standard Cartesian coordinates on $\R^4$. 
A sufficiently differentiable field $\Phi = (e,\gamma,w) : \mathcal{U}\to \mathcal{R}$ is a \emph{vacuum field} iff $(F,\Gamma,W) = \mathrm{Param}\circ (e,\gamma,w)$ satisfies (a), (b), (c):
\begin{enumerate}
\item[(a)] $(F_a) = ({F_a}^{\mu}\tfrac{\p}{\p x^{\mu}})$ is a complex frame on $\mathcal{U}$, that is $\det ({F_a}^{\mu}) \neq 0$. And ${F_4}^3 > 0$.
\end{enumerate}
Let $g$ be the real, Lorentzian metric with real\footnote{Reality follows from $\overline{(F_1,F_2,F_3,F_4)} = (F_2,F_1,F_3,F_4)$.} inverse $g^{-1} = \b{g}^{ab}F_a \otimes F_b$.
\begin{enumerate}
\item[(b)] The connection given by $\nabla_{F_a}F_b = {\Gamma_{ab}}^cF_c$ is the Levi-Civita connection of $g$. 
\item[(c)] $g([\nabla_{F_j},\nabla_{F_k}] F_b - \nabla_{[F_j,F_k]}F_b,\,F_a) = W_{abjk}$.
\end{enumerate}
Observe that (a), (b), (c)
and Remark \ref{remehekhee} imply that $g$ is Ricci-flat.
\end{definition}
\begin{proposition}[Geometric interpretation of the gauge fixed by a vacuum field] \label{prop:fdkhfdkhfdkh} Suppose $\Phi = (e,\gamma,w)$ is a vacuum field.
Throughout this paper, we use the notation
$$(x^{\mu}) = (\xi^1,\xi^2,\u{u},u)
\qquad (F_a) = (D,\overline{D},N,L)
$$
Then, with respect to the metric $g$ and its Levi-Civita connection $\nabla$:\\
\emph{Part 1.} $u$ and $\u{u}$ solve the eikonal equation $g^{-1}(\dd u,\dd u) = g^{-1}(\dd \u{u},\dd \u{u}) = 0$ with $g^{-1}(\dd u,\dd \u{u}) < 0$. The real vector field $L$ is minus the gradient of $u$, that is $L = -g^{-1}(\dd u,\,\cdot\,)$, and $e_3 = L(\u{u}) > 0$. The real vector field $e_3N$ is minus the gradient of $\u{u}$, that is
 $e_3N = -g^{-1}(\dd \u{u},\,\cdot\,)$. The coordinates $\xi^1,\xi^2$ satisfy the transport equations $L(\xi^1) = L(\xi^2) = 0$. The complex vector field $D$ is such that $2^{1/2}(\Re D,\Im D)$ is a real orthonormal frame for the spacelike $\ker \dd u \cap \ker \dd \u{u}$, and $D$ satisfies the transport equation $g(\nabla_L D,\overline{D}) = 0$.\\
\emph{Part 2.} Declare $N+L$ to be future-directed and let $S_{\u{u},u}$ be the intersection of the level sets of $\u{u}$ and $u$. The traces of the future-directed second fundamental forms of 
 $S_{\u{u},u}$ relative to $L$ and $N$ are  $g(\nabla_D L,\overline{D}) + g(\nabla_{\overline{D}}L,D) = 2\PhiTwo_2$ and
$g(\nabla_D N,\overline{D}) + g(\nabla_{\overline{D}}N,D) = 2\PhiTwo_6$. Recall that
$\PhiTwo_2$ and $\PhiTwo_6$ are real. The traceless parts of the second fundamental forms are determined by the complex $\PhiTwo_1$ and $\PhiTwo_7$.
\end{proposition}
\begin{proof} 
These are transcriptions of different properties of the matrices in Definition \ref{def:paramparam}, in the context of Definition \ref{def:vacvacvac}. For instance, $u$ solves the Eikonal equation, because
$g^{-1}(\dd u,\dd u) = \mathbf{g}^{ab}F_a(u)F_b(u) = 2{F_1}^4{F_2}^4 - 2{F_3}^4{F_4}^4 = 0.$ \qed
\end{proof}
\begin{remark}\label{trappedpenrose}
$S_{\u{u},u}$ is locally \emph{trapped}, in the sense of  \cite{Pen}, if and only if $\gamma_2,\gamma_6 < 0$.
\end{remark}
\begin{proposition}[Local realizability of the gauge fixed by a vacuum field]
Every Ricci-flat Lorentzian manifold is locally isometric to a pair $(\mathcal{U},g)$ 
as in Definition \ref{def:vacvacvac}.
\end{proposition}
\begin{proof} Given $g$ and its Levi-Civita connection $\nabla$ on a manifold, Proposition \ref{prop:fdkhfdkhfdkh} Part 1 is a step by step outline for the local construction of $u$, $\u{u}$, $L$, $e_3$, $N$, $\xi^1$, $\xi^2$, $D$. The construction is not unique, because one has to specify initial data for the eikonal and transport equations. The construction yields a field $\Phi$ with the desired properties. \qed
\end{proof}
\begin{definition}\label{def:fdhshskhd}
To any sufficiently differentiable $\Phi: \mathcal{U}\to \mathcal{R}$ associate $(T,U,V)$ by:
\begin{alignat*}{4}
{T_{a b}}^\mu  &=\,&&   F_b({F_a}^\mu)- F_a ({F_b}^\mu) + {\Gamma_{a b}}^c{F_c}^\mu  - {\Gamma_{b a}}^c {F_c}^\mu\\
U_{k \ell a b} &=\,&&  F_a(\Gamma_{b\ell k}) - F_b(\Gamma_{a\ell k}) + {\Gamma_{b \ell}}^m {\Gamma_{a m k}} - 
{\Gamma_{a\ell}}^m \Gamma_{b m k}\\
&\,&&\qquad - ({\Gamma_{a b}}^m - {\Gamma_{b a}}^m){\Gamma_{m \ell k} } - W_{k\ell ab}\\
V_{abijk}  &= && F_i(W_{abjk}) - {\Gamma_{ia}}^mW_{mbjk} - {\Gamma_{ib}}^m W_{amjk} - {\Gamma_{ij}}^m W_{abmk} - {\Gamma_{ik}}^m W_{abjm}  + \\  &&&  F_k(W_{abij}) - {\Gamma_{ka}}^mW_{mbij} - {\Gamma_{kb}}^m W_{amij} - {\Gamma_{ki}}^m W_{abmj} - {\Gamma_{kj}}^m W_{abim} + \\  &&&  F_j(W_{abki}) - {\Gamma_{ja}}^mW_{mbki} - {\Gamma_{jb}}^m W_{amki} - {\Gamma_{jk}}^m W_{abmi} - {\Gamma_{ji}}^m W_{abkm}
\end{alignat*}
Here $\Phi = (e,\gamma,w)$ and $(F,\Gamma,W) = \mathrm{Param} \circ (e,\gamma,w)$.
\end{definition}
\begin{proposition}[Equivalent characterization of a vacuum field] \label{prop,fdldfkf} $\Phi: \mathcal{U}\to \mathcal{R}$ is a vacuum field if and only if  $e_3 > 0$ and $\Im(e_1\overline{e}_2) \neq 0$ 
and $(T,U,V) = 0$.
\end{proposition}
\begin{proof}
First, $\det ({F_a}^{\mu}) = -2i e_3 \Im(e_1\overline{e}_2)$ and ${F_4}^3 = e_3$. Then
 $T=0$ iff Definition \ref{def:vacvacvac} (b) holds, $T$ is torsion;
 $U=0$ iff Definition \ref{def:vacvacvac} (c) holds, given (b);
 $V=0$ iff the Weyl field associated to $W$ satisfies the differential Bianchi identities.
 $V=0$ is not necessary in Proposition \ref{prop,fdldfkf}. However, it is used in the hyperbolic reduction below. Cf. \cite{Fr}. \qed
\end{proof}
\begin{remark} \label{rem:fdkhfdkhdkdhkdj}
${T_{ab}}^{\mu} = -{T_{ba}}^{\mu}$ and $U_{abk \ell} = - U_{ba k \ell} = - U_{ab \ell k}$ and $V_{kab} = -V_{kba}$ and ${V^k}_{kb}= 0$ and $V_{abc} + V_{bca} + V_{cab}= 0$, where $V_{bjk}=\b{g}^{ai} V_{abijk}$.
The components of $V_{abijk}$ all vanish if and only if the components of $V_{bjk}$ all vanish, pointwise on $\mathcal{U}$.
\end{remark}

\begin{definition}\label{def:kdhkueuuueze} Introduce the vector space of real dimension $32$:
$$\ranges = \big\{(t,u,v) \in \C^5 \oplus \C^9 \oplus \C^3\; \big| \; t_1,\, t_2 \in \R\,\big\}$$
\end{definition}
\begin{proposition} \label{consmeXX}
Let $\Phi: \mathcal{U}\to \mathcal{R}$ and $(T,U,V)$ be as in Definition \ref{def:fdhshskhd}.
Let $\lambda=(\lambda_1,\lambda_2,\lambda_3,\lambda_4)$ be strictly positive weight functions on $\mathcal{U}$. Then, there are unique fields
\begin{align*}
(\shs{t},\shs{u},\shs{v})=\big(\shs{t}(\Phi,\lambda),\shs{u}(\Phi,\lambda),\shs{v}(\Phi,\lambda)\big):
\quad  \mathcal{U} & \to \range \subset \C^5 \oplus \C^8 \oplus \C^5 \\
(t,u,v)=\big(t(\Phi,\lambda),u(\Phi,\lambda),v(\Phi,\lambda)\big):
\quad \mathcal{U} & \to \ranges \subset  \C^5 \oplus \C^9 \oplus \C^3
\end{align*}
whose components are quadratic polynomials in $\Phi,\,\tfrac{\p}{\p x^{\mu}} \Phi,\,\overline{\Phi},\,\tfrac{\p}{\p x^{\mu}} \overline{\Phi}$, such that
\setcounter{MaxMatrixCols}{15}
\begin{align*}
\big({T_{(ab)}}^{\mu}\big) &= \input{matrixT.tex}\\ \displaybreak[0]
\big(U_{(ab)(jk)}\big) & = \input{matrixUNEW.tex}\\ \displaybreak[0]
\big(V_{a(jk)}\big) &= \input{matrixVNEW.tex}
\end{align*}
In particular, $(T,U,V) = 0$ if and only if $(\shs{t},\shs{u},\shs{v}) = 0$ and $(t,u,v) = 0$.
\end{proposition}
\begin{proof}
$T$, $U$, $V$ lie pointwise in spaces of real dimension $24$, $36$, $16$, 
by Remark \ref{rem:fdkhfdkhdkdhkdj}. For all $\Phi$:
 ${T_{12}}^3 = 0$ (1 {r.e.}\footnote{$n$ {r.e.} is short for $n$ real equations.}); ${T_{31}}^3  = 0$ (2 {r.e.}); ${T_{ab}}^4 = 0$ (6 {r.e.}); $U_{3441} = U_{4134}$ (2 {r.e.}); $\Im U_{3132} = 0$ (1 {r.e.}); $\Im U_{4142} = 0$ (1 {r.e.}). Thus, $T$, $U$, $V$ lie in subspaces of dimension $24-9 = 15$, $36-4 = 32$, $16-0 = 16$. The matrices on the right hand sides lie in these subspaces. The linear map from $(\shs{t},\shs{u},\shs{v}) \oplus (t,u,v)$ to these matrices has maximal rank
$\dim_{\R} \range \oplus \ranges = 31+32 = 63$. Since this is equal to $15 + 32 + 16$, the fields $(\shs{t},\shs{u},\shs{v}) \oplus (t,u,v)$ exist and are unique. \qed
\end{proof}
\begin{remark} A hyperbolic reduction of the vacuum Einstein equations in the Newman Penrose \cite{NP} formalism was given by Friedrich \cite{Fr}, using quasilinear symmetric hyperbolic systems.
We introduce a new hyperbolic reduction, using the `vacuum field' gauge of Definition \ref{def:vacvacvac}.
As in \cite{Fr}, there is one quasilinear symmetric hyperbolic system (Proposition \ref{hrpart1}) and one linear homogeneous symmetric hyperbolic system (Proposition \ref{hrpart2}). The latter is used to show that the constraints propagate.
\end{remark}
\begin{proposition}[Hyperbolic reduction, Part 1] \label{hrpart1}
$(\shs{t},\shs{u},\shs{v})=0$ in Proposition \ref{consmeXX} is component by component equivalent to a system of the form $\mathbf{A}(\Phi)\Phi = \mathbf{f}(\Phi)$, where 
\begin{align*}
\mathbf{A}(\Phi) =  \mathbf{A}(\Phi)^{\mu}\tfrac{\p}{\p x^{\mu}} =  \,&\mathrm{diag}\big(L, L,  N,  L,  L\big) \oplus
                                    \mathrm{diag}\big(L,L,L,L,N,N,N,L\big)\\
\,&\oplus \begin{pmatrix}
\lambda_1 N && \lambda_1 D & 0 & 0 && 0 \\
\lambda_1 \overline{D} && \lambda_1 L  +\lambda_2 N & \lambda_2 D & 0 && 0 \\
0 && \lambda_2 \overline{D} & \lambda_2 L + \lambda_3 N & \lambda_3 D && 0 \\
0 && 0 & \lambda_3 \overline{D} & \lambda_3 L + \lambda_4 N && \lambda_4 D\\
0 && 0 & 0 & \lambda_4 \overline{D} && \lambda_4 L
\end{pmatrix}
\end{align*}
and $\mathbf{A}(\Phi)^{\mu}$ is a matrix (resp. $\mathbf{f}(\Phi)$ is a vector) that is a linear (resp. quadratic) polynomial in $\Phi$, $\overline{\Phi}$. The coefficients depend only on $\lambda = (\lambda_1,\lambda_2,\lambda_3,\lambda_4)$. If $e_3 > 0$, then $\mathbf{A}(\Phi)^{\mu}$ are Hermitian and $\mathbf{A}(\Phi)^3+\mathbf{A}(\Phi)^4$ is positive definite.
\end{proposition}
\begin{proof} The proof is by direct hand or machine calculation. For example,
\begin{align*}
-\shs{t}_2 \stackrel{\text{(1)}}{=} {T_{41}}^2 & \stackrel{\text{(2)}}{=}
 F_1({F_4}^2) - F_4 ({F_1}^2) + \b{g}^{cd}\Gamma_{41c}{F_d}^2 - \b{g}^{cd}\Gamma_{14c} {F_d}^2\\
& \stackrel{\text{(3)}}{=} F_1({F_4}^2) - F_4 ({F_1}^2) + \Gamma_{412}{F_1}^2 + \Gamma_{411}{F_2}^2 - \Gamma_{414} {F_3}^2 - \Gamma_{413} {F_4}^2\\
& \qquad 
- \Gamma_{142} {F_1}^2 - \Gamma_{141} {F_2}^2 + \Gamma_{144} {F_3}^2 + \Gamma_{143} {F_4}^2\\
& \stackrel{\text{(4)}}{=} D(0) - L(e_2) + 0 \cdot e_2  + 0\cdot \overline{e}_2 + 0\cdot e_5 + \overline{\gamma}_5 \cdot 0\\
& \qquad 
- \gamma_2 e_2 - \gamma_1 \overline{e}_2 + 0\cdot e_5 - (\gamma_3 - \overline{\gamma}_4)\cdot 0
\end{align*}
(1) Proposition \ref{consmeXX}, (2)  Definition \ref{def:fdhshskhd},
(3) sum, (4) $(F,\Gamma,W) = \mathrm{Param}\circ (e,\gamma,w)$. Hence, $\shs{t}_2 = 0$ is equivalent to $L(e_2) = -\gamma_2 e_2 - \gamma_1 \overline{e}_2$, line two of $\mathbf{A}(\Phi)\Phi = \mathbf{f}(\Phi)$.
\qed
\end{proof}
\begin{proposition} \label{prop:fddfkjhfdkjfdhfdk}
Let $\Phi$ and the associated $(T,U,V)$ be as in Definition \ref{def:fdhshskhd}. Set
\begin{align*}
{{\mathfrak T}_a}^{\mu} \; &=\;  {\epsilon_a}^{ijk}\,\big(\,F_i({T_{jk}}^{\mu}) - {\Gamma_{ij}}^m {T_{mk}}^{\mu} - 
{\Gamma_{ik}}^m {T_{jm}}^{\mu} - {U^m}_{kij}\,{F_m}^{\mu} - {T_{jk}}^{\nu}\tfrac{\p}{\p x^{\nu}}\,{F_i}^{\mu}\,\big)\\
{\mathfrak U}_{cab} \; &= \;   {\epsilon_c}^{ijk}\,\big(\,
F_i(U_{abjk})-{\Gamma_{ia}}^mU_{mbjk}-{\Gamma_{ib}}^m U_{amjk}-{\Gamma_{ij}}^m U_{abmk}-{\Gamma_{ik}}^m U_{abjm}\\
 &\hskip35pt  -{U^m}_{ijk}\,\Gamma_{mab} + \tfrac{1}{3}V_{abijk} - {T_{ij}}^{\mu}\tfrac{\p}{\p x^{\mu}}\,\Gamma_{kab} \,\big)\\
{\mathfrak V}_{jk} \; &=\;  F_b({V^b}_{jk})+{\Gamma_{bm}}^b {V^m}_{jk}-{\Gamma_{bj}}^m{V^b}_{mk}-{\Gamma_{bk}}^m{V^b}_{jm} + 
{U^a}_{mab}\,{W^{mb}}_{jk} \\
 &\hskip15pt -\tfrac{1}{2}U_{mjab}\,{W^{abm}}_k 
+ \tfrac{1}{2}U_{mkab}\,{W^{abm}}_j - \tfrac{1}{2}{T_{ab}}^{\mu}\,\tfrac{\p}{\p x^{\mu}}{W^{ab}}_{jk}
\end{align*}
where $\epsilon_{abcd}$ is totally antisymmetric with $\epsilon_{1234} = -1$. 
Then $({\mathfrak T},{\mathfrak U},{\mathfrak V}) = 0$.
\end{proposition}
\begin{proof}
By direct hand or machine calculation, using $[\tfrac{\p}{\p x^{\mu}},\tfrac{\p}{\p x^{\nu}}] = 0$. \qed
\end{proof}
\begin{proposition}[Hyperbolic reduction, Part 2] \label{hrpart2}
Let $\Phi$ and the associated $(T,U,V)$ be as in Definition \ref{def:fdhshskhd}, and let
$(\shs{t},\shs{u},\shs{v})\oplus (t,u,v)$ be as in Proposition \ref{consmeXX}. Suppose
$$ (\shs{t},\shs{u},\shs{v})(\Phi,\lambda) = 0$$ Then, the vanishing of ${{\mathfrak T}_4}^1$, ${{\mathfrak T}_4}^2$, ${{\mathfrak T}_1}^3$, ${{\mathfrak T}_2}^1$, ${{\mathfrak T}_2}^2$, ${\mathfrak U}_{414}$, $({\mathfrak U}_{412}+{\mathfrak U}_{434})/2$,
${\mathfrak U}_{214}$,
${\mathfrak U}_{114}$,
${\mathfrak U}_{223}$,
${\mathfrak U}_{123}$,
$({\mathfrak U}_{112} + {\mathfrak U}_{134})/2$,
$({\mathfrak U}_{212} + {\mathfrak U}_{234})/2$,
${\mathfrak U}_{332}$,
${\mathfrak V}_{41}$, $({\mathfrak V}_{12} + {\mathfrak V}_{34})/2$, ${\mathfrak V}_{23}$,
asserted by Proposition \ref{prop:fddfkjhfdkjfdhfdk},
is a linear homogeneous system $\widehat{\b{A}}(\Phi)\Phi^{\sharp} = \widehat{\b{f}}(\Phi,\p_x \Phi) \Phi^{\sharp}$ for $\Phi^{\sharp}$, where, here and in the rest of this paper,
\begin{equation} \label{consconscons}
\Phi^{\sharp} = (t,u,v)(\Phi,\lambda)
\end{equation}
is called the \emph{constraint field} associated to $\Phi$, 
and
\begin{align*}
\widehat{\b{A}}(\Phi) = \widehat{\b{A}}(\Phi)^{\mu}\tfrac{\p}{\p x^{\mu}} = &\mathrm{diag}\big(L, L, N, L, L\big) \oplus 
\mathrm{diag}\big(L, L, L, L, N, N, L, L, N\big) \\
 &\oplus \begin{pmatrix}
\frac{1}{\lambda_1} N + \frac{1}{\lambda_2} L & \frac{1}{\lambda_2} D & 0 \\
\noalign{\vskip2.5pt}
\frac{1}{\lambda_2} \overline{D} & \frac{1}{\lambda_2} N + \frac{1}{\lambda_3} L & \frac{1}{\lambda_3} D\\
\noalign{\vskip2.5pt}
0 & \frac{1}{\lambda_3} \overline{ D} & \frac{1}{\lambda_3} N + \frac{1}{\lambda_4} L
\end{pmatrix}
\end{align*}
and $\widehat{\mathbf{f}}(\Phi,\p_x\Phi)$ is pointwise an $\R$-linear transformation.
If $e_3 > 0$, then the matrices $\widehat{\mathbf{A}}(\Phi)^{\mu}$ are Hermitian and $\widehat{\mathbf{A}}(\Phi)^3+\widehat{\mathbf{A}}(\Phi)^4$ is positive definite.
\end{proposition}
\begin{proof}
By direct hand or machine calculation. For example ${{\mathfrak T}_4}^2 = 0$
and $(\shs{t},\shs{u},\shs{v}) = 0$ imply $L(t_2) = -2\gamma_2 t_2 + 2 \Im(t_3 \tfrac{\p}{\p \u{u}} \overline{e}_2) + 
2 \Im(\overline{e}_2 u_1)$,  line two of $\widehat{\b{A}}(\Phi)\Phi^{\sharp} = \widehat{\b{f}}(\Phi,\p_x \Phi) \Phi^{\sharp}$. There are no derivatives of $(t,u,v)$ on the right hand side. \qed
\end{proof}

%% file: matrixG.tex
\begin{pmatrix}  \phantom{-} \, \PhiTwo_3 \, + \, \overline{\PhiTwo}_4 \,  \hskip1mm&\hskip1mm   \, \overline{\PhiTwo}_7 \,  \hskip1mm&\hskip1mm   \, \PhiTwo_6 \,  \hskip1mm&\hskip1mm   \, \PhiTwo_1 \,  \hskip1mm&\hskip1mm   \, \PhiTwo_2 \,  \hskip1mm&\hskip1mm  \phantom{-}\, \PhiTwo_3 \, - \, \overline{\PhiTwo}_4 \,  \\ - \, \PhiTwo_4 \, - \, \overline{\PhiTwo}_3 \,  \hskip1mm&\hskip1mm   \, \PhiTwo_6 \,  \hskip1mm&\hskip1mm   \, \PhiTwo_7 \,  \hskip1mm&\hskip1mm   \, \PhiTwo_2 \,  \hskip1mm&\hskip1mm   \, \overline{\PhiTwo}_1 \,  \hskip1mm&\hskip1mm - \, \PhiTwo_4 \, + \, \overline{\PhiTwo}_3 \,  \\   \phantom{-} \, \PhiTwo_8 \, - \, \overline{\PhiTwo}_8 \,  \hskip1mm&\hskip1mm 0 \hskip1mm&\hskip1mm 0 \hskip1mm&\hskip1mm - \, \PhiTwo_3 \, + \, \overline{\PhiTwo}_4 \,  \hskip1mm&\hskip1mm   \, \PhiTwo_4 \, - \, \overline{\PhiTwo}_3 \,  \hskip1mm&\hskip1mm   \phantom{-}\, \PhiTwo_8 \, + \, \overline{\PhiTwo}_8 \,  \\ \phantom{-}0 \hskip1mm&\hskip1mm   \, \overline{\PhiTwo}_5 \,  \hskip1mm&\hskip1mm   \, \PhiTwo_5 \,  \hskip1mm&\hskip1mm 0 \hskip1mm&\hskip1mm 0 \hskip1mm&\hskip1mm \phantom{-}0\end{pmatrix}

%% file: matrixW.tex
\begin{pmatrix} 
 \, \PhiThree_3 \, + \, \overline{\PhiThree}_3 \,  \hskip1mm&\hskip1mm   \phantom{-}\, \overline{\PhiThree}_4 \,  \hskip1mm&\hskip1mm - \, \PhiThree_4 \,  \hskip1mm&\hskip1mm   \phantom{-}\, \PhiThree_2 \,  \hskip1mm&\hskip1mm - \, \overline{\PhiThree}_2 \,  \hskip1mm&\hskip1mm   \, \PhiThree_3 \, - \, \overline{\PhiThree}_3 \,  \\ 
  \phantom{-}\, \overline{\PhiThree}_4 \,  \hskip1mm&\hskip1mm   \phantom{-}\, \overline{\PhiThree}_5 \,  \hskip1mm&\hskip1mm \phantom{-} 0 \hskip1mm&\hskip1mm \phantom{-} 0 \hskip1mm&\hskip1mm - \, \overline{\PhiThree}_3 \,  \hskip1mm&\hskip1mm - \, \overline{\PhiThree}_4 \,  \\
 - \, \PhiThree_4 \,  \hskip1mm&\hskip1mm \phantom{-} 0 \hskip1mm&\hskip1mm   \phantom{-}\, \PhiThree_5 \,  \hskip1mm&\hskip1mm - \, \PhiThree_3 \,  \hskip1mm&\hskip1mm \phantom{-} 0 \hskip1mm&\hskip1mm - \, \PhiThree_4 \,  \\
  \phantom{-} \, \PhiThree_2 \,  \hskip1mm&\hskip1mm \phantom{-} 0 \hskip1mm&\hskip1mm - \, \PhiThree_3 \,  \hskip1mm&\hskip1mm   \phantom{-}\, \PhiThree_1 \,  \hskip1mm&\hskip1mm \phantom{-} 0 \hskip1mm&\hskip1mm   \phantom{-}\, \PhiThree_2 \,  \\
 - \, \overline{\PhiThree}_2 \,  \hskip1mm&\hskip1mm - \, \overline{\PhiThree}_3 \,  \hskip1mm&\hskip1mm \phantom{-} 0 \hskip1mm&\hskip1mm \phantom{-} 0 \hskip1mm&\hskip1mm   \phantom{-}\, \overline{\PhiThree}_1 \,  \hskip1mm&\hskip1mm   \phantom{-}\, \overline{\PhiThree}_2 \,  \\  
 \, \PhiThree_3 \, - \, \overline{\PhiThree}_3 \,  \hskip1mm&\hskip1mm - \, \overline{\PhiThree}_4 \,  \hskip1mm&\hskip1mm - \, \PhiThree_4 \,  \hskip1mm&\hskip1mm  \phantom{-} \, \PhiThree_2 \,  \hskip1mm&\hskip1mm  \phantom{-} \, \overline{\PhiThree}_2 \,  \hskip1mm&\hskip1mm   \, \PhiThree_3 \, + \, \overline{\PhiThree}_3 \, \end{pmatrix}

%% file: matrixT.tex
\begin{pmatrix} i \, t_1 \,  &  i \, t_2 \,  & 0 & 0 \\   \, \overline{t}_4 \,  &   \, \overline{t}_5 \,  & 0 & 0 \\   \, t_4 \,  &   \, t_5 \,  & 0 & 0 \\ - \, \shs{t}_1 \,  & - \, \shs{t}_2 \,  &   \, t_3 \,  & 0 \\ - \, \overline{\shs{t}}_1 \,  & - \, \overline{\shs{t}}_2 \,  &   \, \overline{t}_3 \,  & 0 \\   \, \shs{t}_4 \,  &   \, \shs{t}_5 \,  & - \, \shs{t}_3 \,  & 0\end{pmatrix}

%% file: matrixUNEW.tex
\begin{pmatrix}   u_2  +  \overline{u}_2   &\hskip 2pt&    u_7  -  \overline{u}_8   &\hskip 2pt&    u_8  -  \overline{u}_7   &\hskip 2pt& -  \shs{u}_3  -  \overline{\shs{u}}_4   &\hskip 2pt&    \shs{u}_4  +  \overline{\shs{u}}_3   &\hskip 2pt&    \shs{u}_8  -  \overline{\shs{u}}_8   \\    \overline{u}_9   &\hskip 2pt& -  \overline{\shs{u}}_7   &\hskip 2pt& -  \shs{u}_6   &\hskip 2pt&    \overline{u}_5   &\hskip 2pt& -  \overline{u}_6   &\hskip 2pt& -  \overline{\shs{u}}_5   \\ -  u_9   &\hskip 2pt& -  \shs{u}_6   &\hskip 2pt& -  \shs{u}_7   &\hskip 2pt& -  u_6   &\hskip 2pt&    u_5   &\hskip 2pt& -  \shs{u}_5   \\ -  u_1   &\hskip 2pt&    u_4   &\hskip 2pt& -  u_3   &\hskip 2pt& -  \shs{u}_1   &\hskip 2pt& -  \shs{u}_2   &\hskip 2pt& -  \shs{u}_3  +  \overline{\shs{u}}_4   \\    \overline{u}_1   &\hskip 2pt& -  \overline{u}_3   &\hskip 2pt&    \overline{u}_4   &\hskip 2pt& -  \shs{u}_2   &\hskip 2pt& -  \overline{\shs{u}}_1   &\hskip 2pt&    \shs{u}_4  -  \overline{\shs{u}}_3   \\    u_2  -  \overline{u}_2   &\hskip 2pt&    u_7  +  \overline{u}_8   &\hskip 2pt&    u_8  +  \overline{u}_7   &\hskip 2pt& -  \shs{u}_3  +  \overline{\shs{u}}_4   &\hskip 2pt&    \shs{u}_4  -  \overline{\shs{u}}_3   &\hskip 2pt&    \shs{u}_8  +  \overline{\shs{u}}_8  \end{pmatrix}

%% file: matrixVNEW.tex
\left( \begin{matrix}
\frac{1}{\lambda_2}(- {\mathfrak v}_2  +   v_1)  -  \frac{1}{\lambda_3}   \overline{v}_3   & \;\; &
 -  \frac{1}{\lambda_4}   \overline{{\mathfrak v}}_5   & \;\; &
\frac{1}{\lambda_3}(    {\mathfrak v}_3  - v_2)    \\
\frac{1}{\lambda_2}(      \overline{{\mathfrak v}}_2  -   \overline{v}_1)  +  \frac{1}{\lambda_3}   v_3   &&
\frac{1}{\lambda_3} (     \overline{{\mathfrak v}}_3  -    \overline{v}_2 )  &&
-  \frac{1}{\lambda_4}   {\mathfrak v}_5    \\
\frac{1}{\lambda_3}(      {\mathfrak v}_3  - \overline{{\mathfrak v}}_3  - v_2  + \overline{v}_2)   &&
\frac{1}{\lambda_4}( -    \overline{{\mathfrak v}}_4  +    \overline{v}_3)   && 
\frac{1}{\lambda_4} (-   {\mathfrak v}_4  +    v_3)   \\
\frac{1}{\lambda_2}( -     v_2  +  \overline{v}_2)   && 
   \frac{1}{\lambda_3}   \overline{v}_3   &&
    \frac{1}{\lambda_3}   v_3  
 \end{matrix} \right.
\\
\noalign{\vskip1mm}
& \hskip 13mm \left. \begin{matrix}
 & \;\; &  -  \frac{1}{\lambda_1}   {\mathfrak v}_1   &\;\;& 
   \frac{1}{\lambda_2}   \overline{v}_2   &\;\;& 
\frac{1}{\lambda_2}(-    {\mathfrak v}_2  +  v_1)  +  \frac{1}{\lambda_3}   \overline{v}_3   \\
 &&    \frac{1}{\lambda_2}   v_2   &&
 -  \frac{1}{\lambda_1}   \overline{{\mathfrak v}}_1   &&
\frac{1}{\lambda_2}( -    \overline{{\mathfrak v}}_2  +  \overline{v}_1)  +  \frac{1}{\lambda_3}   v_3   \\
 && \frac{1}{\lambda_2}(     {\mathfrak v}_2  -  v_1)   && 
\frac{1}{\lambda_2}(     \overline{{\mathfrak v}}_2  -  \overline{v}_1)   && 
\frac{1}{\lambda_3} (    {\mathfrak v}_3  +  \overline{{\mathfrak v}}_3  -   v_2  -   \overline{v}_2)   \\
 && -  \frac{1}{\lambda_1}   v_1   &&
 -  \frac{1}{\lambda_1}   \overline{v}_1   &&
-\frac{1}{\lambda_2}(   v_2  +   \overline{v}_2 )
 \end{matrix}
\right)

%% file: SectionSymmetries.tex
\section{Symmetries} \label{sec:symmetries}
Let $x$, $x'$ be Cartesian coordinates on the open subsets $\mathcal{U},\, \mathcal{U}'\, \subset\R^4\, $. 
\begin{proposition} \label{prop:sedlkdfdfhkdhgkjd}
Let $S$ be any one of the transformations
${\mathfrak C}$, ${\mathfrak Z}$, ${\mathfrak J}$, ${\mathfrak A}$ defined below. Then, $S$ is 
a symmetry transformation acting on triples $(x,\, \Phi,\, \lambda)$, in the sense that :
\begin{itemize}
\item $x \mapsto x'=S \cdot x$ is a diffeomorphism $\mathcal{U}\to \mathcal{U}'$.
\item $\Phi \mapsto \Phi' =S\cdot \Phi $ is a map from fields 
      $\Phi: \mathcal{U}\to \mathcal{R}$ to fields $\Phi': \mathcal{U}'\to \mathcal{R}$ such that
$e'_3 > 0$ and $\Im(e'_1\overline{e}'_2) \neq 0$ on $\mathcal{U}'$ if and only if
$e_3 > 0$ and $\Im(e_1\overline{e}_2) \neq 0$ on $\mathcal{U}$.
\item $\lambda \mapsto \lambda' =S\cdot \lambda$ is a map from weights on $\mathcal{U} $ to weights on $\mathcal{U}'$.
\item $\mathbf{A}(\Phi')\Phi'=\mathbf{f}(\Phi')$ on $\mathcal{U}'$ if and only if $\mathbf{A}(\Phi)\Phi=\mathbf{f}(\Phi)$ on $\mathcal{U}$.
\item $(\Phi')^{\sharp}=0$  on $\mathcal{U}'$ if and only if $\Phi^{\sharp}=0$ on $\mathcal{U}$. 
\end{itemize}
\end{proposition}
\begin{definition}\label{def:allsyms}  Let 
${\mathfrak C}^A(x^1,x^2)\in \R$,\; $\zeta(x^1,x^2) \in U(1)$,\;
${\mathfrak J} > 0$ and $\mathfrak A \neq 0$. Then,
\begin{itemize}
\item ${\mathfrak C}$\, acts by\footnote{${\mathfrak C}$ is not defined unless $x \mapsto x' $ is a diffeomorphism }:\ \, $x'  = ({\mathfrak C}^1 , {\mathfrak C}^2 ,x^3,x^4)$, 
\begin{flushleft}
$(\PhiOne_3',\PhiTwo',\PhiThree',\lambda')(x') = (\PhiOne_3,\PhiTwo,\PhiThree,\, \lambda)(x) $\\
\vskip4pt
$ e_{A+3B}'(x') = \textstyle{\sum_{C=1}^2 \frac{\p {\mathfrak C}^A}{\p x^C}(x) e_{C+3B}}(x)
\hskip15pt A = 1,2\;\; B = 0,1
$
\end{flushleft}
\item ${\mathfrak Z}$ acts by\footnote{
$D(\zeta) = (e_1 \tfrac{\p}{\p x^1} + e_2\tfrac{\p}{\p x^2})(\zeta)$
and
$N(\zeta) = (e_4 \tfrac{\p}{\p x^1} + e_5\tfrac{\p}{\p x^2})(\zeta)$
}:\ \, $x' = x  $,
\begin{flushleft}
$(\Phi',\, \lambda')(x')  = (\zeta^A \Phi 
+
\mathbf{0}_5 \oplus \tfrac{1}{2} \big(\mathbf{0}_2,D(\zeta),\overline{D(\zeta)},
\mathbf{0}_3,\zeta^{-1}N(\zeta)\big)\oplus \mathbf{0}_5,\; \lambda)(x)$ \\
\vskip3pt
$A  = \mathrm{diag}\big(1,1,0,0,0,\; 2,0,1,-1,-1,0,-2,0,\; 2,1,0,-1,-2\big)$
\end{flushleft}
\item ${\mathfrak J}$ acts by:\ \, $x' = \mathrm{diag}(1,1,{\mathfrak J},{\mathfrak J})x$, 
\begin{flushleft}
$ (\Phi',\, \lambda')(x')  = ({\mathfrak J}^A \Phi,\; \lambda)(x)$  \\
\vskip3pt
$ A = (-1)\, \mathrm{diag}\big(1,1,0,1,1,\;\; 1,1,1,1,1,1,1,1,\;\; 2,2,2,2,2\big)$
\end{flushleft}
\item ${\mathfrak A}$ acts by:\ \,
$x' = \mathrm{diag}({\mathfrak A}^{-1},{\mathfrak A}^{-1},1,{\mathfrak A}^2)x$,
\begin{flushleft} 
$(\Phi',\, \lambda')(x') =\big( {\mathfrak A}^A \Phi,\;
\mathrm{diag}(1,{\mathfrak A}^2,{\mathfrak A}^4,{\mathfrak A}^6)\, \lambda\big)(x)$\\
\vskip3pt
$ A = (-1)\, \mathrm{diag}\big(2,2,0,3,3,\;\; 0,0,1,1,1,2,2,2,\;\; 0,1,2,3,4\big)$
\end{flushleft}
\end{itemize}
\end{definition}
\begin{proof}[of Proposition \ref{prop:sedlkdfdfhkdhgkjd}] 
Set $(F,\Gamma,W) = \mathrm{Param}\circ \Phi$,\; $(F',\Gamma',W') = \mathrm{Param}\circ \Phi'$.
Set $F = ({F_a}^{\mu}\tfrac{\p}{\p x^{\mu}})$ and $F' = ({F_a'}^{\mu}\tfrac{\p}{\p (x')^{\mu}})$ with $a=1,2,3,4$. By the above definitions,
\begin{align*}
{\mathfrak C}:\;\; F' & = F &
{\mathfrak Z}:\;\; F' & = \mathrm{diag}(\zeta,\zeta^{-1},1,1)F\\
{\mathfrak J}:\;\; F' & = {\mathfrak J}^{-1} F &
{\mathfrak A}:\;\; F' & = \mathrm{diag}({\mathfrak A}^{-1},{\mathfrak A}^{-1},{\mathfrak A}^{-2},1) F
\end{align*}
In these equalities, we regard $(x,\mathcal{U})$ and $(x',\mathcal{U}')$ as global coordinates on the same manifold. In each case, $\b{g}^{ab}F_a\otimes F_b = C^2 \b{g}^{ab} F_a'\otimes F_b'$ for a constant $C > 0$. Given the transformation law for coordinates and frame, the proof is by direct verification. \qed
\end{proof}
\begin{definition} \label{def:poleflip}  Let $\alpha \neq 0$.
Set $\FLIP_{\alpha}  =  {\mathfrak Z} \circ {\mathfrak C}$ where
$\zeta(\xi) = - \xi / \overline{\xi}$ and ${\mathfrak C}(\xi) = \alpha^2 / \xi$.
Here $\xi = \xi^1 + i\xi^2$ and ${\mathfrak C}(\xi) = {\mathfrak C}^1(\xi) + i {\mathfrak C}^2(\xi)$
and
$x = (\xi^1,\xi^2,\u{u},u)$.
\end{definition}
\begin{remark}
$\FLIP_{\alpha} \circ \FLIP_{\alpha}$ is the identity. The symmetry transformation
 $\FLIP_{\alpha}$ will be used to match constructions between stereographic coordinate patches. 
\end{remark}

%% file: SectionMinkowski.tex
\section{The Minkowski Field} \label{sec:minkowski}
Set
$
\STRIP_{\infty} 
= \R^2 \times (0,\infty)\times (-\infty,0)
$
and use the coordinates $x = (\xi=\xi^1+i\xi^2,\u{u},u)$. 
\begin{definition} \label{def:minkowski} For all $a,\mathfrak A \neq 0$, let $\Minkowski_{a,\mathfrak A}: \STRIP_{\infty} \to \range$
be the field
\begin{gather}\label{eq:minkowskidefinition}
\Minkowski_{a,{\mathfrak A}} = 
\big(\rho^{-1}\;
\b{e},\, 
i\rho^{-1}\,
\b{e}
,\,  1,\,  {\bf 0}_2\big)
\oplus 
\  \rho^{-1} 
\big(0 ,\, {\mathfrak A}^2,\, \boldsymbol{\lambda},\; \overline{\boldsymbol{\lambda}},\,  0 ,\, - 1,\,  
{\bf 0}_2\big)
\oplus\,  {\bf 0}_5
\\
\label{eq:defelambda}
\rho_{a,\mathfrak A}(u,\u{u}) = {\mathfrak A}^2\u{u} - u\qquad \b{e}_{a,\mathfrak A}(\xi) = \tfrac{a}{2}\big( 1 + \tfrac{\mathfrak A^2}{a^2}|\xi|^2\big)
\qquad \boldsymbol{\lambda}_{a,\mathfrak A}(\xi) = -\tfrac{\mathfrak A^2}{2a}\,\xi
\end{gather}
\end{definition}
We often omit the subscripts on $\rho$, $\b{e}$, $\boldsymbol{\lambda}$. Let $S = S(u,\u{u})$ be given by
\begin{equation}\label{eq:SdefSdefSdef}
S = \rho^{-1}u^2 + u
\qquad \text{or, equivalently,}
\qquad
\rho^{-1} = -u^{-1} + u^{-2}S
\end{equation}
\begin{proposition} \label{minkprop}Recall, ${\mathfrak C}$, ${\mathfrak A}$, $\mathfrak J$, $\FLIP_{\alpha}$ from   
Section \ref{sec:symmetries}. For all $a,\mathfrak A \neq 0$:
\begin{itemize}
\item[(a)] $\Minkowski_{a,\mathfrak A} = ({\mathfrak C}\circ {\mathfrak A}) \cdot {\Minkowski}_{1,1}$ on $\STRIP_{\infty}$, where ${\mathfrak C}(\xi) = a\,\xi$.
\item[(b)] $\Minkowski_{a,\mathfrak A} = \FLIP_{\frac{a}{\mathfrak A}}\cdot \Minkowski_{a,\mathfrak A}$ on $\STRIP_{\infty}\cap \{\xi \neq 0\}$.
\item[(c)] $\Minkowski_{a,\mathfrak A} = {\mathfrak J}\cdot \Minkowski_{a,\mathfrak A}$ on $\STRIP_{\infty}$, for all $\mathfrak J > 0$.
\item[(d)] $\Minkowski_{a,\mathfrak A}$ is a vacuum field on $\STRIP_{\infty}$ (see Definition \ref{def:vacvacvac}).
\item[(e)] $\mathcal{M}_{a,\mathfrak A}$ is isometric, as a Lorentz manifold, to 
the open subset of Minkowski space
 \begin{equation*}
\big\{(X^0,\b{X}) \in \R\times \R^3\;\; \big|\;\; |X^0| < |\b{X}|,\; \b{X}\notin \{0\}\times \{0\}\times [0,\infty)\big\}
\end{equation*} where $(X^0,\b{X})$ are the standard Minkowski coordinates, and
$$
X^0 = \frac{1}{\sqrt{2}\, |{\mathfrak A}|}({\mathfrak A}^2 \u{u} + u) \qquad \begin{pmatrix} X^1 \\ X^2 \\ X^3 \end{pmatrix} = \frac{1}{\sqrt{2}\, |{\mathfrak A}|}\;\frac{{\mathfrak A}^2\u{u} - u}{1 + \frac{\mathfrak A^2}{a^2}|\xi|^2}\begin{pmatrix} \frac{2\mathfrak A}{a}\xi^1 \\  \rule{0pt}{11pt}\frac{2 \mathfrak A}{a}\xi^2\\ \rule{0pt}{11pt} -1 + \frac{\mathfrak A^2}{a^2}|\xi|^2
\end{pmatrix}
$$
\end{itemize}
\end{proposition}
\begin{proof} (a) Set $C_a = \mathrm{diag}(a,a,1,a,a)\oplus \mathbbm{1}_{13}$. Let $A$ be as in Definition \ref{def:allsyms}, $\anisotropic$. Then,
\begin{multline*}
\big((\mathfrak C\circ \mathfrak A)\cdot \mathcal{M}_{1,1}\big)(\xi,\u{u},u)
 = \big(\mathfrak C \cdot (\mathfrak A\cdot \mathcal{M}_{1,1})\big)(\xi,\u{u},u)\\
 = C_a (\mathfrak A\cdot \mathcal{M}_{1,1})(a^{-1}\xi,\u{u},u)
 = C_a {\mathfrak A}^A \mathcal{M}_{1,1}\big(\tfrac{\mathfrak A}{a} \xi,\u{u},{\mathfrak A}^{-2} u\big) 
 = \Minkowski_{a,\mathfrak A}(\xi,\u{u},u)
\end{multline*}
Parts (b),(c),(d),(e) are by direct calculation. We check (d) for $\Minkowski_{1,1}$, then use (a). \qed
\end{proof}
\begin{remark} \label{rem:interpr}  The level sets of $u = 2^{-\frac{1}{2}} |{\mathfrak A}| (X^0 - |\b{X}|)$ and $\u{u} = 2^{-\frac{1}{2}} |\tfrac{1}{\mathfrak A}|(X^0 + |\b{X}|)$ are null hypersurfaces. They intersect in a standard sphere of radius $|\b{X}| =  2^{-\frac{1}{2}} |\tfrac{1}{\mathfrak A}|\rho$, with the north pole removed, on which $\frac{\mathfrak A}{a}\xi$ is the standard stereographic coordinate system. The southern hemisphere corresponds to $|\xi| < |\frac{a}{\mathfrak A}|$.
\end{remark}

%% file: SectionAnsatz.tex
\section{The Relevant\,/\,Irrelevant Form} \label{sec:ansatz}
Recall $(x^{\mu}) = (\xi^1,\xi^2,\u{u},u)$. By definition the change of fields, for $u<0$, 
\begin{equation}\label{eq:farfieldansatz}
\Phi = \Minkowski_{a,\mathfrak A} + u^{-M}\Psi
\end{equation}
is called the \emph{far field ansatz}, see \eqref{eq:minkowskidefinition}.  Here,
\begin{align*}
M & = \mathrm{diag}(2,2,2,3,3) \oplus \mathrm{diag}(1,2,2,2,2,2,2,3) \oplus \mathrm{diag} (1,2,3,4,4)\\
\Psi & = (\Psi_1,\Psi_2,\Psi_3)  =  (\PsiOne, \PsiTwo, \PsiThree)
\; \text{is an $\range$-valued field}
\end{align*}
We will construct vacuum fields for which $\Psi$ has a nontrivial limit as $u\to -\infty$. 
\begin{proposition} In this proposition, ignore \eqref{consconscons}, regard $\Phi$ and $\Phi^{\sharp}$ as independent, sufficiently differentiable fields with values in $\range$ and $\ranges$, respectively.
\label{prop:equivform} Set
\begin{align*}
 M & = \mathrm{diag}(2,2,2,3,3) \oplus\mathrm{diag}(1,2,2,2,2,2,2,3) \oplus \mathrm{diag}(1,2,3,4,4) \\
E & = \mathrm{diag}(4,4,4,6,6) \oplus \mathrm{diag}(2,4,4,4,4,4,4,6) \oplus \mathrm{diag}(0,0,0,0,0)\\
M^\sharp & = \mathrm{diag}(2,2,2,3,3) \oplus \mathrm{diag}(2,2,2,2,2,2,3,3,3) \oplus \mathrm{diag}(0,-1,-2)\\
E^{\sharp} & = \mathrm{diag}(4,4,4,6,6) \oplus \mathrm{diag}(4,4,4,4,4,4,6,6,6) \oplus \mathrm{diag}(2,2,2)
\end{align*}
and
\begin{subequations}
\begin{align}
\label{eq:setset3} \Phi(x) & = \Minkowski_{a,\mathfrak A}(x) + u^{-M}\,\Psi(x) & \Psi(x) & \in \range \\
\label{eq:setset4} \Phi^\sharp(x) & = u^{-M^\sharp}\,\Psi^\sharp(x) & \Psi^{\sharp}(x) &\in \ranges\\
\label{eq:setset5} \lambda_j(x) & = u^{2j} & j =1,2,3,4 \quad  & \text{(see, Proposition \ref{consmeXX})}
\end{align}
\end{subequations}
The systems (see, Section \ref{sec:geomety}) $\b{A}(\Phi)\Phi = \b{f}(\Phi)$ and $\widehat{\b{A}}(\Phi)\,\Phi^{\sharp} = \widehat{\b{f}}(\Phi,\p_x \Phi)\, \Phi^{\sharp}$ for $\Phi$ and $\Phi^{\sharp}$ are equivalent to the following systems for $\Psi$ and $\Psi^{\sharp}$:
\begin{subequations}
\begin{align}
\label{shspsi1}
\b{A}_{a,{\mathfrak A}}(x,\Psi)\,\Psi^{\phantom{\sharp}} & \;=\; \b{f}_{a,{\mathfrak A}}(x,\Psi) \\
\label{shspsi2}
\widehat{\b{A}}_{a,{\mathfrak A}}(x,\Psi)\,\Psi^{\sharp} & \;=\; \widehat{\b{f}}_{a,{\mathfrak A}}(x,\Psi,\p_x \Psi)\,\Psi^{\sharp}
\end{align}
\end{subequations}
where $\b{A}_{a,{\mathfrak A}}(x,\Psi) = \b{A}_{a,{\mathfrak A}}^{\mu}(x,\Psi) \tfrac{\p}{\p x^{\mu}}$ and
$\widehat{\b{A}}_{a,{\mathfrak A}}(x,\Psi)  = \widehat{\b{A}}_{a,{\mathfrak A}}^{\mu}(x,\Psi) \tfrac{\p}{\p x^{\mu}}$ and
\begin{align*}
\b{A}^{\mu}_{a,{\mathfrak A}}(x,\Psi) & = u^{E}\big(u^{-M}\b{A}^{\mu}(\Phi)u^{-M}\big) \\
\b{f}_{a,{\mathfrak A}}(x,\Psi) & =  u^{E-M}\big(-\b{A}^{\mu}(\Phi) \big(\tfrac{\p}{\p x^{\mu}} u^{-M}\big)\,\Psi + \b{f}(\Phi) - \b{A}^{\mu}(\Phi) \tfrac{\p}{\p x^{\mu}} \Minkowski_{a,\mathfrak A}\big)\\
\widehat{\b{A}}^{\mu}_{a,{\mathfrak A}}(x,\Psi) & = u^{E^\sharp} \big(u^{-M^{\sharp}}\widehat{\b{A}}^{\mu}(\Phi)\,u^{-M^\sharp}\big)\\
\widehat{\b{f}}_{a,{\mathfrak A}}(x,\Psi,\p_x \Psi) & = u^{E^{\sharp}-M^{\sharp}}\big(- \widehat{\b{A}}^{\mu}(\Phi)\,\big(\tfrac{\p}{\p x^{\mu}}u^{-M^\sharp}\big) +
\widehat{\b{f}}(\Phi,\p_x \Phi)\, u^{-M^\sharp}\big)
\end{align*}
Here, $\Phi$, $\p_x \Phi$ have to be expressed in terms of $\Psi$, $\p_x \Psi$ using \eqref{eq:setset3}. We will sometimes drop the $a$, $\mathfrak A$ and write $\b{A}(x,\Psi)$, $\b{f}(x,\Psi)$, $\widehat{\b{A}}(x,\Psi)$, $\widehat{\b{f}}(x,\Psi,\p_x \Psi)$. They are notationally distinguished from $\b{A}(\Phi)$, $\b{f}(\Phi)$, $\widehat{\b{A}}(\Phi)$, $\widehat{\b{f}}(\Phi,\p_x \Phi)$ by the number of arguments.
\end{proposition}
\begin{remark} \label{rem:equivformprop} $\b{A}^{\mu}(x,\Psi)$,\; $\widehat{\b{A}}^{\mu}(x,\Psi)$ are Hermitian, so that \eqref{shspsi1}, \eqref{shspsi2} are  symmetric hyperbolic.
They are affine $\R$-linear functions of $\Psi$.
The $\R$-linear transformation $\widehat{\b{f}}(x,\Psi,\p_x \Psi)$
depends affine $\R$-linearly on $\Psi \oplus \p_x \Psi$. On the other hand, $\b{f}(x,\Psi)$ is a quadratic polynomial in the components of $\Psi$, $\overline{\Psi}$ without constant term. There is no constant term, because $\Minkowski_{a,\mathfrak A}$ is a vacuum field. Neither derivatives of $\b{e}_{a,\mathfrak A}$ nor derivatives of $\boldsymbol{\lambda}_{a,\mathfrak A}$ appear in the term $\b{A}^{\mu}(\Phi) \tfrac{\p}{\p x^{\mu}} \Minkowski_{a,\mathfrak A}$.
\end{remark}
\begin{definition}[Generic Symbols] \label{def:gspandps} Let $S$ be as in \eqref{eq:SdefSdefSdef}.
\begin{list}{\labelitemi}{\leftmargin=1em}
\item $\mathcal{P}$ is a generic symbol for a quadratic polynomial in the components of the fields $\Psi$ and $\overline{\Psi}$ without constant term, whose coefficients are (complex) polynomials in $\frac{1}{u}$, $\mathfrak A$, $S$, $\b{e}$, $\boldsymbol{\lambda}$, $\overline{\boldsymbol{\lambda}}$.
\item $\mathcal{P}^{\sharp}$ is a generic symbol for a polynomial in the components of the fields $\Psi$ and $\overline{\Psi}$ and all their first order coordinate derivatives, whose coefficients are (complex) polynomials in $\frac{1}{u}$, $\mathfrak A$, $S$, $\b{e}$, $\boldsymbol{\lambda}$, $\overline{\boldsymbol{\lambda}}$, and all their first order coordinate derivatives.
\end{list}
We use the same symbols $\mathcal{P}$ and $\mathcal{P}^{\sharp}$ for a vector or matrix all of whose entries are polynomials of this kind.
Here and below, $\b{e} = \b{e}_{a,\mathfrak A}$ and $\boldsymbol{\lambda} = \boldsymbol{\lambda}_{a,\mathfrak A}$.
\end{definition}
\begin{remark} \label{rem:DNL} The vector fields $D$, $N$ and $L$ corresponding to $\Phi = \Minkowski_{a,\mathfrak A} + u^{-M}\Psi$ are:
\begin{equation}\label{eq:DNL}
\begin{aligned}
D & =-\tfrac{2}{u}\,\b{e}\,\tfrac{\p}{\p \overline{\xi}} + \tfrac{2}{u^2}\,\b{e} \,S\, \tfrac{\p}{\p \overline{\xi}} + \tfrac{1}{u^2}\big(\PsiOne_1 \tfrac{\p}{\p \xi^1} + \PsiOne_2\tfrac{\p}{\p \xi^2}\big)
& \quad\tfrac{\p}{\p \overline{\xi}} & = \tfrac{1}{2}(\tfrac{\p}{\p \xi^1} + i \tfrac{\p}{\p \xi^2})
\\
N & =\tfrac{\p}{\p u} + \tfrac{1}{u^3}\big(\PsiOne_4 \tfrac{\p}{\p \xi^1} + \PsiOne_5 \tfrac{\p}{\p \xi^2}\big) &
L & = \tfrac{\p}{\p \u{u}} + \tfrac{1}{u^2}\PsiOne_3\tfrac{\p}{\p \u{u}}
\end{aligned}
\end{equation}
\end{remark}
\begin{proposition}[Relevant\,/\,Irrelevant Form] \label{prop:ri1}
The system \eqref{shspsi1} takes the form
\begin{subequations} \label{eq:ri1}
\begin{align}
\label{ri1:1} L \begin{pmatrix} \PsiOne_1\\ \PsiOne_2\\ \PsiOne_4\\ \PsiOne_5 \\ \PsiTwo_1\\ \PsiTwo_2\\ \PsiTwo_3\\ \PsiTwo_4 \\ \PsiTwo_8
\end{pmatrix} \; & = \; \begin{pmatrix}
\phantom{-i}\,\b{e}\,\PsiTwo_1 \\
-i\,\b{e}\, \PsiTwo_1 \\
\phantom{-}2\,\b{e}\, \Re (\PsiTwo_4 - \overline{\PsiTwo}_3 - \PsiTwo_5)\\
-2\,\b{e}\,\Im (\PsiTwo_4 - \overline{\PsiTwo}_3 - \PsiTwo_5)\\
-\PsiThree_1\\
 - |\PsiTwo_1|^2\\
-\PsiThree_2  - \overline{\boldsymbol{\lambda}}\,\PsiTwo_1\\
 - \boldsymbol{\lambda}\,\overline{\PsiTwo}_1\\
\PsiThree_3 + 2i\,\Im \,\big(\boldsymbol{\lambda}\,(\PsiTwo_4-\PsiTwo_5-\overline{\PsiTwo}_3)\big)
\end{pmatrix} + \frac{1}{u}\mathcal{P} \displaybreak[0] \\
\noalign{\vskip 1mm}
\label{ri1:2} N \begin{pmatrix} \PsiOne_3 \\ \PsiTwo_5 \\ \PsiTwo_6 \\ \PsiTwo_7
\end{pmatrix} \; & = \; \frac{1}{u}
\begin{pmatrix}2\,(\PsiOne_3 + \Re\,\PsiTwo_8) \\
\PsiTwo_4 + \PsiTwo_5 - \overline{\PsiTwo}_3 \\
0 \\ 0
\end{pmatrix} + \frac{1}{u^2}\mathcal{P}
\end{align}
\begin{align}
\notag \begin{pmatrix} N && \tfrac{1}{u}D & 0 & 0 && 0\\
\tfrac{1}{u}\overline{D} && N + \tfrac{1}{u^2}L & \tfrac{1}{u}D & 0 && 0 \\
0 && \tfrac{1}{u}\overline{D} & N + \tfrac{1}{u^2}L & \tfrac{1}{u}D && 0 \\
0 && 0 & \tfrac{1}{u}\overline{D} & N + \tfrac{1}{u^2}L && D \\
0 && 0 & 0 & \overline{D} && L\end{pmatrix} \begin{pmatrix} \PsiThree_1\\\PsiThree_2\\\PsiThree_3\\\PsiThree_4\\\PsiThree_5
\end{pmatrix}
& = \frac{1}{u} \begin{pmatrix} 0 \\ 0 \\ 0 \\ 4\boldsymbol{\lambda}\PsiThree_5 \\ \anisotropic^2 \PsiThree_5 - 2\overline{\boldsymbol{\lambda}}\PsiThree_4 - 3\PsiTwo_7 \PsiThree_3 \end{pmatrix} \\
\label{ri1:3} & \hskip 5mm + \frac{1}{u^2}\mathcal{P}
\end{align}
\end{subequations}
\end{proposition}
\begin{proof} By direct machine or hand calculation. See Appendix \ref{app:appprog}. \qed
\end{proof}
\begin{remark} \label{rem:dkskjhkjhdskjds} We comment on the relevant\,/\,irrelevant equations \eqref{eq:ri1}:
\begin{itemize}
\item Terms containing
$\mathcal{P}$
are referred to as irrelevant, the others as relevant. Relevant terms are principal in the number of derivatives or leading order in powers of $\tfrac{1}{u}$. 
\item
The only relevant, nonlinear term on the right hand side of
\eqref{ri1:1} is $-|\PsiTwo_1|^2$.  This term generates trapped spheres. 
\item The linear terms in the relevant part of the right hand sides of \eqref{eq:ri1}
 are of two kinds. Either there is an explicit factor of $\b{e}$, $\boldsymbol{\lambda}$, $\overline{\boldsymbol{\lambda}}$ or there is a numerical factor (besides powers of $\tfrac{1}{u}$). In the first case, we can make the factor small by requiring $|\mathfrak A|  \leq |a|$ and making $|a|$ small. In the second case, we arrange the terms into a linear over $\R$ matrix applied to $\Psi$. We exploit the structure of this matrix, it motivates the assumptions {\bf \small (RE11)}, {\bf \small (RE13a)}, {\bf \small (RE13b)} 
 of Proposition \ref{prop:fdsdsfkjhdskhds}.
\end{itemize}
\end{remark}
\begin{proposition} \label{prop:ri2}
Recall \eqref{consconscons} and \eqref{eq:setset4}.
Let $\Psi^\sharp = (\PsiSharpOne,\PsiSharpTwo,\PsiSharpThree)$.
 Then,
\begin{subequations}\label{eq:ri2}
\begin{align}
\label{ri2:1} & \begin{pmatrix} \PsiSharpOne_1 \\ \PsiSharpOne_2 \\ \PsiSharpOne_3 \end{pmatrix} \; = \; \mathcal{P}^{\sharp} \hskip 21mm \begin{pmatrix} \PsiSharpOne_4 \\ \PsiSharpOne_5
\end{pmatrix} \; = \; -u\, N\, \begin{pmatrix} \overline{\PsiOne}_1 \\ \overline{\PsiOne}_2 \end{pmatrix} + \mathcal{P}^{\sharp} \displaybreak[0]\\
\noalign{\vskip 1mm}
\label{ri2:2} & \begin{pmatrix} \PsiSharpTwo_1 \\ \PsiSharpTwo_2 \\ \PsiSharpTwo_3 \end{pmatrix} \; = \; \mathcal{P}^{\sharp} \hskip 20mm \begin{pmatrix} \PsiSharpTwo_4 \\ \PsiSharpTwo_7 \\ \PsiSharpTwo_8
\end{pmatrix} \; = \; -u\, N\, \begin{pmatrix} 
\phantom{-} \PsiTwo_1 \\ \phantom{-} \PsiTwo_3 \\ - \PsiTwo_4 \end{pmatrix} + \mathcal{P}^{\sharp} \displaybreak[0]\\
\noalign{\vskip 1mm}
\label{ri2:3} & \begin{pmatrix} \PsiSharpTwo_5 \\ \PsiSharpTwo_6 \\ \PsiSharpTwo_9 \end{pmatrix} \; = \; \mathcal{P}^{\sharp}\; = \; 
\begin{pmatrix}
-L(\PsiTwo_7) - \overline{\PsiTwo}_1\\
L(\PsiTwo_6)\\
u\,D(\PsiTwo_7) - u\, \overline{D}(\PsiTwo_6) + \PsiTwo_4 - \overline{\PsiTwo}_3 - 4\boldsymbol{\lambda} \,\PsiTwo_7
\end{pmatrix} + \frac{1}{u}\mathcal{P}^{\sharp} \displaybreak[0]\\
\noalign{\vskip 1mm}
\label{ri2:4} & \begin{pmatrix} \PsiSharpThree_1 \\ \PsiSharpThree_2 \\ \PsiSharpThree_3
\end{pmatrix}\; =\; \mathcal{P}^{\sharp} \; = \; \begin{pmatrix} u\,\overline{D} - 4\overline{\boldsymbol{\lambda}} & L && 0 && 0 \\
\PsiTwo_7 & u\,\overline{D} - 2 \overline{\boldsymbol{\lambda}} && L && 0 \\
0 & 2\,\PsiTwo_7 && u\,\overline{D} && L \end{pmatrix}\begin{pmatrix}\PsiThree_1 \\ \PsiThree_2 \\ \PsiThree_3 \\ \PsiThree_4 \end{pmatrix} + \frac{1}{u}\mathcal{P}^{\sharp}
\end{align}
\end{subequations}
\end{proposition}
\begin{proof} By direct machine or hand calculation. \qed
\end{proof}
\begin{remark} Every $\mathcal{P}^\sharp$ in \eqref{eq:ri2} has no constant term as a polynomial in the components of $\Psi$, $\overline{\Psi}$ and their first coordinate derivatives, because $\Minkowski_{a,\mathfrak A}$ is a vacuum field.
\end{remark}
\begin{proposition}[Relevant\,/\,Irrelevant Form] \label{prop:ri3}
Let $\Psi^{\sharp} = (\PsiSharpOne,\PsiSharpTwo,\PsiSharpThree)$. The system \eqref{shspsi2} takes the form
\begin{subequations}\label{eq:ri3}
\begin{align}
\label{ri3:1} L \begin{pmatrix} \PsiSharpOne_1 \\ \PsiSharpOne_2 \\ \PsiSharpOne_4 \\ \PsiSharpOne_5 \\ \PsiSharpTwo_1 \\ \PsiSharpTwo_2 \\ \PsiSharpTwo_3 \\ \PsiSharpTwo_4 \\ \PsiSharpTwo_7 \\ \PsiSharpTwo_8 \end{pmatrix}
\; & = \; \begin{pmatrix} 0 \\ 0 \\ 
\phantom{i}\,\b{e}\,(\overline{\PsiSharpTwo}_4-\PsiSharpTwo_5) + \phantom{i}\b{e}\,(\overline{\PsiSharpTwo}_6 - \PsiSharpTwo_3)\\
i\,\b{e}\,(\overline{\PsiSharpTwo}_4-\PsiSharpTwo_5) -i\,\b{e}\,(\overline{\PsiSharpTwo}_6 - \PsiSharpTwo_3)\\
\PsiSharpThree_1\\ 0\\ 0\\ 0\\
\phantom{-}\boldsymbol{\lambda}(\PsiSharpTwo_6- \overline{\PsiSharpTwo}_3) - \overline{\boldsymbol{\lambda}} (\PsiSharpTwo_4 - \overline{\PsiSharpTwo}_5)\\
-\overline{\boldsymbol{\lambda}} (\overline{\PsiSharpTwo}_6-\PsiSharpTwo_3) + \boldsymbol{\lambda}(\overline{\PsiSharpTwo}_4 - \PsiSharpTwo_5)
\end{pmatrix} + \frac{1}{u}\mathcal{P}^{\sharp} \Psi^{\sharp} \displaybreak[0]\\
\noalign{\vskip1mm}
\label{ri3:2} N \begin{pmatrix}
\PsiSharpOne_3 \\ \PsiSharpTwo_5 \\ \PsiSharpTwo_6 \\ \PsiSharpTwo_9 
\end{pmatrix}
\; & = \;  \frac{1}{u}\begin{pmatrix}
\PsiSharpOne_3 + \PsiSharpTwo_7 + \overline{\PsiSharpTwo}_8 \\
\overline{\PsiSharpTwo}_4 \\
\overline{\PsiSharpTwo}_3 \\
\overline{\PsiSharpTwo}_7 + \PsiSharpTwo_8
\end{pmatrix} + \frac{1}{u^2}\mathcal{P}^{\sharp}\Psi^{\sharp}
\end{align}
\begin{equation} \label{ri3:3} 
\begin{pmatrix}
N + \tfrac{1}{u^2}L & \tfrac{1}{u}D & 0\\
\tfrac{1}{u}\overline{D} & N + \tfrac{1}{u^2}L & \tfrac{1}{u}D\\
0 &  \tfrac{1}{u}\overline{D} & N + \tfrac{1}{u^2}L
\end{pmatrix}\begin{pmatrix} \PsiSharpThree_1 \\ \PsiSharpThree_2 \\ \PsiSharpThree_3\end{pmatrix}  = \frac{1}{u^2}\mathcal{P}^{\sharp}\Psi^{\sharp}
\end{equation}
\end{subequations}
Above, the symbols $\mathcal{P}^{\sharp}$ are linear over $\R$ generic transformations, see Definition \ref{def:gspandps}.
\end{proposition}
\begin{proof} By direct machine or hand calculation. \qed
\end{proof}
\begin{remark} The overall factors $u^E$ and $u^{E^{\sharp}}$ appear in \eqref{shspsi1} and \eqref{shspsi2}, so that these systems are line by line (up to a permutation of the lines) equivalent to
\eqref{eq:ri1} and \eqref{eq:ri3}.
\end{remark}

%% file: SectionFormalSolution.tex
\section{Formal Solutions} \label{sec:formalsolution}
We consider formal power series on $\STRIP_{\infty} \subset \R^4$ (see, Section \ref{sec:minkowski}),
\begin{equation}\label{fps}
[\,\Psi\,](x) \; = \; \textstyle\sum_{k=0}^{\infty}\;(\tfrac{1}{u})^k\; \Psi(k)(\xi,\u{u})\quad \text{where} \quad
\Psi(k) \in C^{\infty}(\R^2\times (0,\infty), \range)
\end{equation}
\begin{remark} \label{rem:grgrhz}
By Proposition \ref{prop:ri2}, the associated formal constraint field $[\,\Psi^{\sharp}\,]$ is 
itself a formal power series $[\,\Psi^{\sharp}\,](x) = \sum_{k=0}^{\infty}(\tfrac{1}{u})^k\, \Psi^{\sharp}(k)(\xi,\u{u})$, with $\Psi^{\sharp}(k)$ depending only on $\Psi(\ell)$, $0\leq \ell \leq k$.
\end{remark}
The characteristic initial problem in Proposition \ref{prop:fps} is motivated by \cite{Chr}.
\begin{proposition} \label{prop:fps}
For all $a,\mathfrak A \neq 0$, $\u{u}_0 > 0$, all smooth $\text{\bf \small DATA}(\xi,\u{u}): \R^2 \times (0,\infty) \to \C$ that vanish when $\u{u} < \u{u}_0$, there is a unique formal power series $[\,\Psi\,]$ on $\STRIP_{\infty}$, which satisfies \eqref{shspsi1} and $[\,\Psi^\sharp\,] = 0$ and (the formal characteristic initial data)
\begin{equation} \label{eq:fordhhkhkshkfhfdforfor}
[\,\Psi\,] = 0\;\; \text{when $\u{u} < \u{u}_0$}
\qquad \qquad \PsiTwo_1(0) = \text{\bf \small DATA}
\end{equation}
Moreover, for all $k \geq 0$, the value of $\Psi(k)$ at $(\xi,\u{u}) \in \R^2 \times (0,\infty)$ depends only on the restriction of $\text{\bf \small DATA}(\xi,\u{u})$ and its derivatives of all orders to the half-open line segment $\{\xi\}\times (0,\u{u}\,]$ (formal finite speed of propagation). Explicitly, $\Psi(0)$ is given by:
\begin{equation} \label{eq:thezeroorderequations}
\begin{aligned}
\PsiTwo_1(0) & = \text{\bf \small DATA} &
\PsiThree_5(0) & = 0 \\
\PsiTwo_7(0) & = -\p_{\u{u}}^{-1} \overline{\PsiTwo_1(0)} &
\PsiTwo_2(0) & = - \p_{\u{u}}^{-1}  |\PsiTwo_1(0)|^2\\
\PsiThree_1(0) & = -\tfrac{\p}{\p \u{u}} \PsiTwo_1(0) &
\PsiTwo_4(0) & = - \boldsymbol{\lambda}\, \p_{\u{u}}^{-1} \overline{\PsiTwo_1(0)}\\
\PsiThree_2(0) & = 2\big(\b{e} \tfrac{\p}{\p \xi} + 2\overline{\boldsymbol{\lambda}}\big)\, \p_{\u{u}}^{-1} \PsiThree_1(0) &
\PsiTwo_6(0) & = 0 \\
\PsiThree_3(0) & = 2\big(\b{e} \tfrac{\p}{\p \xi} + \overline{\boldsymbol{\lambda}}\big)\, \p_{\u{u}}^{-1} \PsiThree_2(0) - \p_{\u{u}}^{-1}\big( \PsiTwo_7(0)\PsiThree_1(0)\big) &
\PsiOne_1(0) & = \phantom{-i}\, \b{e}\, \p_{\u{u}}^{-1} \PsiTwo_1(0) \\
\PsiThree_4(0) & = 2\,\b{e}\tfrac{\p}{\p \xi}\, \p_{\u{u}}^{-1}\PsiThree_3(0) - 2\p_{\u{u}}^{-1}\big(\PsiTwo_7(0)\PsiThree_2(0)\big) &
\PsiOne_2(0) & = -i\,\b{e}\, \p_{\u{u}}^{-1} \PsiTwo_1(0) \\
\PsiTwo_3(0) & = -\p_{\u{u}}^{-1} \PsiThree_2(0) - \overline{\boldsymbol{\lambda}}\, \p_{\u{u}}^{-1} \PsiTwo_1(0) &
\PsiOne_3(0) & = - \Re \,\PsiTwo_8(0) \\
\PsiTwo_5(0) & = - \p_{\u{u}}^{-1} \overline{\PsiThree_2(0)} &
\PsiOne_4(0) & = -4\,\b{e}\,\p_{\u{u}}^{-1}\,\Re \,\PsiTwo_5(0)\\
\PsiTwo_8(0) & = \p_{\u{u}}^{-1}\PsiThree_3(0) -4i\, \p_{\u{u}}^{-1}\Im \big(\boldsymbol{\lambda} \,\PsiTwo_5(0)\big) &
\PsiOne_5(0) & = \phantom{-}4\,\b{e}\,\p_{\u{u}}^{-1}\,\Im\,\PsiTwo_5(0)
\end{aligned}
\end{equation}
\rule{0pt}{13pt} with $\b{e} = \b{e}_{a,\mathfrak A}$ and $\boldsymbol{\lambda} = \boldsymbol{\lambda}_{a,\mathfrak A}$ and $\frac{\p}{\p \xi} = \tfrac{1}{2}\big(\frac{\p}{\p \xi^1} -i \frac{\p}{\p \xi^2}\big)$ and
$\big(\p_{\u{u}}^{-1} g\big)(\u{u}) = \int_0^{\u{u}} \dd \u{u}' g\big(\u{u}'\big)$.
\end{proposition}
We now prepare for the proof of Proposition \ref{prop:fps}. 
\begin{definition}\label{minkfps}
Recall \eqref{eq:minkowskidefinition}.
Let
$[\,\Minkowski_{a,\mathfrak A}\,](x) = \textstyle\sum_{k=0}^{\infty}(\tfrac{1}{u})^k \Minkowski_{a,\mathfrak A}(k)(\xi,\u{u})$
be the formal expansion in $\frac{1}{u}$ for the Minkowski vacuum field $\mathcal{M}_{a,\mathfrak A}$ with
\begin{equation} \label{eq:sexpan}
[\,\tfrac{1}{\rho}\,] \, = \, -\tfrac{1}{u} + \tfrac{1}{u^2}\,[\,S\,] \qquad \qquad [\,S\,] \; = \; - \textstyle\sum_{k=0}^{\infty}(\tfrac{1}{u})^k\, {\mathfrak A}^{2(k+1)}\, \u{u}^{k+1}
\end{equation}
\end{definition}
\begin{definition} \label{def:gsn} Regard the components of $\Psi(k)$ and $\overline{\Psi(k)}$, $k\geq 0$, and their formal first coordinate derivatives, as an infinite family of independent abstract variables.\\
Set $\mathcal{P}_0 = 0$. The generic symbol $\mathcal{P}_k$, $k\geq 1$, is an arbitrary polynomial in the components of $\Psi(\ell)$ and $\overline{\Psi(\ell)}$, $0 \leq \ell \leq k-1$, and all their first coordinate derivatives ($\tfrac{\p}{\p x^{\mu}}\Psi(\ell)$ and $\tfrac{\p}{\p x^{\mu}}\overline{\Psi(\ell)}$, $\mu=1,2,3$), whose coefficients are (complex) polynomials in $\mathfrak A$, $\u{u}$, $\b{e}_{a,\mathfrak A}$, $\boldsymbol{\lambda}_{a,\mathfrak A}$, $\overline{\boldsymbol{\lambda}_{a,\mathfrak A}}$, and all their first coordinate derivatives. It is further required that the polynomial $\mathcal{P}_k$ have no constant term, that is, $\mathcal{P}_k$ vanishes when $\Psi(\ell)$ and $\frac{\p}{\p x^{\mu}}\Psi(\ell)$ vanish for all $0 \leq \ell \leq k-1$ and $\mu = 1,2,3$. We use the same symbol $\mathcal{P}_k$ for a vector or matrix all whose entries are polynomials of this kind.
\end{definition}
\begin{proposition} \label{prop:gff}
$[\,\Psi\,]$ is a formal power series solution to \eqref{shspsi1}, with
 $[\,\Minkowski_{a,\mathfrak A}\,]$ in the role of $\Minkowski_{a,\mathfrak A}$,
 if and only if its coefficients $\Psi(k)$, $k\geq 0$, satisfy a system of the form
\begin{subequations}\label{SHS_}
\begin{align}
\label{SHS_z1} \PsiThree_1(k) & = \mathcal{P}_k  & k & > 0 \displaybreak[0] \\
\label{SHS_z2} \PsiThree_2(k) & = \mathcal{P}_k  & k & > 0 \displaybreak[0]\\ 
\label{SHS_z3} \PsiThree_3(k) & = \mathcal{P}_k  & k & > 0 \displaybreak[0]\\  
\label{SHS_z5}  \tfrac{\p}{\p \u{u}}\, \PsiThree_5(k) & = \mathcal{P}_k & k & \geq 0 \displaybreak[0]\\
\label{SHS_z4} (1-\delta_{k0})\PsiThree_4(k) & = -\tfrac{2}{k-\delta_{k0}}\,\big(\b{e}\,\tfrac{\p}{\p \overline{\xi}} +2\boldsymbol{\lambda} \big)\PsiThree_5(k) + \mathcal{P}_k & k & \geq 0 \displaybreak[0]\\
\label{SHS_o1} \tfrac{\p}{\p \u{u}}\, \PsiTwo_1(k) & = -\PsiThree_1(k) + \mathcal{P}_k & k & \geq 0 \displaybreak[0]\\
\label{SHS_o2}  \tfrac{\p}{\p \u{u}}\, \PsiTwo_2(k) & =  - (2-\delta_{k0})\, \Re \big(\PsiTwo_1(0)\, \overline{\PsiTwo_1(k)}\,\big) + \mathcal{P}_k & k & \geq 0 \displaybreak[0]\\
\label{SHS_o3} \tfrac{\p}{\p \u{u}}\, \PsiTwo_3(k) & = -\PsiThree_2(k)  - \overline{\boldsymbol{\lambda}}\,\PsiTwo_1(k) + \mathcal{P}_k & k & \geq 0 \displaybreak[0]\\
\label{SHS_o4}  \tfrac{\p}{\p \u{u}}\, \PsiTwo_4(k) & =  -  \boldsymbol{\lambda}\,\overline{\PsiTwo_1(k)} + \mathcal{P}_k & k & \geq 0 \displaybreak[0]\\
\label{SHS_o5} \PsiTwo_5(k) & = -\tfrac{1}{k+1}\big(\PsiTwo_4(k) - \overline{\PsiTwo_3(k)}\,\big) + \mathcal{P}_k & k & \geq 0 \displaybreak[0]\\
\label{SHS_o6} \PsiTwo_6(k) & = \mathcal{P}_k & k & > 0 \displaybreak[0]\\
\label{SHS_o7} \PsiTwo_7(k) & = \mathcal{P}_k & k & > 0 \displaybreak[0]\\
\label{SHS_o8}  \tfrac{\p}{\p \u{u}}\, \PsiTwo_8(k) & = \PsiThree_3(k) + 2i \Im\,\big(\boldsymbol{\lambda} \big(\PsiTwo_4(k)  - \overline{\PsiTwo_3(k)} - \PsiTwo_5(k)\big)\big) + \mathcal{P}_k & k & \geq 0 \displaybreak[0]\\
\label{SHS_f1}  \tfrac{\p}{\p \u{u}}\, \PsiOne_1(k) & =  \phantom{-i} \,\b{e}\, \PsiTwo_1(k) + \mathcal{P}_k & k & \geq 0 \displaybreak[0]\\
\label{SHS_f2}  \tfrac{\p}{\p \u{u}}\, \PsiOne_2(k) & = -i\,\b{e}\, \PsiTwo_1(k) + \mathcal{P}_k & k & \geq 0 \displaybreak[0]\\
\label{SHS_f3}  \PsiOne_3(k) & = -\tfrac{2}{k+2}\, \Re\, \PsiTwo_8(k) + \mathcal{P}_k & k & \geq 0 \displaybreak[0]\\
\label{SHS_f4}  \tfrac{\p}{\p \u{u}}\, \PsiOne_4(k) & = \phantom{-}2\,\b{e}\,\Re\,\big(\PsiTwo_4(k) - \overline{\PsiTwo_3(k)} - \PsiTwo_5(k)\big) + \mathcal{P}_k & k & \geq 0 \displaybreak[0]\\
\label{SHS_f5} \tfrac{\p}{\p \u{u}}\, \PsiOne_5(k) & = -2\,\b{e}\, \Im \,\big(\PsiTwo_4(k) - \overline{\PsiTwo_3(k)} - \PsiTwo_5(k)\big) + \mathcal{P}_k & k & \geq 0
\end{align}
\end{subequations}
with $\b{e} = \b{e}_{a,\mathfrak A}$ and $\boldsymbol{\lambda} = \boldsymbol{\lambda}_{a,\mathfrak A}$.
\end{proposition}
\begin{proof} Substitute the formal series \eqref{fps} into the relevant\,/\,irrelevant form of system \eqref{shspsi1} given in Proposition \ref{prop:ri1}. Collect all coefficients of common powers of $\frac{1}{u}$. \qed
\end{proof}
\begin{lemma} \label{lem:fpsXX}
For all $a,\mathfrak A \neq 0$, $\u{u}_0 > 0$, all smooth $\text{\bf \small DATA}(\xi,\u{u}): \R^2 \times (0,\infty) \to \C$ that vanish when $\u{u} < \u{u}_0$, there is a unique formal power series $[\,\Psi\,]$ on $\STRIP_{\infty}$, which satisfies \eqref{shspsi1} and
\begin{equation}\label{eq:skdhkdshkdshksdhskhfd}
[\,\Psi\,] = 0\;\; \text{when $\u{u} < \u{u}_0$}
\qquad \qquad \text{$\Psi(0)$ is given by $\eqref{eq:thezeroorderequations}$}
\end{equation}
\end{lemma}
\begin{proof}
 $\Psi(0)$, as given by $\eqref{eq:thezeroorderequations}$, satisfies the $k=0$ equations in \eqref{SHS_}. The coefficient functions $\Psi(k)$, $k\geq 1$, are constructed by induction. For each step $k$, equations \eqref{SHS_z1} to \eqref{SHS_f5} are solved exactly in this order to obtain $\Psi(k)$. The right hand side is explicitly known by induction and the ``upper triangular'' structure of \eqref{SHS_z1} to \eqref{SHS_f5}. Whenever $\tfrac{\p}{\p \u{u}}$ appears on the left hand side, it is inverted using $\p_{\u{u}}^{-1}$, because the constant of integration is zero by the first condition in \eqref{eq:skdhkdshkdshksdhskhfd}. By induction, one also verifies that $\Psi(k)$, $k\geq 0$, vanishes when $\u{u} < \u{u}_0$, so that the first condition in \eqref{eq:skdhkdshkdshksdhskhfd} is satisfied at all orders. It is essential at precisely this point that the generic polynomial $\mathcal{P}_k$ in Definition \ref{def:gsn} has no constant term. Finally, by Proposition \ref{prop:gff}, there exists a formal power series solutions satisfying the hypothesis of the lemma. The construction given here is forced at every step, and therefore generates a unique formal power series. \qed
\end{proof}
\begin{remark} Lemma \ref{lem:fpsXX} is simpler than Proposition \ref{prop:fps}, because it \emph{assumes} equations \eqref{eq:thezeroorderequations} and it makes no statement about the formal constraint field $[\,\Psi^\sharp\,]$.
\end{remark}
\begin{proof}[of Proposition \ref{prop:fps}] \label{page:dkhdskhfkhfdk} 
We first prove existence. It suffices to show that the formal power series $[\,\Psi\,]$ produced by Lemma \ref{lem:fpsXX} satisfies $[\,\Psi^{\sharp}\,] = 0$. (The formal finite speed of propagation
statement in Proposition \ref{prop:fps} follows from an examination of the construction of $[\,\Psi\,]$ in the proof of Lemma \ref{lem:fpsXX}.)
 Note that
\begin{list}{\labelitemi}{\leftmargin=1em}
\item $[\,\Psi^{\sharp}\,]$ is a formal power series solution to the linear homogeneous system \eqref{shspsi2}.
\item $[\,\Psi^{\sharp}\,] = 0$ when $\u{u} < \u{u}_0$. \rule{0pt}{10pt} 
\item $\Psi^{\sharp}(0) = 0$ on $\R^2 \times (0,\infty)$. \rule{0pt}{10pt} 
\end{list}
The first bullet follows from Proposition \ref{hrpart2}, because $[\,\Psi\,]$ is a formal power series solution to \eqref{shspsi1}. The second bullet follows from the first condition in  \eqref{eq:skdhkdshkdshksdhskhfd}, which implies $[\,\Phi\,] = [\,\Minkowski_{a,\mathfrak A}\,]$ when $\u{u} < \u{u}_0$, and $[\,\Phi^{\sharp}\,] = [\,\Minkowski_{a,\mathfrak A}^{\hskip2pt \sharp}\,] = 0$.
For the third bullet, note that $\PsiSharpTwo_5(0)$, $\PsiSharpThree_1(0)$, $\PsiSharpThree_2(0)$, $\PsiSharpThree_3(0)$, $\PsiSharpTwo_6(0)$ all vanish on $\R^2 \times (0,\infty)$ by
the second condition in \eqref{eq:skdhkdshkdshksdhskhfd}.
 By the first two bullets and by equation \eqref{ri3:1}, we conclude, step by step, that $\PsiSharpOne_1(0)$, $\PsiSharpOne_2(0)$, $\PsiSharpTwo_1(0)$, $\PsiSharpTwo_2(0)$, $\PsiSharpTwo_3(0)$, $\PsiSharpTwo_4(0)$, $\PsiSharpOne_4(0)$, $\PsiSharpOne_5(0)$, $\PsiSharpTwo_7(0)$, $\PsiSharpTwo_8(0)$ also vanish. The first equation in \eqref{ri3:2} gives $\PsiSharpOne_3(0) = 0$. It remains to show that $\PsiSharpTwo_9(0) = 0$ on $\R^2 \times (0,\infty)$. By \eqref{ri2:3},
$$\PsiSharpTwo_9(0) = -2\big(\b{e}\, \tfrac{\p}{\p \overline{\xi}}
+2\,\boldsymbol{\lambda}\big)\, \PsiTwo_7(0)
 + 2\,\b{e}\, \tfrac{\p}{\p \xi} \PsiTwo_6(0) + \PsiTwo_4(0) - \overline{\PsiTwo_3(0)}.$$
The second condition in \eqref{eq:skdhkdshkdshksdhskhfd} implies $(\tfrac{\p}{\p \u{u}})^2 \PsiSharpTwo_9(0) = 0$. By the second bullet, $p_9(0) \equiv 0$.

The three bullets imply, by induction on $k\geq 1$, that $\Psi^{\sharp}(k) = 0$ on $\R^2 \times (0,\infty)$. In fact, at each  step $k$, one verifies, in the given order, that $\PsiSharpThree_1(k)$, $\PsiSharpThree_2(k)$, $\PsiSharpThree_3(k)$ all vanish by \eqref{ri3:3}, $\PsiSharpTwo_1(k)$, $\PsiSharpTwo_2(k)$, $\PsiSharpTwo_3(k)$, $\PsiSharpTwo_4(k)$ all vanish by \eqref{ri3:1}, $\PsiSharpTwo_5(k)$, $\PsiSharpTwo_6(k)$ both vanish by \eqref{ri3:2}, $\PsiSharpTwo_7(k)$, $\PsiSharpTwo_8(k)$ both vanish by \eqref{ri3:1}, $\PsiSharpTwo_9(k)$, $\PsiSharpOne_3(k)$ both vanish by \eqref{ri3:2}, and $\PsiSharpOne_1(k)$, $\PsiSharpOne_2(k)$, $\PsiSharpOne_4(k)$, $\PsiSharpOne_5(k)$ all vanish by \eqref{ri3:1}.
This concludes the existence proof.

Uniqueness in Lemma \ref{lem:fpsXX} implies uniqueness in Proposition \ref{prop:fps}, because we now show that \eqref{shspsi1} and $[\,\Psi^{\sharp}\,] = 0$ and \eqref{eq:fordhhkhkshkfhfdforfor} together imply  \eqref{eq:thezeroorderequations}, which is the second condition in \eqref{eq:skdhkdshkdshksdhskhfd}. Condition \eqref{eq:fordhhkhkshkfhfdforfor} and the $k=0$ equations in \eqref{SHS_}
imply \eqref{eq:thezeroorderequations}, apart from the formulas for $\PsiTwo_7(0)$, $\PsiThree_2(0)$, $\PsiThree_3(0)$, $\PsiThree_4(0)$, $\PsiTwo_6(0)$.
The remaining five formulas follow from the vanishing of $\PsiSharpTwo_5(0)$, $\PsiSharpThree_1(0)$, $\PsiSharpThree_2(0)$, $\PsiSharpThree_3(0)$ and $\PsiSharpTwo_6(0)$, see \eqref{ri2:3} and \eqref{ri2:4}. Here, $\Psi^{\sharp}(0) = (\PsiSharpOne(0),\PsiSharpTwo(0),\PsiSharpThree(0))$. \qed
\end{proof}
\begin{proposition} \label{prop:coeffest} For all $k,R \geq 0$, all $0 < |\mathfrak A| \leq |a| \leq 1$, and all $\text{\bf \small DATA}$,
$$\|\Psi(k)\|_{C^R(\mathcal{Q})} \; \leq \; p_{k,R}\big( \| \text{\bf \small DATA}\|_{C^{R+2k+3}(\mathcal{Q})}\big)\qquad \qquad \mathcal{Q} = D_{4|\frac{a}{\mathfrak A}|}(0) \times (0,2)$$
 where $[\,\Psi\,]$ is the corresponding formal solution in Proposition \ref{prop:fps}, and $p_{k,R}: \R \to \R$, is an infinite family, indexed by $k,R \geq 0$, of universal polynomials without constant term. Here, $D_r(0)$ is the open disk of radius $r>0$ in the $(\xi^1,\xi^2)$-plane. Here
$\|f\|_{C^R(\mathcal{Q})} = \sup_{|\alpha|\leq R} \|\p^{\alpha}f\|_{C^0(\mathcal{Q})}$, where $\alpha \in \N_0^3$.
\end{proposition}
\begin{proof} Observe that:
\begin{list}{\labelitemi}{\leftmargin=1em}
\item $\|\b{e}_{a,\mathfrak A}\|_{C^R(\mathcal{Q})} \leq \frac{17}{2}$ and $\|\boldsymbol{\lambda}_{a,\mathfrak A}\|_{C^R(\mathcal{Q})} \leq \frac{17}{2}$ for all $R\geq 0$.
\item $\|\p_{\u{u}}^{-1} g\|_{C^R(\mathcal{Q})} \leq 2\,\|g\|_{C^R(\mathcal{Q})}$ for all $R\geq 0$ and all functions $g = g(\xi,\u{u})$ on $\mathcal{Q}$.
\end{list}
The existence of polynomials $p_{0,R}$, $R\geq 0$, follow by direct inspection of \eqref{eq:thezeroorderequations}. The existence of polynomials $p_{k,R}$, $R \geq 0$, is shown by induction over $k\geq 0$. At each step $k \geq 1$, we use \eqref{SHS_}.
By the inductive hypothesis and Definition \ref{def:gsn} there is a polynomial $p'_{k,R}$ (depending only on $k$ and $R$) so that each generic term $\mathcal{P}_k$ on the right hand sides of \eqref{SHS_} satisfies $\|\mathcal{P}_k\|_{C^R(\mathcal{Q})} \leq p'_{k,R}\big(\|\text{\bf \small DATA}\|_{C^{R+2k+2}(\mathcal{Q})}\big)$. We can assume that $p'_{k,R}$ has no constant term, because $\mathcal{P}_k$ does not have one (see, Definition \ref{def:gsn}). Now, the existence of $p_{k,R}$, $R\geq 0$ follows directly from estimating the non generic terms on the right hand sides of \eqref{SHS_z1} to \eqref{SHS_f5}, exploiting the upper triangular structure. Only in one equation, \eqref{SHS_z4}, a coordinate derivative appears. \qed
\end{proof}
\begin{remark}
In Proposition \ref{prop:coeffest}, the uniformity of the estimate in $a$, $\mathfrak A$, when $0 < |\mathfrak A| \leq |a| \leq 1$, will be exploited later. It is compatible with taking the limit $a = \mathfrak A \downarrow 0$.
\end{remark}
\begin{remark} \label{rem:subSHSformalXX}
Fix $\text{\bf \small DATA}$ and let $[\,\Psi_{a,\mathfrak A}\,]$ be the formal power series solution in Proposition \ref{prop:fps}. The indices have been added to make the dependence on $a,{\mathfrak A}\neq 0$ explicit.
One can show, by induction, that $\Psi_{\mathfrak A,\mathfrak A}(k)(\xi,\u{u})$, $k\geq 0$, are polynomials in $\mathfrak A$.
\end{remark}
\begin{proposition}[Matching Stereographic Charts] \label{prop:flipchart}
Choose $a,\mathfrak A \neq 0$. Pick $\text{\bf \small DATA}^{\sigma}$ as in
 Proposition \ref{prop:fps}, for $\sigma = -,+$, and let $[\,\Psi^{\sigma}\,]$ be the associated solution in Proposition \ref{prop:fps}. The following statements are equivalent:
\begin{list}{\labelitemi}{\leftmargin=1em}
\item $\frac{|\xi|^2}{\xi^2}\text{\bf \small DATA}^\sigma\big(\tfrac{a}{\mathfrak A} \xi,\,\u{u}\big) \; = \; \tfrac{\xi^2}{|\xi|^2}\, \text{\bf \small DATA}^{-\sigma}\big(\tfrac{a}{\mathfrak A}\, \tfrac{1}{\xi},\,\u{u}\big)$ when $\xi \neq 0$.
\item $\FLIP_{\frac{a}{\mathfrak A}}\cdot [\,\Phi^{\sigma}\,] = [\,\Phi^{-\sigma}\,]$ when $\xi \neq 0$. Here, $[\,\Phi^{\sigma}\,] = [\,\Minkowski_{a,\mathfrak A}\,] + u^{-M}[\,\Psi^{\sigma}\,]$.
\item $\FLIP_{\frac{a}{\mathfrak A}}\cdot [\,\Psi^{\sigma}\,] = [\,\Psi^{-\sigma}\,]$ when $\xi \neq 0$.
\end{list}
Here, $-\sigma = +$ when $\sigma = -$, and conversely, $-\sigma = -$ when $\sigma = +$.
\end{proposition}
\begin{proof} The equivalence of the last two bullets follows from Proposition \ref{minkprop}, (b), and the fact that $\FLIP_{\frac{a}{\mathfrak A}}$ commutes with multiplication by $u^{-M}$. Each of the last two bullets implies the first. Just look at how $\FLIP_{\frac{a}{\mathfrak A}}$ acts on the component $\PsiTwo_1$. The first bullet implies the last two, because $\FLIP_{\frac{a}{\mathfrak A}}$ is a field symmetry, and by uniqueness in Proposition \ref{prop:fps} (more precisely, by formal finite speed of propagation).  \qed
\end{proof}
\begin{definition} \label{def:flipdata}
For all $(\xi,\u{u}) \in \R^2 \times (0,\infty)$ with $\xi \neq 0$, set
$$\big(\FLIP_{\frac{a}{\mathfrak A}}\cdot \text{\bf \small DATA}\big)(\xi,\u{u}) \; = \; \tfrac{\xi^2}{\overline{\xi}^2}\, \text{\bf \small DATA}\big(\tfrac{a^2}{{\mathfrak A}^2}\, \tfrac{1}{\xi},\,\u{u}\big)$$
\end{definition}
\begin{remark} \label{endof6}
Proposition \ref{prop:fps}, the main result of this section, is the formal analog of Theorem \ref{mainenergy}, the main result of this paper.
They give solutions to an asymptotic characteristic initial value problem that is motivated by \cite{Chr}. Informally:
\begin{subequations}
\begin{align}
\label{eq:ic1XX} \textstyle\lim_{u \to -\infty} \Psi(\xi,\u{u},u) \; & = \; \Psi(0)(\xi,\u{u})\\
\label{eq:ic2XX} \Psi(\xi,\u{u},u) \; & = \; 0 \qquad \text{when $\u{u} < \u{u}_0$}
\end{align}
\end{subequations}
with the understanding that $\Psi(0)$ is given in terms of $\text{\bf \small DATA}(\xi,\u{u})$ by equations \eqref{eq:thezeroorderequations}. Equation \eqref{eq:ic2XX} stipulates that $\Phi$ coincides with the Minkowski vacuum field $\mathcal{M}_{a,\mathfrak A}$ when $\u{u} < \u{u}_0$. On the other hand, \eqref{eq:ic1XX} is an asymptotic initial condition at $u \to -\infty$, past null infinity.  \emph{All the notation, definitions and concepts required for Theorem \ref{mainenergy} have now been introduced. It can be read on its own.}
\end{remark}

%% file: SectionEnergyEstimates.tex
\section{Quasilinear Symmetric Hyperbolic Systems} \label{sec:energyestimates}
We discuss an abstract local existence theorem and an abstract energy estimate.
\begin{convention} \label{conv:notation}
In this section, we use the coordinates $q$ instead of $x$:
\begin{align*}
x  & = (x^1,x^2,x^3,x^4) = (\xi^1,\xi^2,\u{u},u)  \\
\iv &  = (\iv^0,\iv^1,\iv^2,\iv^3)= (t,\xi^1,\xi^2,\u{u})
& t &= u + \u{u}
& \b{\iv} & = (\iv^1,\iv^2,\iv^3)
\end{align*}
$B_r(p)\subset \R^n$ is the open ball of radius $r > 0$ around $p$, and $D_r(p) = B_r(p)$
if $n=2$.\\
${\rm Mat}_n$ (${\rm Sym}_n$) is the vector space of all $n\times n$ real (symmetric) matrices.\\
$X \lesssim_a Y$ means $X \leq C Y$ for a constant $C > 0$ that depends only on $a \in \R^k$.\\
$X \lesssim Y$ means $X \leq CY$ for a universal constant $C > 0$.
\end{convention}
\begin{proposition}[Existence\,/\,Breakdown Theorem] \label{exexex} Suppose a tuple
$$
\big(T,\,P,\,\b{M}^{\mu},\,h,\,
\asympt{\b{M}}^{\mu},\,
\asympt{H},\,\mathcal{K},\,\mathcal{Q}\big)
$$
($\mu=0,1,2,3$)
and the derived tuple $(\mathcal{U},\mathcal{A})$ satisfy:
\begin{itemize}
\item[{\bf \small (EB0)}] $T \in \R$ and $\mathcal{U} = (-\infty,T) \times \R^3$ and $P \in \N$ and $\mathcal{A} = \mathcal{U}\times B_2(0) \subset \R^4 \times \R^P$.
\item[{\bf \small (EB1)}] $\mathbf{M}^{\mu} \in C^\infty(\,\overline{\mathcal{A}},{\rm Sym}_P)$ and
 $\b{M}^0 \geq \tfrac{1}{2}$ and $h \in C^\infty(\,\overline{\mathcal{A}},\R^P)$. 
\item[{\bf \small (EB2)}] $\asympt{\mathbf{M}}^{\mu} \in {\rm Sym}_P$ and
 $\asympt{\b{M}}^0 \geq \tfrac{1}{2}$ and $\asympt{H} \in C^{\infty}(\, (-\infty,T], {\rm Mat}_P)$.
\item[{\bf \small (EB3)}] $\mathcal{K} \subset \mathcal{Q} \subset \R^3$ and $\mathcal{K}$ compact and $\mathcal{Q}$ open. Moreover $\mathbf{M}^{\mu}(q,\Theta) = \asympt{\mathbf{M}}^{\mu}$
and $h(\iv,\Theta) = \asympt{H}(t)\,\Theta$ for all 
$(\iv,\Theta) \in ((-\infty,T) \times (\R^3 \setminus \mathcal{K})) \times B_2(0)$. Here $t = \iv^0$.
\end{itemize}
Extend $\b{M}^{\mu}$ and $h$, by $\asympt{\b{M}}^\mu$ and $\asympt{H}(t)\Theta$, to $(-\infty,T)\times (\R^3 \setminus \mathcal{K})\times \R^P$. Then:
\vskip 1mm
\noindent {\rm Part 1.}
For each $t_0 < T$, there is a $t_1 \in (t_0,T]$ and a $C^{\infty}$-solution
\begin{equation}\label{eq:SHS}
\Theta: [t_0,t_1)\times \R^3 \to \R^P
\qquad \text{of} \qquad
\mathbf{M}^{\mu}(\iv,\Theta)\, \tfrac{\p}{\p \iv^{\mu}}\Theta = h(\iv,\Theta)
\end{equation}
with $\Theta(t_0,\,\cdot\,) \equiv 0$, such that
$\supp \Theta \subset [t_0,t_1)\times B_r(0)$ for some finite $r > 0$, and
\begin{equation}\label{eq:bound}
\Theta\big([t_0,t_1)\times \mathcal{Q}\big) \subset B_1(0) \subset \R^P
\end{equation}
 and such that $t_1 \neq T$ implies either one or both of:
\begin{description}
\item[$(\text{\bf \small Break})_1$:] \quad $\overline{\Theta([t_0,t_1)\times \mathcal{Q})}\not\subset B_1(0) \subset \R^P$.
\item[$(\text{\bf \small Break})_2$:] \quad The vector field $\p_q\, \Theta$ is unbounded on $[t_0,t_1)\times \mathcal{Q}$.
\end{description}
{\rm Part 2.} If $\b{M}^3 \geq 0$ on $\mathcal{A}$ and if $h(\iv,0) = 0$ when $\iv^3 < \tfrac{1}{2}$, then $\Theta|_{q^3 < 1/2} \equiv 0$.
\end{proposition}
\begin{remark} Proposition \ref{exexex} is stated without proof. The main ingredient is a standard existence theorem for quasilinear symmetric hyperbolic systems as in \cite{Tay}, pages 360-370. The specific geometry (of space, time \emph{and} target space) in Proposition \ref{exexex} can be reduced to the standard geometry in \cite{Tay} using partitions of unity and finite speed of propagation. {\bf \small (EB3)} implies that \eqref{eq:SHS} reduces to 
the linear homogeneous $\asympt{\mathbf{M}}^{\mu} \tfrac{\p}{\p \iv^{\mu}} \Theta = \asympt{H}(t) \Theta$
for all $\mathbf{q}$ not in the compact $\mathcal{K}$, with
$\asympt{\mathbf{M}}^{\mu}$ constant.
\end{remark}
\begin{lemma}
Recall the matrix differential operators $\b{A}(\Phi)$ and $\widehat{\b{A}}(\Phi)$ associated to $\Phi = (\PhiOne,\PhiTwo,\PhiThree)$, see Section \ref{sec:geomety}. Suppose $\theta$ is a one-form and suppose
\begin{align}\label{eq:condicondi}
\theta(L) & \geq 0 &
\theta(N) & \geq 0 &
\begin{pmatrix} \theta(N) & \theta(D)\\ \theta(\overline{D}) & \theta(L) \end{pmatrix} & \geq 0
\end{align}
($ \geq 0$ as Hermitian matrices). Then 
\begin{align}\label{eq:condi}
\theta\big(\b{A}(\Phi)\big) & \geq 0 & 
\theta\big(\widehat{\b{A}}(\Phi)\big) & \geq 0
\end{align}
\end{lemma}
\begin{definition} \label{def:jhauwdhdfkhfkdd}
For all $\xi_0 \in \R^2$ and $b \in [1,2]$ set
\begin{align*}
r_b'(t,\u{u}) & = \tfrac{1}{4} + \tfrac{1}{0.018}|b+ 0.001 - \u{u}|\cdot\big|0.002 - \tfrac{1}{\u{u} + |t|}\big|\\
\theta_{b,\xi_0} & = \text{the differential of the function $q \mapsto |\xi - \xi_0| - r_b'(t,\u{u})$}
\end{align*}
\end{definition}
\begin{proposition} \label{prop:uzewuztuztwqe873}
Let $\xi_0 \in \R^2$ and $b \in [1,2]$. Suppose $q = (t,\xi^1,\xi^2,\u{u})$ and the parameters $a,\anisotropic \in \R$ and the field $\Phi = \Minkowski_{a,\anisotropic} + u^{-M}\,\Psi$ (see Section \ref{sec:ansatz}) 
 satisfy
\begin{itemize}
\item $\u{u}\in (0,b)$ and $t \in (-\infty,-1000)$.
\item $|\xi - \xi_0| = r_b'(t,\u{u})$ and $|\xi| < 4|\tfrac{a}{\anisotropic}|$.
\item $0 < |\anisotropic| \leq |a| \leq 10^{-3}$ and $|\Psi(q)| \leq 5$.
\end{itemize}
Then \eqref{eq:condi} holds at $q$, for  the one-form $\theta_{b,\xi_0}$ in Definition \ref{def:jhauwdhdfkhfkdd}.
That is, the hypersurface $|\xi - \xi_0| = r_b'(t,\u{u})$ is non-timelike at $q$ with respect to
$\b{A}(\Phi)$ and $\widehat{\b{A}}(\Phi)$.
\end{proposition}
\begin{remark}
To prove Proposition \ref{prop:uzewuztuztwqe873}, calculate $\theta_{b,\xi_0} = \dd(|\xi - \xi_0| - r_b'(t,\u{u}))$
and use \eqref{eq:DNL} to check \eqref{eq:condicondi}, keeping in mind the coordinate transformation between $x$ and $q$ in Convention \ref{conv:notation}. Assumption $|\Psi(q)| \leq 5$ implies $|f_1(q)|,\ldots,|f_5(q)| \leq 5$, and only these five inequalities are used. Here $\Psi = (\PsiOne,\PsiTwo,\PsiThree)$.
\end{remark}
\begin{definition} [Energy and Supremum] \label{def:aasaa}
For every integer $k \geq 0$, every open $\mathcal{X}\subset \R^3$, every $t \in \R$ and every scalar\,/\,vector\,/\,matrix valued $C^k$-function $f = f(t,\mathbf{q})$, set\footnote{$\p^{\alpha} = \prod_{\mu=0}^3 (\tfrac{\p}{\p \iv^{\mu}})^{\alpha_{\mu}}$. The pointwise norm $|\,\cdot\,|$ is the Euclidean norm. For a matrix, $|A|^2 = \tr (A^T A)$.}
\begin{equation*} 
E^k_{\mathcal{X}}\{f\}(t)  = \sum_{\substack{|\alpha|\leq k \\ \alpha\in \N_0^4}} \int_{\mathcal{X}} \dd^3\mathbf{\iv}\, |\p^{\alpha} f(t,\mathbf{\iv})|^2 \qquad\;\;\;
\SUP^{(k)}_{\mathcal{X}}\{f\}(t)  = \sup_{\substack{|\alpha|\leq k \\ \alpha \in \N_0^4}}\; \sup_{\mathbf{\iv} \in \mathcal{X}} |\p^{\alpha}f(t,\mathbf{\iv})|
\end{equation*}
\end{definition}
\begin{proposition}[Energy estimate] \label{prop:fdsdsfkjhdskhds} Suppose a tuple
\begin{equation*}\Big(
(t_0,  t^{\ast}), \xi_0,  b, 
 P_m,  R, \mathbf{M}^{\mu}, H,  \source,  \Theta,  
 \asympt{\mathbf{M}}^{\mu},
 \asympt{H}, \mathbf{c}_1, \mathbf{c}_2,  J\Big)
 \end{equation*}
($\mu = 0,1,2,3$ and $m=1,2,3$)
and the derived tuple $(\mathcal{I},\mathcal{U},P,M^{\mu}_m,\asympt{H}_m)$ 
satisfy
{\bf \small (RE0)} through {\bf \small (RE12)} and one of {\bf \small (RE13a)}, {\bf \small (RE13b)}:
\begin{itemize}
\item[{\bf\small (RE0)}] $-\infty < t_0 < t^{\ast} < -1000$ and
$\mathcal{I} = (t_0,t^{\ast})$ and $\xi_0 \in \R^2$ and $b \in [1,2]$ and
\begin{align*}
\mathcal{U} & = \textstyle\bigcup_{t \in \mathcal{I}} \big(\{t\} \times \mathcal{O}(\xi_0,b,t)\big) &\;\;\;
\mathcal{O}(\xi_0,b,t) & = \textstyle\bigcup_{\u{u} \in (0,b)} \big(D_{r_b'(t,\u{u})}(\xi_0)\times \{\u{u}\}\big)
\end{align*}
\item[{\bf\small (RE1)}] $P_1,P_2,P_3\in \N$ and $P = P_1 + P_2 + P_3 \leq 10^9$
and $R \in \N_0$ and $\smallRho,\smallSHS,J\in \R$.
\item[{\bf\small (RE2)}]
$\b{M}^{\mu} \in C^{R+1}(\overline{\mathcal{U}},{\rm Sym}_P)$
and $H \in  C^{R}(\overline{\mathcal{U}},{\rm Mat}_P)$
and $\source \in C^{R}(\overline{\mathcal{U}},\R^P)$.
\item[{\bf\small (RE3)}] $\tfrac{1}{2} \leq \mathbf{M}^0 \leq 2$ and $\b{M}^3 \geq 0$.
\item[{\bf\small (RE4)}]
$\theta_{\mu}\,\mathbf{M}^{\mu} \geq 0$ on $(\p \mathcal{U}) \cap (\mathcal{I}\times \R^2 \times (0,b))$
with $\theta = \theta_{b,\xi_0}$ as in Definition \ref{def:jhauwdhdfkhfkdd}.
\item[{\bf\small (RE5)}] $\Theta \in C^{R+1}(\overline{\mathcal{U}},\R^P)$
 is a solution to the linear symmetric hyperbolic system
\begin{equation}\label{eq:SHSSHS}
\mathbf{M}^{\mu}(\iv) \tfrac{\p}{\p \iv^\mu} \Theta = H(\iv)\Theta + \source(\iv)
\qquad \text{on $\mathcal{U}$}
\end{equation}
\item[{\bf\small (RE6)}] 
$\Theta$ and $\source$ vanish identically when $\iv^3 < \tfrac{1}{2}$.
\item[{\bf\small (RE7)}]  $\asympt{\b{M}}^{\mu} \in {\rm Sym}_P$
and $\s{H} \in C^{\infty}(\overline{\mathcal{I}},{\rm Mat}_P)$.
\item[{\bf\small (RE8)}] $\tfrac{1}{2}\leq \asympt{\b{M}}^0\leq 2$ and  $\asympt{\b{M}}^1 = \asympt{\b{M}}^2 = 0$ and $\asympt{\b{M}}^3 \geq 0$.
\end{itemize}
Make the $\R^P = \R^{P_1} \oplus \R^{P_2} \oplus \R^{P_3}$ block decompositions
\begin{align*}
\Theta & = (\Theta_m) &
\source & = (\source_m) & 
\mathbf{M}^\mu & = (M_{mn}^{\mu}) &
H & = (H_{mn}) &
\asympt{H} & = (\asympt{H}_{mn})
\end{align*}
with $m,n = 1,2,3$. 
\begin{itemize}
\item[{\bf \small (RE9)}] $M^{\mu}_{mn} = 0$ if $m \neq n$ and $M^{\mu}_{mm} = M^{\mu}_m$ for $m=1,2,3$.
\item[{\bf \small (RE10)}] $M_{22}^{\mu}\tfrac{\p}{\p \iv^{\mu}} = \nu(\iv)\, \mathbbm{1}_{P_2}\,(\frac{\p}{\p \iv^0} + \frac{\p}{\p \iv^3})$ for some function $\nu$ on $\mathcal{U}$.
\item[{\bf \small (RE11)}] 
\begin{equation*}
\big(\asympt{H}_{mn}\big)  = \begin{pmatrix} 0 & 0 & 0 \\ \asympt{H}_1 & 0 & 0 \\ 0 & |t|^{-1}\asympt{H}_2 & |t|^{-1}\asympt{H}_3 \end{pmatrix}
\end{equation*}
for real, constant matrices $\asympt{H}_1$, $\asympt{H}_2$, $\asympt{H}_3$, with $\asympt{H}_3 \in {\rm Sym}_{P_3}$ and $\asympt{H}_3 \leq 0$.
\item[{\bf \small (RE12)}] $\smallRho \geq 0$ and $J > 0$ and for all $t \in \mathcal{I}$:
\begin{equation*}
\left.\begin{aligned}
|t|^{2J+2}\, E^R_{\mathcal{O}(\xi_0,b,t)}\{\source_1\}(t)\\
|t|^{2J}\, E^R_{\mathcal{O}(\xi_0,b,t)}\{\source_2\}(t)\\
|t|^{2J+2}\, E^R_{\mathcal{O}(\xi_0,b,t)}\{\source_3\}(t)\\
\end{aligned} \right \} \leq (\smallRho)^2
\end{equation*}
\item[{\bf \small (RE13a)}] $R \geq 4$ and $\smallSHS > 0$ and for all $t \in \mathcal{I}$:
\begin{equation*}
\left.\begin{aligned}
|t|^2\, E^R_{\mathcal{O}(\xi_0,b,t)}\{\mathbf{M}^{\mu} -\, \asympt{\mathbf{M}}^{\mu}\}(t)\\
|t|^2\, E^R_{\mathcal{O}(\xi_0,b,t)}\{H_{1n}-\asympt{H}_{1n}\}(t)\\
E^R_{\mathcal{O}(\xi_0,b,t)}\{H_{2n}-\asympt{H}_{2n}\}(t)\\
|t|^2\, E^R_{\mathcal{O}(\xi_0,b,t)}\{H_{3n}-\asympt{H}_{3n}\}(t)\\
\end{aligned}\right \} \leq (\smallSHS)^2
\end{equation*}
\item[{\bf \small (RE13b)}] $R \geq 0$ and $\smallSHS > 0$ and for all $t \in \mathcal{I}$:
\begin{equation*}
\left.\begin{aligned}
|t|\, \SUP^{(\max\{1,R\})}_{\mathcal{O}(\xi_0,b,t)}\{\mathbf{M}^{\mu} -\, \asympt{\mathbf{M}}^{\mu}\}(t)\\
|t|\, \SUP^{(R)}_{\mathcal{O}(\xi_0,b,t)}\{H_{1n}-\asympt{H}_{1n}\}(t)\\
\SUP^{(R)}_{\mathcal{O}(\xi_0,b,t)}\{H_{2n}-\asympt{H}_{2n}\}(t)\\
|t|\, \SUP^{(R)}_{\mathcal{O}(\xi_0,b,t)}\{H_{3n}-\asympt{H}_{3n}\}(t)\\
\end{aligned}\right \} \leq \smallSHS
\end{equation*}
\end{itemize}
Then, for all $J_0 > 0$, there are constants $\smallEE(X) \in (0,1)$, $\bigEE(X) > 0$ depending only on $X = \big(R,J_0,|\asympt{H}_m|\big)$, such that
$J \geq J_0$ and $\smallSHS \leq \smallEE(X)$ and $|t^{\ast}|^{-1}\leq \smallEE(X)$ imply
\begin{equation}\label{eq:refprop:p2XX}
\sup_{\tau \in \mathcal{I}} |\tau|^J \sqrt{E^R_{\mathcal{O}(\xi_0,b,\tau)}\{\Theta\}(\tau)} \;\leq \bigEE(X)\;\Big(|t_0|^J\,\sqrt{E^R_{\mathcal{O}(\xi_0,b,t_0)}\{\Theta\}(t_0)} \,+\, \smallRho\Big)
\end{equation}
\end{proposition}
\begin{center}
\input{ree.pstex_t}
\end{center}
\begin{proof} In this proof $E^R = E^R_{\mathcal{O}(\xi_0,b,t)}$ and $\SUP^{(R)} = \SUP^{(R)}_{\mathcal{O}(\xi_0,b,t)}$, and $\alpha,\beta \in \N_0^4$.

\emph{Step 1:} For any function $f$ with values in $\R^{Pi}$, $i=1,2,3$, set 
\begin{equation*}
\b{E}_i^0\{f\}(t) = \textstyle\int_{\mathcal{O}(\xi_0,b,t)} \dd^3 \b{\iv} \big(f^TM_i^0f\big)(t,\b{\iv})\qquad 
\b{E}_i^R\{f\}(t) = \textstyle\sum_{|\alpha|\leq R} \b{E}_i^0\{\p^{\alpha}f\}(t)
\end{equation*}
the energy naturally associated to the symmetric hyperbolic system \eqref{eq:SHSSHS}.
By {\bf \small (RE3)}:
\begin{equation}\label{eq:comparable}
E^R\{f\}(t) \leq 2\, \b{E}_i^R\{f\}(t) \qquad \b{E}_i^R\{f\}(t) \leq 2\, E^R\{f\}(t)
\end{equation}
If $R \geq 2$ and $f$ is a vector or matrix valued $C^R$-function, then (Sobolev inequality):
\begin{equation}
\label{eq:sobolev}
\SUP^{(R-2)}\{f\}(t) \lesssim_R \textstyle\sqrt{E^R\{f\}(t)}.
\end{equation}
Here, it is legitimate to use $\lesssim_R$ for $(\xi_0,b,t) \in \R^2\times [1,2]\times (-\infty,-1000)$
. In fact, $D_1(0) \times (0,1)$ is diffeomorphic to $\mathcal{O}(\xi_0,b,t)$ by
$(\xi,\u{u}) \mapsto (\xi_0 + r_b'(t,b\u{u})\xi, b\u{u})$. 
Since all derivatives of order up to $R-1$ of the Jacobians of both 
the diffeomorphism and its inverse have finite sup-norms \emph{not depending on} $\xi_0,b,t$, inequality \eqref{eq:sobolev} follows from the analogous Sobolev inequality on $D_1(0)\times (0,1)$, by a change of variables.

By the Sobolev inequality and the product rule (suppress the argument $t \in \mathcal{I}$):
\begin{equation} \label{eq:est12}
\begin{aligned}
E^R\{f_1f_2\} & \lesssim_R E^R\{f_1\}\,E^R\{f_2\}
&& \text{if $R \geq 4$}\\
 E^0\big\{\,[\p^{\alpha},f_1]f_2\,\big\} & \lesssim_R \textstyle\sum_{|\beta| = 1} E^{R-1}\{\p^{\beta} f_1\}\, E^{R-1}\{f_2\}
&& \text{if $R \geq 4$, $|\alpha| \leq R$}\\
E^R\{f_1f_2\} & \lesssim_R \big(\SUP^{(R)}\{f_1\}\big)^2\,E^R\{f_2\}
&& \text{if $R \geq 0$}\\
E^0\big\{\,[\p^{\alpha},f_1]f_2\,\big\} & \lesssim_R \textstyle\sum_{|\beta| = 1} \big(\SUP^{(R-1)}\{\p^{\beta}f_1\}\big)^2\, E^{R-1}\{f_2\}
&& \text{if $R \geq 0$, $|\alpha| \leq R$}\\
\end{aligned}
\end{equation}
In the fourth inequality, the left hand side vanishes when $R = 0$, because then $\alpha = 0$.

\emph{Step 2:} 
Here $t\in \mathcal{I}$ and $|\alpha|\leq R$. Apply $\p^{\alpha}$ to \eqref{eq:SHSSHS}:
\begin{align}\label{eq:diffeqs}
\mathbf{M}^{\mu}\tfrac{\p}{\p \iv^{\mu}} (\p^{\alpha}\Theta) & = \asympt{H}\, \p^{\alpha}\Theta +  (S_1^{\alpha},\,S_2^{\alpha},\,S_3^{\alpha}) + \p^{\alpha}\source\\
\notag (S^{\alpha}_1,\,S_2^{\alpha},\,S_3^{\alpha}) & \stackrel{\text{def}}{=}  \p^{\alpha}\big((H-\asympt{H})\Theta\big) + [\p^{\alpha},\asympt{H}]\Theta  + [\b{M}^{\mu} - \asympt{\b{M}}^{\mu},\p^{\alpha}]\p_{\mu} \Theta
\end{align}
Assumption {\bf \small (RE11)} and inequalities \eqref{eq:est12} imply
\begin{align*}
& \begin{aligned}
E^0\{S_i^{\alpha}\} \lesssim_R\big\{\textstyle\sum_{j=1}^3 E^R\{H_{ij}-\asympt{H}_{ij}\} + |t|^{-4} |\asympt{H}_2|^2+ |t|^{-4}  |\asympt{H}_3|^2 & \\
+ \textstyle\sum_{\mu = 0}^3 E^R\{\b{M}^\mu - \asympt{\b{M}}^\mu\} \big\} E^R\{\Theta\} &
\end{aligned} & \text{if $R \geq 4$}\\
\rule{0pt}{24pt} & \begin{aligned}
E^0\{S_i^{\alpha}\} \lesssim_R\big\{\textstyle\sum_{j=1}^3 \big(\SUP^{(R)}\{H_{ij}-\asympt{H}_{ij}\}\big)^2 + |t|^{-4}|\asympt{H}_2|^2 + |t|^{-4}|\asympt{H}_3|^2 &\\
+ \textstyle\sum_{\mu = 0}^3 \big(\SUP^{(R)}\{\b{M}^{\mu}-\asympt{\b{M}}^{\mu}\}\big)^2\big\} E^R\{\Theta\} & 
\end{aligned} & \text{if $R \geq 0$}
\end{align*}
Set $\smallEEProof = \max\big\{\smallSHS,|t^{\ast}|^{-1}(|\asympt{H}_2|^2 + |\asympt{H}_3|^2)^{1/2}\big\}$.
If {\bf \small (RE13a)} or {\bf \small (RE13b)}, then
\begin{equation} \label{S:E}
\max \big\{\;
|t|^2\, E^0\{S_1^{\alpha}\},\;
E^0\{S_2^{\alpha}\},\;
|t|^2\, E^0\{S_3^{\alpha}\}
\;\big\}
\;\lesssim_R \;\smallEEProof^2\, E^R\{\Theta\}
\end{equation}


\emph{Step 3:}
We now derive the inequalities \eqref{ineq1}. Recall
{\bf \small (RE2)}, {\bf \small (RE9)}.
The ``energy currents'' associated to $\Theta = (\Theta_i)$, $i=1,2,3$, and their Euclidean divergences are
\begin{align}\label{eq:jcurrentXX}
j^{\mu}_i[\Theta_i] & \stackrel{\text{def}}{=} \Theta_i^TM_i^{\mu}\Theta_i &
\p_{\mu}j^{\mu}_i[\Theta_i] & = \Theta_i^T (\p_{\mu}M_i^{\mu})\Theta_i
+ 2\Theta_i^T \big(M_i^{\mu}\p_{\mu} \Theta_i)
\end{align}
by $(M_i^{\mu})^T=M_i^{\mu}$. \emph{We do not sum over repeated lower indices.}
For $\tau \in \overline{\mathcal{I}}$ set
\begin{align*}
\mathcal{D}_1(\tau) & = \mathcal{D}_3(\tau) = \big\{(t,\b{q}) \in \mathcal{U}\;\big|\; t \in (\tau_-,\tau)\big\}& \;\;\; \tau_- & = \tau_-(\tau) = \max \{t_0,\,\tau -2\} \\
\mathcal{D}_2(\tau) & = \mathcal{D}_1(\tau) \cap \{\,\iv\;|\;\iv^3 -\iv^0 < b - \tau\}
\end{align*}
\noindent \begin{minipage}{8.6cm}
For $\mathcal{D}_2$, see the figure, where $t_0 < \tau_2 < t_0 + b < \tau_1 < t^{\ast}$. Energy estimates are obtained by integrating the divergence identity in \eqref{eq:jcurrentXX} over $\mathcal{D}_i(\tau) \subset \mathcal{U}$ and applying the Euclidean divergence theorem. The divergence theorem generates integrals over the components of $\p \mathcal{D}_i(\tau)$. Their contributions are:
\vskip 2mm
\begin{tabular}{c || c | c | c}
\emph{Boundary} & \emph{Contribution} & $i=$ & \emph{Remark}\\
\hline
$q^0 = \tau$ & $\b{E}_i^0\{\Theta_i\}(\tau)$ & all & \\
$q^3 = 0$ & $0$ & all & {\bf \small (RE6)}\\
$q^0 = \tau_-$ & $ -\b{E}_i^0\{\Theta_i\}(\tau_-)$ & $1,3$\\
$q^3 = b$ & $\geq 0$ & $1,3$ & {\bf \small (RE3)}\\
$q^3-q^0 = b-\tau$ & $0$ & $2$ & {\bf \small (RE10)}\\
$q^0 = t_0$ & $\geq - \b{E}_i^0\{\Theta_i\}(t_0)$ & $2$ & $\tau < t_0 + b$\\
$\xi \in \p D_{r_b'(t,\u{u})}(\xi_0)$ & $\geq 0$ & all & {\bf \small (RE4)}
\end{tabular}
\end{minipage}
\begin{minipage}{3.7cm}
\begin{center}
\input{D2.pstex_t}
\end{center}
\end{minipage}
\vskip 2mm
\noindent The discussion literally transposes to $\p^{\alpha}\Theta_i$ and $j_i^{\mu}[\p^{\alpha}\Theta_i]$, for $|\alpha|\leq R$. Hence 
\begin{align}\label{eq:estimatebasic}
\b{E}_i^0\{\p^{\alpha}\Theta_i\}(\tau) - k_i(\tau) \b{E}_i^0\{\p^{\alpha}\Theta_i\}(\tau_-) & \leq \textstyle \int_{\mathcal{D}_i(\tau)} \dd^4\iv\;\p_{\mu}j^{\mu}_i[\p^{\alpha}\Theta_i]\\
\label{eq:estimate}
\b{E}_i^R\{\Theta_i\}(\tau) - k_i(\tau) \b{E}_i^R\{\Theta_i\}(\tau_-) & \leq \textstyle\int_{\mathcal{D}_i(\tau)} \dd^4\iv\; \textstyle\sum_{|\alpha|\leq R} \p_{\mu}j^{\mu}_i[\p^{\alpha}\Theta_i]
\end{align}
The first inequality implies the second, by summing over $|\alpha|\leq R$. By definition, $k_1(\tau) = k_3(\tau) \equiv 1$, whereas $k_2(\tau)=0$ if $\tau_- > t_0$ and $k_2(\tau)=1$ if $\tau_- = t_0$. 

The divergence identity in \eqref{eq:jcurrentXX}, with $\p^{\alpha} \Theta_i$ in the role of $\Theta_i$, and \eqref{eq:diffeqs} imply
\begin{equation*}
\p_{\mu}j^{\mu}_i[\p^{\alpha}\Theta_i] 
= 2(\p^{\alpha}\Theta_i)^T \big\{\textstyle\sum_{j=1}^3 \asympt{H}_{ij}\,\p^{\alpha}\Theta_j \; +\; \tfrac{1}{2}\,(\p_{\mu}M_i^{\mu})(\p^{\alpha}\Theta_i)  +  \p^{\alpha}\source_i + S_i^{\alpha}\big\}
\end{equation*}
For $i=1,2$, we directly estimate the right hand side of \eqref{eq:estimate}, by using Schwarz's inequality for the spatial part of the integral, and {\bf \small (RE12)}, \eqref{S:E} and \eqref{eq:comparable}. For $i=3$, we first exploit $\s{H}_3 \leq 0$ from {\bf \small (RE11)} to drop the term
$2(\p^{\alpha}\Theta_3)^T \asympt{H}_{33}(\p^{\alpha}\Theta_3)$, and then go on as before. We also use $|\p_{\mu}\b{M}^{\mu}| = |\p_{\mu}(\b{M}^{\mu}-\asympt{\b{M}}^{\mu})| \lesssim_R \smallSHS\,|t|^{-1} \leq \smallEEProof\,|t|^{-1}$, which holds when either {\bf \small (RE13a)} or {\bf \small (RE13b)} is assumed; in the first case use \eqref{eq:sobolev}. Abbreviating $\b{E}_i = \b{E}_i^R\{\Theta_i\}$ and $\b{E} = \b{E}_1 + \b{E}_2 + \b{E}_3$, we have for all $\tau \in \overline{\mathcal{I}}$:
\begin{equation}\label{ineq1}
\begin{aligned}
\b{E}_1(\tau) - \b{E}_1(\tau_-) & \lesssim_X \textstyle\int_{\tau_-}^{\tau} \dd t \, \sqrt{\b{E}_1(t)}\big(\smallEEProof\,\sqrt{\b{E}(t)} + \smallRho |t|^{-J}\big)|t|^{-1} \\
\b{E}_2(\tau) - \b{E}_2(t_0) & \lesssim_X \textstyle\int_{\tau_-}^{\tau}\dd t\, \sqrt{\b{E}_2(t)}\big( \sqrt{\b{E}_1(t)} + \smallEEProof\,\sqrt{\b{E}(t)} + \smallRho |t|^{-J}\big)\\
\b{E}_3(\tau) - \b{E}_3(\tau_-) & \textstyle\lesssim_X \int_{\tau_-}^{\tau} \dd t\, \sqrt{\b{E}_3(t)}\big( \sqrt{\b{E}_2(t)} + \smallEEProof\,\sqrt{\b{E}(t)} + \smallRho |t|^{-J}\big)|t|^{-1}
\end{aligned}
\end{equation}
where $X$ is defined as in the proposition.

\emph{Step 4:}  For each $A = (A_1,A_2,A_3) \in (0,\infty)^3$, define
$$\mathcal{J}(A) = \big\{t \in \overline{\mathcal{I}}\;\;\big|\;\;\textstyle\sup_{\tau \in [t_0,t]}\,|\tau|^{2J} \b{E}_i(\tau) \leq A_i^2,\quad i=1,2,3
\big\}$$
Recall $J \geq J_0 > 0$. Assume $A$ satisfies
\begin{equation}\label{ineq11}
\begin{aligned}
A_1 & > |t_0|^{J} \sqrt{\b{E}_1(t_0)} &\quad  A_1 & > CJ_0^{-1}\,\smallRho
& \quad A_1 & > CJ_0^{-1}\smallEEProof |A|\\
A_2 & > 2|t_0|^{J} \sqrt{\b{E}_2(t_0)} &\quad  A_2 & > 8C \smallRho
& A_2 & > 8C(A_1 + \smallEEProof |A|)\\
A_3 & > |t_0|^{J} \sqrt{\b{E}_3(t_0)}  &\quad  A_3 & > CJ_0^{-1}\smallRho & A_3 & > CJ_0^{-1}(A_2 + \smallEEProof |A|)
\end{aligned}
\end{equation}
where $|A|^2 = A_1^2 + A_2^2 + A_3^2$ and where $C = C(X) > 0$ is the maximum of the three constants of proportionality in the inequalities \eqref{ineq1}. By \eqref{ineq1}, \eqref{ineq11} and the continuity of $\overline{\mathcal{I}}\ni \tau \mapsto \b{E}_i(\tau)$, the set $\mathcal{J}(A)$ is an open and closed sub-interval of $\overline{\mathcal{I}}$ and contains $t_0$. Therefore, $\mathcal{J}(A) = \mathcal{\overline{I}}$. To see that $\mathcal{J}(A)$ is open in $\overline{\mathcal{I}}$, first observe that for every $\tau \in \mathcal{J}(A)$, the inequalities \eqref{ineq1}, \eqref{ineq11} imply the strict inequalities $\b{E}_i(\tau) < (A_i|\tau|^{-J})^2$, and then use continuity.

For each $\lambda \geq 0$, set
$A(\lambda) = \lambda \big(1,1+8C, 1 + CJ_0^{-1}(1+8C))$.
The three rightmost inequalities in \eqref{ineq11} are homogeneous in $A$, and hold for $A(\lambda)$, $\lambda > 0$, iff they hold for $A(1)$, which they do if $\smallEEProof > 0$ is sufficiently small depending on $X$, because it is true for $\smallEEProof = 0$. The definition of $\smallEEProof$ before \eqref{S:E} implies $\smallEEProof \leq (1 + |\asympt{H}_2|^2 + |\asympt{H}_3|^2)^{1/2} \smallEE(X)$. Consequently, the condition on $\smallEEProof$ holds if $\smallEE(X)$ is suitably small.
Set $$\lambda_0 =  2|t_0|^J \sqrt{\b{E}(t_0)} + \max\{8, J_0^{-1}\} C \smallRho \geq  0$$
If $\lambda >  \lambda_0$, then the remaining six inequalities in \eqref{ineq11} hold for $A(\lambda)$, and $\mathcal{J}(A(\lambda)) = \overline{\mathcal{I}}$. By the definition of $\mathcal{J}(A)$, we have $\mathcal{J}(A(\lambda_0)) = \overline{\mathcal{I}}$. By \eqref{eq:comparable}, inequality \eqref{eq:refprop:p2XX} follows if $\bigEE(X)$ is sufficiently big. \qed
\end{proof}

%% file: ree.pstex_t
\begin{picture}(0,0)%
\includegraphics{ree.pstex}%
\end{picture}%
\setlength{\unitlength}{2763sp}%
\begingroup\makeatletter\ifx\SetFigFont\undefined%
\gdef\SetFigFont#1#2#3#4#5{%
  \reset@font\fontsize{#1}{#2pt}%
  \fontfamily{#3}\fontseries{#4}\fontshape{#5}%
  \selectfont}%
\fi\endgroup%
\begin{picture}(3504,2499)(2359,-4873)
\put(5401,-3286){\makebox(0,0)[lb]{\smash{{\SetFigFont{8}{9.6}{\familydefault}{\mddefault}{\updefault}{\color[rgb]{0,0,0}$t = t^{\ast}$}%
}}}}
\put(2476,-2836){\rotatebox{45.0}{\makebox(0,0)[lb]{\smash{{\SetFigFont{8}{9.6}{\familydefault}{\mddefault}{\updefault}{\color[rgb]{0,0,0}$u = -1000$}%
}}}}}
\put(5401,-2986){\makebox(0,0)[lb]{\smash{{\SetFigFont{8}{9.6}{\familydefault}{\mddefault}{\updefault}{\color[rgb]{0,0,0}$t = -1000$}%
}}}}
\put(4126,-4636){\rotatebox{315.0}{\makebox(0,0)[lb]{\smash{{\SetFigFont{8}{9.6}{\familydefault}{\mddefault}{\updefault}{\color[rgb]{0,0,0}$\u{u} = 0$}%
}}}}}
\put(5401,-4561){\rotatebox{315.0}{\makebox(0,0)[lb]{\smash{{\SetFigFont{8}{9.6}{\familydefault}{\mddefault}{\updefault}{\color[rgb]{0,0,0}$\u{u} = b$}%
}}}}}
\put(5401,-4186){\makebox(0,0)[lb]{\smash{{\SetFigFont{8}{9.6}{\familydefault}{\mddefault}{\updefault}{\color[rgb]{0,0,0}$t = t_0$}%
}}}}
\end{picture}%

%% file: D2.pstex_t
\begin{picture}(0,0)%
\includegraphics{D2.pstex}%
\end{picture}%
\setlength{\unitlength}{2763sp}%
\begingroup\makeatletter\ifx\SetFigFont\undefined%
\gdef\SetFigFont#1#2#3#4#5{%
  \reset@font\fontsize{#1}{#2pt}%
  \fontfamily{#3}\fontseries{#4}\fontshape{#5}%
  \selectfont}%
\fi\endgroup%
\begin{picture}(1987,4053)(6436,-5680)
\put(6751,-5611){\makebox(0,0)[lb]{\smash{{\SetFigFont{8}{9.6}{\familydefault}{\mddefault}{\updefault}{\color[rgb]{0,0,0}$t_0$}%
}}}}
\put(6751,-2236){\makebox(0,0)[lb]{\smash{{\SetFigFont{8}{9.6}{\familydefault}{\mddefault}{\updefault}{\color[rgb]{0,0,0}$\tau_1$}%
}}}}
\put(6751,-5086){\makebox(0,0)[lb]{\smash{{\SetFigFont{8}{9.6}{\familydefault}{\mddefault}{\updefault}{\color[rgb]{0,0,0}$\tau_2$}%
}}}}
\put(8326,-4486){\rotatebox{90.0}{\makebox(0,0)[lb]{\smash{{\SetFigFont{8}{9.6}{\familydefault}{\mddefault}{\updefault}{\color[rgb]{0,0,0}$\iv^3 = b$}%
}}}}}
\put(6976,-4486){\rotatebox{90.0}{\makebox(0,0)[lb]{\smash{{\SetFigFont{8}{9.6}{\familydefault}{\mddefault}{\updefault}{\color[rgb]{0,0,0}$\iv^3 = 0$}%
}}}}}
\put(6451,-3511){\makebox(0,0)[lb]{\smash{{\SetFigFont{8}{9.6}{\familydefault}{\mddefault}{\updefault}{\color[rgb]{0,0,0}$\tau_1-b$}%
}}}}
\put(7201,-5311){\makebox(0,0)[lb]{\smash{{\SetFigFont{8}{9.6}{\familydefault}{\mddefault}{\updefault}{\color[rgb]{0,0,0}$D_2(\tau_2)$}%
}}}}
\put(7201,-2536){\makebox(0,0)[lb]{\smash{{\SetFigFont{8}{9.6}{\familydefault}{\mddefault}{\updefault}{\color[rgb]{0,0,0}$D_2(\tau_1)$}%
}}}}
\put(6751,-1786){\makebox(0,0)[lb]{\smash{{\SetFigFont{8}{9.6}{\familydefault}{\mddefault}{\updefault}{\color[rgb]{0,0,0}$t^{\ast}$}%
}}}}
\end{picture}%

%% file: SectionClassicalSolution.tex
\section{Classical Vacuum Fields} \label{sec:classicalvacuumfields}
We show in Proposition \ref{prop:MAINNEW} that Section \ref{sec:energyestimates} can be applied to  \eqref{shspsi1}, \eqref{shspsi2}. Then we prove the main Theorem \ref{mainenergy}. See Remark \ref{endof6}.
\begin{convention} \label{conv:ppawq}
Theorem \ref{mainenergy} is formulated in terms of the coordinates $x$. Otherwise $q$ is used, see Convention \ref{conv:notation}. Recall $u = \iv^0-\iv^3$.
For any function $f = f(x)$, the notation $f(q)$ stands for $f(x(q))$. For any vector field $v = v^{\mu}(x) (\p/\p x^{\mu})$, the notation $v^{\mu}(q)$ stands for $v^{\nu}(x(q)) ( \p q^{\mu} / \p x^{\nu})(x(q))$. 
\end{convention}
\begin{convention} \label{conv:cconjop}
In this section, $\C^m$ is a real vector space over $\R$ with dimension $2m$. A linear map $\C^m \to \C^n$ is, by convention, $\R$-linear. It is either a real $2n \times 2m$ matrix, or a complex $n\times m$ matrix which may have complex conjugation $C$ as matrix elements. Adopt similar conventions for the real subspaces
in Definitions \ref{def:paramparam} and \ref{def:kdhkueuuueze}:
$$\range \subset \C^5 \oplus \C^8 \oplus \C^5 \qquad \ranges \subset \C^5 \oplus \C^9 \oplus \C^3$$
\end{convention}
\vskip 2mm
To put \eqref{shspsi1}
and \eqref{shspsi2}  in the form required by Section \ref{sec:energyestimates}, we use:
\newcommand{\spspsp}{\rule{0pt}{13pt}}
\begin{list}{\labelitemi}{\leftmargin=1em}
\item[{ \bf \small (S1)}] $a,\anisotropic \in \R$ satisfy $0 < |\anisotropic| \leq |a| \leq 10^{-3}$.
\item[{ \bf \small (S2)}] \spspsp Let $[\,\Psi\,]$ be as in Proposition \ref{prop:fps} with $\text{\bf \footnotesize DATA}|_{q^3 < 1/2} \equiv 0$. Fix an integer $K \geq 0$. Set
$$\Psi_K = \textstyle\sum_{k=0}^{K+1} (\frac{1}{u})^k \Psi(k)$$
\item[{ \bf \small (S3)}]\spspsp Recall $\Psi = (\PsiOne,\PsiTwo,\PsiThree)$. Change fields to
$(\SpecPsiOne,\SpecPsiTwo,\SpecPsiThree) = \Psi - \Psi_K$ and then rearrange to
\begin{align*}
\Xi = (\Xi_1,\Xi_2,\Xi_3) &= \SpecPsiThree \oplus (\SpecPsiOne_1,\SpecPsiOne_2,\SpecPsiOne_4,\SpecPsiOne_5,\SpecPsiTwo_1,\SpecPsiTwo_2,\SpecPsiTwo_3,\SpecPsiTwo_4,\SpecPsiTwo_8) \oplus (\SpecPsiOne_3,\SpecPsiTwo_5,\SpecPsiTwo_6,\SpecPsiTwo_7)
\end{align*}
Let $\pi$ be the permutation matrix with
$(\SpecPsiOne,\SpecPsiTwo,\SpecPsiThree) = \pi  (\Xi_1,\Xi_2,\Xi_3)$.\\
The field $\Xi = (\Xi_1,\Xi_2,\Xi_3)$ takes values in $\pi^{-1}\range \,\subset\, \C^5 \oplus \C^9 \oplus \C^4$. 
\item[{ \bf \small (S4)}]\spspsp System \eqref{shspsi1}, $\b{A}(\iv,\Psi)\Psi = \b{f}(\iv,\Psi)$, is equivalent to
\begin{equation}\label{shsxi}
\b{B}(\iv,\Xi)\Xi = Q(\iv,\Xi)\Xi + \source
\end{equation}
where
\begin{align*}
\b{B}(\iv,\Xi) & = \pi^{-1} \b{A}(\iv,\Psi_{K} + \pi\Xi) \,\pi\\
Q(\iv,\Xi) \Pi & =  \pi^{-1} \textstyle\frac{\dd}{\dd s}\big|_{s=0}\int_0^1\dd s'\big(-\b{A}(\iv,s\pi\Pi)\Psi_K + \b{f}(\iv,\Psi_K + s'\pi \Xi + s\pi\Pi)\big)\\
\source(\iv) & = \pi^{-1}\,\big(\b{f}(\iv,\Psi_{K}) - \b{A}(\iv,\Psi_{K})\Psi_{K}\big)
\end{align*}
Here, $\pi^{-1}\range \to \pi^{-1}\range$, $\Pi \mapsto Q(\iv,\Xi)\Pi$ is $\R$-linear. The operator $\frac{\dd}{\dd s}|_{s=0}\int_0^1\dd s'$ acts on a quadratic polynomial in $s$, $s'$. 
\item[{\bf \small (S5)}]\spspsp
$\Pi \mapsto \b{B}^{\mu}(q,\Pi)$
and $\Pi \mapsto Q(q,\Pi)$ are affine $\R$-linear. Set
$$\overset{\hskip -6pt \bullet}{\b{B}^\mu}(\iv)\Pi = \tfrac{\dd}{\dd s}\big|_{s=0}\, \b{B}^\mu(\iv,s\Pi)\qquad
\overset{ \bullet}{Q}(\iv)\Pi = \tfrac{\dd}{\dd s}\big|_{s=0}\, Q(\iv,s\Pi)$$
\item[{\bf \small (S6)}]\spspsp The $\C^5\oplus \C^9\oplus \C^4$ block decomposition of $\b{B}$ is $\b{B} = \mathrm{diag}(B_1,B_2,B_3)$,\,  $B_2 = \mathbbm{1}_{9}L$,\, $B_3 = \mathbbm{1}_{4}N$, and $B_1$ is the $5\times 5$ Hermitian matrix operator on the left hand side of
\eqref{ri1:3}. The block decomposition of $Q$ is denoted $Q = (Q_{mn})_{m,n=1,2,3}$.
\item[{\bf \small (S7)}] \spspsp Let $\asympt{Q}(t)$ be the $\C^5\oplus \C^9\oplus \C^4$ block matrix
$$\asympt{Q}(t) = (\asympt{Q}_{mn})_{m,n=1,2,3} = \begin{pmatrix} 0 & 0 & 0 \\
\asympt{Q}_1 & 0 & 0 \\
0 & |t|^{-1}\asympt{Q}_2 & |t|^{-1}\asympt{Q}_3
\end{pmatrix}$$
$\asympt{Q}_1$, $\asympt{Q}_2$, $\asympt{Q}_3$ are the matrices whose only nonzero entries are ($C$ is complex conjugation):
$(\asympt{Q}_1)_{51} = (\asympt{Q}_1)_{72} = 
(\asympt{Q}_2)_{28} = 
(\asympt{Q}_3)_{22} = -1$,
$(\asympt{Q}_1)_{93} = 1$,
$(\asympt{Q}_2)_{19} = -1-C$,
$(\asympt{Q}_2)_{27} = C$,
$(\asympt{Q}_3)_{11} = -2$.
Observe that 
$\asympt{Q}_3$ is symmetric and $\asympt{Q}_3 \leq 0$.
\item[{\bf \small (S8)}] \spspsp $\asympt{\b{B}} = \asympt{\b{B}}^{\mu}\tfrac{\p}{\p q^{\mu}} = \mathrm{diag}(U,U,U,U,V) \oplus \mathbbm{1}_{9} V \oplus \mathbbm{1}_{4}U$
with $U = 
 \frac{\p}{\p q^0}$, $V 
= \frac{\p}{\p q^0} + \frac{\p }{\p q^3}$.
\item[{\bf \small (S9)}] \spspsp Let $\psi = \psi(\b{q}):\;\R^3 \to [0,1]$ be the smooth cutoff function
$$\psi(\b{q}) = \frac{s\big(3 - |\tfrac{\anisotropic}{a}\xi|_{\R^2}\big)}{s\big(3 -  |\tfrac{\anisotropic}{a}\xi|_{\R^2}\big) + s\big( |\tfrac{\anisotropic}{a}\xi|_{\R^2} - \tfrac{5}{2}\big)}
\;
\frac{s\big(\tfrac{3}{4} - |\iv^3-1|\big)}{s\big(\tfrac{3}{4} - |\iv^3-1|\big) + s\big(|\iv^3-1| - \tfrac{2}{3}\big)}$$
where $\xi = (\iv^1,\iv^2)$, and $s(x) = 0$ if $x \leq 0$ and $s(x) = e^{-1/x}$ when $x > 0$. Set
\begin{align*}
\mathcal{K} & = \overline{D_{3|\frac{a}{\anisotropic}|}(0) \times (\tfrac{1}{4},\tfrac{7}{4})} \qquad \subset \qquad
\mathcal{Q} = D_{4|\frac{a}{\anisotropic}|}(0)\times (0,2)
\end{align*}
By construction, $\supp_{\R^3} \psi \subset \mathcal{K}$, and $\psi$ is equal to $1$ on $D_{\frac{5}{2}|\frac{a}{\mathfrak A}|}(0)
\times (\tfrac{1}{3},\tfrac{5}{3})$.\\
By {\bf \small (S1)} we have
$\|\psi\|_{C^R(\R^3)}\lesssim_R 1$ for each integer $R\geq 0$.
\item[{\bf \small (S10)}] \spspsp
Set $\b{M}^{\mu}(\iv,\Xi) = \psi \b{B}^{\mu}(\iv,\Xi) + (1-\psi) \asympt{\b{B}}^{\mu}$
and $H(\iv,\Xi) = \psi Q(\iv,\Xi) + (1-\psi) \asympt{Q}(t)$
and $h(\iv,\Xi) = H(\iv,\Xi)\Xi + \psi\source(\iv)$.
\item[{\bf \small (S11)}]\spspsp Recall {\bf \small (S4)}. If $\Xi^{(1)}$, $\Xi^{(2)}$ are both smooth solutions to $\b{B}\Xi = Q \Xi + \source$, then $\Upsilon = \Xi^{(2)} - \Xi^{(1)}$ is a solution to
$\b{B}(\iv,\Xi^{(1)})\Upsilon = G \Upsilon$. By definition,
$$  G \Pi =
\tfrac{\dd}{\dd s}\big|_{s=0} \big( Q(\iv,\Xi^{(1)}) (s\Pi) - \b{B}(\iv,s\Pi) \Xi^{(2)} + Q(\iv,s\Pi)\Xi^{(2)}\big)
$$
The map $
\pi^{-1}\range
\to \pi^{-1}\range,\; \Pi \mapsto 
G\Pi = G\big(\iv,\Xi^{(1)},\Xi^{(2)},\p_{\iv}\Xi^{(2)}\big)\Pi$ is $\R$-linear.
\item[\secs{1}]\spspsp
Recall the constraint field $\Psi^{\sharp} = (\PsiSharpOne,\PsiSharpTwo,\PsiSharpThree)$. Rearrange to
$$
\Xi^{\sharp} = ({\Xi^{\sharp}}_1, {\Xi^{\sharp}}_2, {\Xi^{\sharp}}_3)  = \PsiSharpThree\oplus (\PsiSharpOne_1, \PsiSharpOne_2, \PsiSharpOne_4, \PsiSharpOne_5, \PsiSharpTwo_1, \PsiSharpTwo_2, \PsiSharpTwo_3, \PsiSharpTwo_4, \PsiSharpTwo_7, \PsiSharpTwo_8)\oplus (\PsiSharpOne_3, \PsiSharpTwo_5, \PsiSharpTwo_6, \PsiSharpTwo_9)
$$
with values in $\widehat{\pi}^{-1} \ranges  \subset  \C^3\oplus \C^{10}\oplus \C^4$. The permutation $\widehat{\pi}$ is such that  $\Psi^{\sharp}  =  \widehat{\pi} {\Xi^{\sharp}}$.
\item[\secs{2}]\spspsp
System \eqref{shspsi2}, $\widehat{\b{A}}(\iv,\Psi)\Psi^{\sharp} = \widehat{\b{f}}(\iv,\Psi,\p_\iv \Psi)\Psi^{\sharp}$, is equivalent to the linear, homogeneous symmetric hyperbolic system
$\widehat{\b{B}}(\iv,\Psi)\Xi^{\sharp} = \widehat{Q}(\iv,\Psi,\p_\iv \Psi)\Xi^{\sharp}$, with
\begin{align*}
\widehat{\b{B}}^{\mu}(\iv,\Psi) & = \widehat{\pi}^{-1}\widehat{\b{A}}^{\mu}(\iv,\Psi)\widehat{\pi} &
\widehat{Q}(\iv,\Psi,\p_\iv \Psi) & = \widehat{\pi}^{-1}\widehat{\b{f}}(\iv,\Psi,\p_\iv \Psi)\widehat{\pi}
\end{align*}
The map $\widehat{\pi}^{-1}\ranges \to \widehat{\pi}^{-1}\ranges,\;\Xi \mapsto \widehat{Q}(\iv,\Psi,\p_\iv \Psi)\Xi$ is $\R$-linear. Moreover,
$\widehat{\b{B}}^{\mu}$ depends affine $\R$-linearly on $\Psi$, and $\widehat{Q}$ depends affine $\R$-linearly on $\Psi \oplus \p_{\iv}\Psi$.
\item[\secs{3}]\spspsp The $\C^3\oplus \C^{10}\oplus \C^4$ block decomposition of $\widehat{\b{B}}$ is $\b{\widehat{B}} = \mathrm{diag}\big(\widehat{B}_1,\widehat{B}_2,\widehat{B}_3\big)$, $\widehat{B}_2 = \mathbbm{1}_{10}L$, $\widehat{B}_3 = \mathbbm{1}_{4}N$, and $\widehat{B}_1$ is the $3\times 3$ Hermitian matrix operator on the left hand side of \eqref{ri3:3}.
The block decomposition of $\widehat{Q}$ is denoted $\widehat{Q} = (\widehat{Q}_{mn})_{m,n=1,2,3}$.
\item[\secs{4}]\spspsp Let $\asympt{\widehat{Q}}(t) $ be the $\C^3 \oplus \C^{10}\oplus \C^4$ block matrix
$$\asympt{\widehat{Q}}(t)= (\asympt{\widehat{Q}}_{mn})_{m,n=1,2,3} = \begin{pmatrix} 0 & 0 & 0 \\
\asympt{\widehat{Q}}_1 & 0 & 0 \\
0 & |t|^{-1}\asympt{\widehat{Q}}_2 & |t|^{-1}\asympt{\widehat{Q}}_3\end{pmatrix}
$$
$\widehat{\asympt{Q}}_1$, $\widehat{\asympt{Q}}_2$, $\widehat{\asympt{Q}}_3$ are the  matrices whose only nonzero entries are
$(\asympt{\widehat{Q}}_1)_{5,1} = 1$,
$(\asympt{\widehat{Q}}_2)_{1,9} = (\asympt{\widehat{Q}}_2)_{4,10} = (\asympt{\widehat{Q}}_3)_{1,1} = -1$, 
$(\asympt{\widehat{Q}}_2)_{1,10} = (\asympt{\widehat{Q}}_2)_{2,8}  = (\asympt{\widehat{Q}}_2)_{3,7} = (\asympt{\widehat{Q}}_2)_{4,9}  = -C$. 
Observe that $\asympt{\widehat{Q}}_3$ is symmetric and $\asympt{\widehat{Q}}_3 \leq 0$.
\item[\secs{5}]\spspsp $\asympt{\b{\widehat{B}}} = \asympt{\b{\widehat{B}}}^{\mu}\tfrac{\p}{\p \iv^\mu} = \mathbbm{1}_{3}U\oplus \mathbbm{1}_{10}V\oplus \mathbbm{1}_{4}U$
with $U = 
 \frac{\p}{\p q^0}$ and $V 
= \frac{\p}{\p q^0} + \frac{\p }{\p q^3}$.
\end{list}
\begin{definition}\label{gs1} Each entry to the left of the vertical bar is a generic symbol for a polynomial (with complex coefficients) in the (components of the) quantities to the right and their complex conjugates.
$$
\begin{array}{c|l}
\mathcal{J} & \;u^{-1} \vspace{1mm}\\
\boldsymbol{\mathcal{J}} & \;1\quad \text{(That is, a generic symbol for a complex number)} \vspace{1mm}\\
\mathcal{H} & \; \text{linear over $\R$ in } \;\Psi(0) \vspace{1mm}\\
\boldsymbol{\mathcal{H}} & \;\text{linear over $\R$ in } \;\anisotropic,\; \b{e},\; \boldsymbol{\lambda} \vspace{1mm}\\
\mathcal{G}_K & \;u^{-1},\; \anisotropic,\; S,\; \b{e},\; \boldsymbol{\lambda},\; \Psi(k)_{k=0\ldots K+1},\; \text{and their first derivatives} \vspace{1mm}\\
\boldsymbol{\mathcal{G}}_K &\;u^{-1},\; \anisotropic,\; \u{u},\; S_K,\; \b{e},\; \boldsymbol{\lambda},\; \Psi(k)_{k=0\ldots K+1},\; \text{and their first derivatives} \\
& \;\text{It has no constant term as a polynomial in $\Psi(k)$ and its derivatives.}\vspace{1mm} \\
\mathcal{G}^{\sharp} & \; u^{-1},\; \anisotropic,\; S,\; \b{e},\;  \boldsymbol{\lambda},\; \Psi(0),\; \Psi - \Psi(0),\; \text{and first derivatives} \vspace{1mm}\\
\mathcal{G}^{\sharp}_0 & \; u^{-1},\; \anisotropic,\; S,\; \b{e},\; \boldsymbol{\lambda},\; \Psi(0),\; \Psi - \Psi(0) \vspace{1mm}\\
\mathcal{G}^{\sharp}_1 & \; \text{like $\mathcal{G}^{\sharp}$, but it has no constant term as a polynomial in $\Psi-\Psi(0)$, $\p_\iv\big(\Psi-\Psi(0)\big)$} \vspace{1mm}\\
\boldsymbol{\mathcal{G}}^{\sharp} & \; u^{-1},\; \anisotropic,\; \u{u}, \;S_0,\; \b{e},\; \boldsymbol{\lambda},\; \Psi(0),\; \text{and first derivatives}
\end{array}
$$
The symbols $\mathcal{G}_K$, $\boldsymbol{\mathcal{G}}_K$ represent polynomials
whose coefficients and degrees \emph{may depend} on $K$.
The remaining eight symbols represent polynomials whose coefficients and degrees are \emph{independent} of $K$.  The functions $S_K$ and $S_0$ are defined through $S = -\sum_{k=0}^{\ell}(\frac{1}{u})^k\,\anisotropic^{2(k+1)}\u{u}^{k+1} + \frac{1}{u^{\ell+1}}S_{\ell}$ for all $\ell \geq 0$, see \eqref{eq:SdefSdefSdef} and \eqref{eq:sexpan}. 
\end{definition}
\begin{proposition}\label{generik1}
\emph{Part 1.}
Let $\source = (\source_m)_{m=1,2,3}$ be the $\C^5\oplus \C^9\oplus \C^4$ decomposition.
\begin{align*} 
&\begin{alignedat}{4}
\b{B}^\mu(\iv,0) \;& = \;\asympt{\b{B}}^\mu &\quad &+ \quad \phantom{u^{-1} \mathcal{H} +}\;\,\, u^{-1} \boldsymbol{\mathcal{H}} &\quad &+ \quad u^{-2}\mathcal{G}_K\\
Q_{1n}(\iv,0)  \;& =\;  \asympt{Q}_{1n}(\iv)&&+ \quad u^{-1} \mathcal{H} + \,u^{-1} \boldsymbol{\mathcal{H}} & & + \quad u^{-2} \mathcal{G}_K\\
Q_{2n}(\iv,0)  \;& =\; \asympt{Q}_{2n}(\iv) &&+ \quad \phantom{u^{-1}} \mathcal{H} + \,\phantom{u^{-1}}\boldsymbol{\mathcal{H}} &&+ \quad u^{-1} \mathcal{G}_K\\
Q_{3n}(\iv,0) \;& =\; \asympt{Q}_{3n}(\iv)  &&+ \quad (|t|^{-1}+u^{-1})\boldsymbol{\mathcal{J}} && +\quad  u^{-2}\mathcal{G}_K
\end{alignedat}\\
\rule{0pt}{16pt}
&\begin{aligned}
\source_1(\iv) & = u^{-(K+2)}\boldsymbol{\mathcal{G}}_K &
\source_2(\iv) & =  u^{-(K+2)}\boldsymbol{\mathcal{G}}_K &
\source_3(\iv) & = u^{-(K+3)}\boldsymbol{\mathcal{G}}_K
\end{aligned}\\
\rule{0pt}{14pt}
 &\begin{aligned} 
\overset{\hskip -6pt \bullet}{\b{B}^{\mu}}(\iv) & = u^{-2}\mathcal{J} &
\overset{\hskip -10pt \bullet}{Q_{1n}}(\iv) & = u^{-1}\mathcal{J} &
\overset{\hskip -10pt \bullet}{Q_{2n}}(\iv) & = \mathcal{J} &
\overset{\hskip -10pt \bullet}{Q_{3n}}(\iv) & = u^{-2}\mathcal{J}
\end{aligned}
\intertext{\emph{Part 2a.}}
& \begin{alignedat}{4}
\b{\widehat{B}}^{\mu}(\iv,\Psi)  & \; =\;\; \asympt{\b{\widehat{B}}}^\mu && && + u^{-2}\mathcal{G}^{\sharp}_{0}\\
\widehat{Q}_{1n}(\iv,\Psi,\p_{\iv}\Psi) &\;=\;\;\asympt{\widehat{Q}}_{1n}(\iv)  && && + u^{-2}\mathcal{G}^{\sharp}\\
\widehat{Q}_{2n}(\iv,\Psi,\p_{\iv}\Psi) & \;= \;\;\asympt{\widehat{Q}}_{2n}(\iv) && + \hskip 20mm\boldsymbol{\mathcal{H}} && + u^{-1}\mathcal{G}^{\sharp}\\
\widehat{Q}_{3n}(\iv,\Psi,\p_{\iv}\Psi)  & \;=\;\; \asympt{\widehat{Q}}_{3n}(\iv) &\quad& + (|t|^{-1} + u^{-1})\boldsymbol{\mathcal{J}} &\quad & + u^{-2}\mathcal{G}^{\sharp}
\end{alignedat}
\intertext{\emph{Part 2b.}}
& 
\begin{aligned}
{\Xi^{\sharp}}_1,\; {\Xi^{\sharp}}_3,\; \PsiSharpOne_1,\; \PsiSharpOne_2,\; \PsiSharpTwo_1,\; \PsiSharpTwo_2,\; \PsiSharpTwo_3 & = u^{-1}\boldsymbol{\mathcal{G}}^{\sharp} + \phantom{u}\, \mathcal{G}_1^{\sharp}\\
\PsiSharpOne_4,\; \PsiSharpOne_5,\; \PsiSharpTwo_4,\; \PsiSharpTwo_7,\; \PsiSharpTwo_8 & = u^{-1}\boldsymbol{\mathcal{G}}^{\sharp} + u \,\mathcal{G}_1^{\sharp}
\end{aligned}
\intertext{\emph{Part 2c.} It is a consequence of \eqref{shspsi2} that}
& L \begin{pmatrix} \PsiSharpOne_4 \\ \PsiSharpOne_5 \\ \PsiSharpTwo_4 \\ \PsiSharpTwo_7 \\ \PsiSharpTwo_8
\end{pmatrix}
= \begin{pmatrix}
\phantom{i}\,\b{e}(\overline{\PsiSharpTwo}_4 - \PsiSharpTwo_5) + \phantom{i}\,\b{e}(\overline{\PsiSharpTwo}_6-\PsiSharpTwo_3)\\
i\,\b{e}(\overline{\PsiSharpTwo}_4 - \PsiSharpTwo_5) -i\,\b{e}(\overline{\PsiSharpTwo}_6-\PsiSharpTwo_3)\\
0\\
\phantom{-}\boldsymbol{\lambda}(\PsiSharpTwo_6 - \overline{\PsiSharpTwo}_3) - \overline{\boldsymbol{\lambda}}(\PsiSharpTwo_4 - \overline{\PsiSharpTwo}_5)\\
-\overline{\boldsymbol{\lambda}}(\overline{\PsiSharpTwo}_6 - \PsiSharpTwo_3) + \boldsymbol{\lambda}(\overline{\PsiSharpTwo}_4 - \PsiSharpTwo_5)
\end{pmatrix} + u^{-1}\,\mathcal{G}^{\sharp}\,\Xi^{\sharp}
\end{align*}
\end{proposition}
\begin{proof} \emph{Part 1.} By direct verification, using $\text{\bf\small (S4)}$, $\text{\bf\small (S6)}$, $\text{\bf\small (S7)}$, Proposition \ref{prop:ri1}, Remark \ref{rem:DNL}, Definition \ref{def:gspandps}. 
 Three representative cases are ($C$ is complex conjugation):
$$
\begin{array}{l l}
\text{entry $(4,5)$ of $\b{B}^1(\iv,0) - \asympt{\b{B}}^1$} & =-\tfrac{1}{u}\b{e} + \tfrac{1}{u^2}\mathcal{G}_K\\
\text{entry $(6,5)$ of $Q_{22}(\iv,0) -\asympt{Q}_{22}(\iv)$} & =-\PsiTwo_1(0)\, C - \overline{\PsiTwo_1(0)} + \tfrac{1}{u}\mathcal{G}_K\\
\text{entry $(1,1)$ of $Q_{33}(\iv,0) - \asympt{Q}_{33}(\iv)$}  & = (\tfrac{2}{u}+\tfrac{1}{u^2}\mathcal{G}_K) - (-\tfrac{2}{|t|})
\end{array}
$$
For $\source_1$, use $\mathbf{f}(q,0) = 0$ and the construction of $[\,\Psi\,]$.
The truncation at $K+1$ in {\bf \small (S2)} is used, in particular, to show that
 that the \emph{last} component of $\source_1$ is 
$u^{-(K+2)}\boldsymbol{\mathcal{G}}_K$.\\
\emph{Parts 2a, 2c.} By direct verification, using Proposition \ref{prop:ri3}, Remark \ref{rem:DNL}, Definition \ref{def:gspandps}.\\
\emph{Part 2b.} Recall Proposition \ref{prop:ri2}. Write $\Psi = \Psi(0) + (\Psi - \Psi(0))$. Consider
each field in Part 2b as a polynomial in $\Psi-\Psi(0)$ and $\p_{\iv}(\Psi-\Psi(0))$, with coefficients possibly depending on $\Psi(0)$, $\p_{\iv}\Psi(0)$. The constant term of this polynomial is $u^{-1}\boldsymbol{\mathcal{G}}^{\sharp}$, by $[\,\Psi^{\sharp}\,] = 0$ and Remark \ref{rem:grgrhz}.
 Everything else is of the form $\mathcal{G}_1^{\sharp}$ or $u\,\mathcal{G}_1^{\sharp}$, respectively. \qed
\end{proof}
\begin{proposition}[Main Technical Proposition] \label{prop:MAINNEW} 
Fix $K \geq 0$
and
$\text{\bf \small DATA}$ as in {\bf \small (S2)}. Recall {\bf \small (S1)} through {\bf \small (S11)} and \secs{1} through \secs{5}. Let $R \geq 4$ be an integer. Set
\begin{equation}\label{prop:MAINNEW:a0}
Y = \big(R,K,\|\text{\bf \small DATA}\|_{C^{R+2K+6}(\mathcal{Q})}\big)
\end{equation}
\vskip 2mm
\noindent Fix $\smallSHSprime \in (0,1)$ and $T < -1000$. There are constants $\smallRxx(R), \smallYxx(Y) \in (0,1)$, non-increasing in all their arguments, such that Parts 1, 2, 3 below hold whenever
\begin{equation} \label{prop:MAINNEW:a1}
|a|  \leq \smallRxx(R)\smallSHSprime \qquad
\|\text{\bf \small DATA}\|_{C^{R+4}(\mathcal{Q})} \leq \smallRxx(R)\smallSHSprime \qquad
|T|^{-1} \leq \smallYxx(Y) \smallSHSprime
\end{equation}
\vskip 2mm
\noindent
{\bf Part 1.}
All the assumptions of Proposition \ref{exexex} (also those of Part 2) hold for the tuple
$(T,31,\b{M}^{\mu},h,\asympt{\b{B}}^{\mu},\asympt{Q},\mathcal{K},\mathcal{Q})$
whose last six entries are defined in $\text{\bf \small (S7), (S8), (S9), (S10)}$.\\
\noindent {\bf Part 2.} If $t_0 < T$ and $\Xi: [t_0,t_0 + \epsilon)\times \R^3 \to \pi^{-1}\range$ ($\epsilon > 0$) is a $C^{\infty}$-solution to $\b{M}(\iv,\Xi)\,\Xi = h(\iv,\Xi)$ which vanishes identically at $t_0$, then 
\begin{equation}\label{brummi}
|t_0|^{K+1}\; \textstyle\sup_{\xi_0 \in \R^2}\; \textstyle\sqrt{E^R_{\mathcal{O}(\xi_0,2,t_0)}\{\Xi\}(t_0)} \;\;\leq\;\; (\smallYxx(Y))^{-1}
\end{equation}
\noindent {\bf Part 3}\emph{, Assumptions}. We distinguish three systems {\bf \small (Sys1)}, {\bf \small (Sys2)}, {\bf \small (Sys3)}. Part 3 of the proposition applies to each of these systems individually. First,
\begin{equation*}
\begin{array}{r r l l}
\text{\bf \small (Sys1):} & \qquad b & = 2  \qquad &\xi_0 \in \R^2
\vspace{2mm}
\\
\text{{\bf \small (Sys2), (Sys3):}}  & b & = 1  &\xi_0 \in D_{2|\frac{a}{\anisotropic}|}(0)
\vspace{2mm}
\end{array}
\end{equation*}
Second, fix $t_0 < T$ and let $\mathcal{V} = \bigcup_{t \in (t_0,T)}\{t\}\times \mathcal{O}(\xi_0,b,t) \subset \R^4$.
For {\bf \small (Sys1)}, {\bf \small (Sys3)} there is a field $\Xi$, for {\bf \small (Sys2)} there are two fields $\Xi^{(1)}$ and $\Xi^{(2)}$, with:
\begin{itemize}
\item[(i)] $\Xi,\Xi^{(1)}, \Xi^{(2)} \in C^p(\,\overline{\mathcal{V}},\pi^{-1}\range)$ with $p = \infty$ for {\bf \small (Sys1)}, {\bf \small (Sys3)} and $p = 1$ for {\bf \small (Sys2)}.
\item[(ii)] 
\begin{equation*}
\begin{cases}\;
\b{M}(\iv,\Xi)\Xi = h(\iv,\Xi)  &\qquad \text{for {\bf \small (Sys1)}}\\
\;\b{B}(\iv,\Xi^{(i)})\Xi^{(i)} = Q(\iv,\Xi^{(i)}) + \source(\iv) &\qquad  \text{for {\bf \small (Sys2)}}\\
\;\b{B}(\iv,\Xi)\Xi = Q(\iv,\Xi) + \source(\iv) &\qquad  \text{for {\bf \small (Sys3)}}
\end{cases}
\end{equation*}
\item[(iii)] $\Xi,\Xi^{(1)}, \Xi^{(2)}$ vanish when $\iv^3 < \tfrac{1}{2}$.
\item[(iv)] For all $t \in (t_0,T)$:
$$
\left\{
\begin{aligned}
E^R_{\mathcal{O}(\xi_0,2,t)}\{\Xi\}(t) &\, \leq \,(\smallRxx(R) \smallSHSprime)^2 & \qquad &\text{for {\bf \small (Sys1)}}\\
\SUP^{(1)}_{\mathcal{O}(\xi_0,1,t)}\{\Xi^{(i)}\}(t) & \,\leq\, \smallRxx(R) \smallSHSprime & &  \text{for {\bf \small (Sys2)}}\\
\SUP^{(1)}_{\mathcal{O}(\xi_0,1,t)}\{\Xi\}(t) & \,\leq\, \smallRxx(R) \smallSHSprime & &  \text{for {\bf \small (Sys3)}}
\end{aligned}\right.
$$
\end{itemize}
\noindent \emph{Part 3, Conclusion 1:} $\Xi(\mathcal{V}),\, \Xi^{(1)}(\mathcal{V}),\, \Xi^{(2)}(\mathcal{V}) \,\subset\,  \overline{B_{1/2}(0)}\, \subset\, \pi^{-1} \range\,\cong\, \R^{31}$.
\vskip 1mm
\noindent \emph{Part 3, Conclusion 2:}
The assumptions {\bf \small (RE0)} through {\bf \small (RE12)} and 
{\bf \small (RE13a)} or {\bf \small (RE13b)} in  Proposition \ref{prop:fdsdsfkjhdskhds} hold if the tuple in Proposition \ref{prop:fdsdsfkjhdskhds} is given by
$$\begin{array}{ 
c | c | c || c}
\text{{\bf \small (Sys1)}, uses {\bf \small (RE13a)}} & \text{{\bf \small (Sys2)}, uses {\bf \small (RE13b)}} & \text{{\bf \small (Sys3)}, uses {\bf \small (RE13b)}} & \text{see} \\
\hline
\hline
(t_0,T) & (t_0,T) & (t_0,T)\\
\xi_0 & \xi_0 & \xi_0 \\
2 & 1 & 1\\
(10,15,6) & (10,15,6) & (6,18,8) \\
R & 0 & 0 \\
\b{M}^{\mu}(\iv,\Xi) & \b{B}^{\mu}(\iv,\Xi^{(1)}) & \widehat{\b{B}}^{\mu}(\iv,\Psi) 
& \text{\bf \small (S10), {\bf \small (S4)}, \secs{2}}
\\
H(\iv,\Xi) & G\big(\iv,\Xi^{(1)},\Xi^{(2)},\p_\iv \Xi^{(2)}\big) & \widehat{Q}(\iv,\Psi, \p_{\iv}\Psi)
& \text{\bf \small (S10), {\bf \small (S11)}, \secs{2}}
\\
\psi\, \source & 0 & 0
& \text{\bf \small (S9), {\bf \small (S4)}}
\\
\Xi & \Upsilon
= \Xi^{(2)} - \Xi^{(1)}
& \Xi^{\sharp} & \text{\secs{1}}\\
\asympt{\b{B}}^{\mu} & \asympt{\b{B}}^{\mu} & \widehat{\asympt{\b{B}}}^{\mu}
& \text{\bf \small (S8), \secs{5}}
\\
\asympt{Q} & \asympt{Q} & \widehat{\asympt{Q}}
& \text{\bf \small (S7), \secs{4}}
\\
(\smallYxx(Y))^{-1}& 0 & 0 \\
\smallSHSprime  & \smallSHSprime & \smallSHSprime\\
K+1 & >0 & > 0
& \text{\bf \small (S2)}
\end{array}
$$
\emph{Part 3, Conclusion 3:} For {\bf \small (Sys3)}, if $t_0 + 1 < T$, then
\begin{equation*}
\textstyle\sup_{t \in (t_0+1,T)} |t| \SUP^{(0)}_{\mathcal{O}(\xi_0,1,t)}\{\Xi^{\sharp}\}(t) \lesssim_Y \big(1 + \sup_{t \in (t_0,T)} |t| \SUP^{(1)}_{\mathcal{O}(\xi_0,1,t)}\{\Xi\}(t)\big)
\end{equation*}
\end{proposition}
\begin{proof} We give a detailed sketch of the proof. The proof makes a finite number of smallness assumptions on $\smallRxx(R)$ and $\smallYxx(Y)$.\\
\emph{Overall Preliminaries.} For all $n \geq 0$ and $0 \leq k \leq K+1$ and $\beta \in \N_0^4$ with $|\beta|\leq 1+R$, the following estimates hold on $(-\infty,T)\times \mathcal{Q}$:
\begin{align*}
|\p^{\beta}u^{-n}| & \lesssim_{(R,n)} |t|^{-n} & |\p^{\beta}\anisotropic| & = |\anisotropic|\, \delta_{\beta 0} \leq |a| & |\p^{\beta}\u{u}| & \leq 2\\
|\p^{\beta}\Psi(0)| & \lesssim_R \smallRxx(R)\, \smallSHSprime  & |\p^{\beta}\b{e}| & \leq \tfrac{17}{2}\,|a| & |\p^{\beta}\boldsymbol{\lambda}| & \leq \tfrac{17}{2}\,|a| \\
|\p^{\beta}\Psi(k)| & \lesssim_Y 1  & |\p^{\beta}S| & \lesssim_R 1 & |\p^{\beta}S_K| & \lesssim_Y 1
\end{align*}
To estimate $\Psi(0)$, $\Psi(k)$ use \eqref{prop:MAINNEW:a0}, \eqref{prop:MAINNEW:a1}, Proposition \ref{prop:coeffest}. See Definition \ref{gs1} for $S_K$. Hence, by the product rule and \eqref{prop:MAINNEW:a1}, for all $\alpha \in \N_0^4$ with $|\alpha|\leq R$:
\begin{align*}
 n & = 0,1 & |t|^n\, |\p^{\alpha}(u^{-(n+1)}\mathcal{G}_K)| & \; \lesssim_Y\;  |T|^{-1}
\lesssim_Y \smallYxx(Y)\, \smallSHSprime\\
n & = 1,2 & |t|^{K+n}|\p^{\alpha}(u^{-(K+n)}\boldsymbol{\mathcal{G}}_K)| & \; \lesssim_Y\; 1\\
 n & = 0,1 & |t|^n\, |\p^{\alpha}(u^{-n}\mathcal{H})| & \; \lesssim_R\; \smallRxx(R)\, \smallSHSprime\\
 n & = 0,1 & |t|^n\, |\p^{\alpha}(u^{-n}\boldsymbol{\mathcal{H}})| & \; \lesssim_R \; |a| \lesssim_R \smallRxx(R)\, \smallSHSprime\\
 n & = 0,1,2 & |t|^n\, |\p^{\alpha}(u^{-n}\mathcal{J})| & \; \lesssim_R 1\\
 && |t|\, |\p^{\alpha}(|t|^{-1}+u^{-1})\boldsymbol{\mathcal{J}}| & \; \lesssim_R \; |T|^{-1} \lesssim_R \smallYxx(Y)\, \smallSHSprime
\end{align*}
at every point of $(-\infty,T)\times \mathcal{Q}$.
In this instance, the constants also depend on the particular polynomial represented by the generic symbols.\\
\noindent \emph{Preliminaries for Part 3.} Let $(\mathcal{V}_1,\mathcal{V}_2)$ be the open cover of $\mathcal{V}$ given by $\mathcal{V}_1 = \mathcal{V}\cap ((t_0,T) \times \mathcal{Q})$ and
$\mathcal{V}_2 = \mathcal{V}\cap ((t_0,T)\times (\R^3 \setminus \mathcal{K}))$.
The sets $\mathcal{Q}$, $\mathcal{K}$ are defined in {\bf \small (S9)}.  
\begin{list}{\labelitemi}{\leftmargin=1em}
\item {\bf \small (Sys1)}: The Overall Preliminaries apply to $\mathcal{V}_1$. On $\mathcal{V}_2$
the equations simplify because $\psi = 0$. The estimate $\|\psi\|_{C^R(\R^3)}\lesssim_R 1$ in {\bf \small (S9)} is used on the transition region.
\item {\bf \small (Sys2)}: $\mathcal{V} = \mathcal{V}_1$. The Overall Preliminaries  suffice.
\item {\bf \small (Sys3)}: $\mathcal{V} = \mathcal{V}_1$. Supplement Overall Preliminaries by:
\begin{align*}
 |\Psi - \Psi(0)|\;\;,\;\;\big|\p_{\iv}\big(\Psi - \Psi(0)\big)\big|& \;  \lesssim_Y \;1\\
 |t|\, |u^{-2}\mathcal{G}^{\sharp}_{0}|\;\;,\;\; |t|\, |\p_{\iv}(u^{-2}\mathcal{G}^{\sharp}_{0})| & \; \lesssim_Y \;|T|^{-1} \lesssim_Y \smallYxx(Y)\, \smallSHSprime\\
 n  = 0,1:\qquad \qquad |t|^n\, |u^{-(n+1)}\mathcal{G}^{\sharp}| & \;\lesssim_Y \;|T|^{-1} \lesssim_Y \smallYxx(Y)\, \smallSHSprime\\
  |\boldsymbol{\mathcal{H}}| & \;\lesssim \;\;|a|\lesssim \smallRxx(R)\, \smallSHSprime\\
 |t|\, |(|t|^{-1}+u^{-1})\boldsymbol{\mathcal{J}}|& \;\lesssim_Y \;|T|^{-1} \lesssim_Y \smallYxx(Y)\, \smallSHSprime
\end{align*}
on $\mathcal{V}$. For the first, use
$\Psi - \Psi(0) = \sum_{k=1}^{K+1} (\tfrac{1}{u})^k\Psi(k) + \pi\Xi$
and (iv) in the proposition. The rest are consequences of the Overall Preliminaries.
\end{list}
\noindent \emph{Proof of Part 1.} If $\smallRxx(R)$, $\smallYxx(Y)$ are small enough, then $\tfrac{1}{2}\leq \b{M}^0(\iv,\Xi) \leq 2$ and $\b{M}^3(\iv,\Xi) \geq 0$ for all $(\iv,\Xi) \in ((-\infty,T)\times \R^3) \times B_2(0)$, with $B_2(0) \subset \R^{31}$. In fact, $\mathbf{M}^{\mu}$ is a convex combination of $\mathbf{B}^{\mu}$ and $\asympt{\mathbf{B}}^{\mu}$, and to estimate $\mathbf{B}^0$ and $\mathbf{B}^3$, use $|\Xi| \leq 2$ and the estimates for $\Psi(0)$ and $\Psi(k)$ in the Overall Preliminaries. To check the assumptions of
 Proposition \ref{exexex} Part 2, use $\Psi_K|_{q^3 < 1/2} \equiv 0$, see {\bf \small (S2)}.\\
\emph{Proof of Part 2.} One must estimate $\p^{\alpha}\Xi(t_0,\,\cdot\,)$ with $\alpha\in \N_0^4$, $|\alpha| \leq R$.  They are determined by $\b{M}(\iv,\Xi)\Xi = h(\iv,\Xi)$ and $\Xi(t_0,\,\cdot\,) \equiv 0$. 
They vanish on $\R^3\setminus \mathcal{K}$, by the support of $\psi$. To obtain decay, use the results about $\source$ in Proposition \ref{generik1} Part 1.\\
\noindent \emph{Proof of Part 3, Conclusion 1.} By (iv) in the proposition and by making $\smallRxx(R)$ small enough. For {\bf \small (Sys1)}, we also use $R \geq 2$ and the Sobolev inequality \eqref{eq:sobolev}.\\
\noindent \emph{Proof of Part 3, Conclusion 2.}
Most of {\bf \small (RE0)} to {\bf \small (RE12)} and {\bf \small (RE13a)} or {\bf \small (RE13b)}
in Proposition \ref{prop:fdsdsfkjhdskhds} are verified directly.
To check {\bf \small (RE4)}, use Proposition \ref{prop:uzewuztuztwqe873}. 
To check {\bf \small (RE5)} for {\bf \small (Sys3)}, recall that $\Psi^{\sharp}$ solves \eqref{shspsi2} because $\Psi$ solves \eqref{shspsi1}.
We now discuss two estimates that are quite representative for the proof: the first inequality in {\bf \small (RE12)} for {\bf \small (Sys1)}, and the second inequality in {\bf \small (RE13a)} for {\bf \small (Sys1)}. Recall Proposition \ref{generik1} Part 1 and observe that $$E^R_{\mathcal{O}(\xi_0,b,t)}\{f\}(t) \lesssim_R \big(\SUP^{(R)}_{\mathcal{O}(\xi_0,b,t)}\{f\}(t)\big)^2$$
Let $t \in (t_0,T)$. The  {\bf \small (RE12)} estimate:
$$
|t|^{2K+4}E^R_{\mathcal{O}(\xi_0,2,t)}\{\psi \,\source_1\}(t) = |t|^{2K+4}E^R_{\mathcal{O}(\xi_0,2,t)}\{\psi\, u^{-(K+2)}\boldsymbol{\mathcal{G}}_K\}(t)
\lesssim_Y \, 1
$$
The left hand side is $\leq (\smallYxx(Y))^{-2}$ if $\smallYxx(Y)> 0$ is small enough. 
The {\bf \small (RE13a)} estimate:
\begin{align*}
& |t|^2\, E^R_{\mathcal{O}(\xi_0,2,t)}\big\{H_{1n}(\iv,\Xi)-\asympt{Q}_{1n}\big\}(t)\\
& = |t|^2\, E^R_{\mathcal{O}(\xi_0,2,t)}\big\{\psi \big(Q_{1n}(\iv,0) -\asympt{Q}_{1n}\big) + \psi\, \overset{\bullet}{Q}_{1n}(\iv)\Xi \big\}(t)\\
& \lesssim_R |t|^2\, E^R_{\mathcal{O}(\xi_0,2,t)}\big\{\psi\big(\tfrac{1}{u}\mathcal{H} + \tfrac{1}{u}\boldsymbol{\mathcal{H}} + \tfrac{1}{u}\mathcal{J}\, \Xi\big) \big\}(t)
+ |t|^2\, E^R_{\mathcal{O}(\xi_0,2,t)}\big\{\psi\,\tfrac{1}{u^2}\mathcal{G}_K\big\}(t)
\end{align*}
The first term is $\lesssim_R (\smallRxx(R)\smallSHSprime)^2$, the second is $\lesssim_Y (\smallYxx(Y)\smallSHSprime)^2$. Here, we use condition (iv) in the proposition. The result is $\leq (\smallSHSprime)^2$ if $\smallRxx(R)$ and $\smallYxx(Y)$ are small enough.\\
\noindent \emph{Proof of Part 3, Conclusion 3.}
Let $\kappa$ be the right hand side of the inequality in Conclusion 3. 
Using Proposition \ref{generik1} Part 2b, one shows that on $\mathcal{V} = \mathcal{V}_1$ the estimates $|{\Xi^{\sharp}}_1|, |{\Xi^{\sharp}}_3|, |\PsiSharpOne_1|, |\PsiSharpOne_2|, |\PsiSharpTwo_1|, |\PsiSharpTwo_2|, |\PsiSharpTwo_3| \hskip-1pt\lesssim_Y\hskip-1pt \kappa |t|^{-1}$ and
$|\PsiSharpOne_4|, |\PsiSharpOne_5|, |\PsiSharpTwo_4|, |\PsiSharpTwo_7|, |\PsiSharpTwo_8| \hskip-1pt\lesssim_Y \hskip-1pt \kappa$ 
hold. The last inequality can be improved to $\lesssim_Y \kappa|t|^{-1}$ for $t \in (t_0+1,T)$ by integrating the equations in Proposition \ref{generik1} Part 2c along the vector field $L$, using $\Xi^{\sharp}|_{q^3 < 1/2} \equiv 0$. \qed
\end{proof}
\begin{theorem} \label{mainenergy} Let $(\xi,\u{u},u)$ be Cartesian coordinates on the truncated unit strip
$$\STRIP(1,\lambda)
\,=\,   \R^2 \times (0,1) \times (-\infty\,,\,-\lambda^{-1}\,)
\qquad 
\lambda > 0$$
Suppose $0 < |\mathfrak A| \leq |a|$. Assume the functions
$\text{\bf \small DATA}^{\sigma}(\xi,\u{u}):\, \R^2 \times (0,\infty) \to \C$
$\sigma \in \{-,+\}$, are $C^{\infty}$, vanish when $\u{u} < \tfrac{1}{2}$, and are $\FLIP$-compatible
\begin{equation}\label{mainenergyflip}
\text{\bf \small DATA}^{\sigma} = \FLIP_{\frac{a}{\mathfrak A}}\cdot \text{\bf \small DATA}^{-\sigma}\qquad \xi \neq 0
\end{equation}
see Definitions \ref{def:poleflip}, \ref{def:flipdata}.
Let $[\,\Psi^{\sigma}\,]$ be the formal power series solution corresponding to $\text{\bf \small DATA}^{\sigma}$ as in Proposition \ref{prop:fps}.
Fix an integer $R \geq 4$ and $\epsilon \in (0,\tfrac{1}{2})$.
Let the constant $\smallFinal = \smallFinal(B) \in (0,1)$ be sufficiently small depending only on $B= (R,\epsilon)$. Suppose
\begin{equation}\label{mainenergysmallness}
0 < |\anisotropic|\leq |a| \leq \smallFinal \qquad
\textstyle\max_{\sigma \in \{-,+\}}\, \|\text{\bf \small DATA}^{\sigma}\|_{C^{R+4}(\mathcal{C}(a,\mathfrak A,2))} \,\leq\, \smallFinal
\end{equation}
Here,  $\mathcal{C}(a,\mathfrak A,b) = D_{4|\frac{a}{\anisotropic}|}(0)\times (0,b)$ for each $b > 0$.
Fix an integer $K \geq 0$ and set 
$$C = \big(R,\; \epsilon,\; K,\; \textstyle\max_{\sigma \in \{-,+\}}\, \|\text{\bf \small DATA}^{\sigma}\|_{C^{R+2K+6}(\mathcal{C}(a,\,\mathfrak A,\,2))}\big)$$
Let the constant $\smallFinalPrime = \smallFinalPrime(C) \in (0,1)$ be sufficiently small depending only on $C$. Then:
{\bf Existence:}\\
\emph{Part 1:} There exists a pair of fields
$\Psi^-,\Psi^+ \in C^1(\,\overline{\STRIP(1,\smallFinalPrime)},\mathcal{R})$
with
$\Psi^{\sigma} = \FLIP_{\frac{a}{\mathfrak A}}\cdot \Psi^{-\sigma}$ ($\xi \neq 0$)
which are both solutions to \eqref{shspsi1}, vanish when $\u{u} < \tfrac{1}{2}$, and satisfy
\begin{equation} \label{zztrt1}
\textstyle\lim_{u\to -\infty}\, |u|^{\epsilon}\; \textstyle\sup_{\alpha \in \N_0^4:\; 
|\alpha|\leq 1}\;
\big\|\,\p^{\alpha}\big(\Psi^{\sigma} - \Psi^{\sigma}(0)\big)(\,\cdot\,,u)\,\big\|_{C^0(\mathcal{C}(a,\mathfrak A,1))} = 0
\end{equation}
\emph{Part 2:} The constraint fields $(\Psi^-)^{\sharp}$, $(\Psi^+)^{\sharp}$ associated to the fields in Part 1 vanish, and
$$(\Phi^-,\Phi^+) = \big(\Minkowski_{a,{\mathfrak A}} + u^{-M}\Psi^-,\;\Minkowski_{a,{\mathfrak A}} + u^{-M}\Psi^+\big)$$
are a pair of vacuum fields with $\Phi^{\sigma} = \FLIP_{\frac{a}{\mathfrak A}}\cdot \Phi^{-\sigma}$ ($\xi \neq 0$)
, see Definitions \ref{def:vacvacvac}, \ref{def:poleflip}.\\
\emph{Part 3:} The fields in Part 1 are actually in $C^{R-3}(\,\overline{\STRIP(1,\smallFinalPrime)},\mathcal{R})$. Moreover
\begin{equation} \label{zztrt2}
\sup_{u < -\smallFinalPrime^{-1}} |u|^{K+1}\,\sup_{\substack{\alpha \in \N_0^4  \\ |\alpha| \leq R-3}}\Big\|\p^{\alpha}\Big(\Psi^{\sigma}(\,\cdot\,,u) - \sum_{k=0}^{K}\frac{\Psi^{\sigma}(k)(\,\cdot\,)}{u^k}\Big)\Big\|_{C^0{\textstyle{(}}\mathcal{C}(a,\mathfrak A,1){\textstyle{)}}} \;\leq \;\frac{1}{\smallFinalPrime}
\end{equation}
{\bf Uniqueness:} Assume $(\Psi^-,\Psi^+)$ and $(\widetilde{\Psi}^-,\widetilde{\Psi}^+)$ have all the properties listed in Part 1. Then they coincide on $\STRIP(1,\frac{\smallFinalPrime}{1+\smallFinalPrime}) \subset \STRIP(1,\smallFinalPrime)$.
\end{theorem}
\begin{proof} Theorem  \ref{mainenergy} is formulated in terms of $x = (\xi,\u{u},u)$. The proof uses $\iv = (t,\xi,\u{u})$. Recall $t = u + \u{u}$. We use {\bf \small (S1)} to {\bf \small (S11)} and \secs{1} to \secs{5}, where {\bf \small DATA} in {\bf \small (S2)} is identified with either one of $\text{\bf \small DATA}^{\sigma}$, $\sigma = -,\, +$. The smallness condition $\smallFinal < 10^{-3}$ ensures that $a$, ${\mathfrak A}$ satisfy {\bf \small (S1)}.\\
In this entire proof, $\smallEE(\,\cdot\,)$ and $\bigEE(\,\cdot\,)$
are as in Proposition \ref{prop:fdsdsfkjhdskhds}, and $\smallRxx(\,\cdot\,)$ and $\smallYxx(\,\cdot\,)$ are as in Proposition \ref{prop:MAINNEW}. Furthermore, $\bigSobolev(R) > 1$ is the constant of proportionality in the Sobolev inequality \eqref{eq:sobolev}, for all $(\xi_0,b,t) \in \R^2 \times [1,2] \times (-\infty,-1000)$.
\vskip 1mm

\noindent \emph{The smallness condition on $\smallFinal$.} Set $J_0 = \epsilon$, with $\epsilon$ as in Theorem \ref{mainenergy}, and
\begin{align*}
X & = (R,J_0,|\asympt{Q}_m|) &
X^{\ast} & = (0,J_0,|\asympt{Q}_m|) &
\widehat{X} & = (0,J_0,|\widehat{\asympt{Q}}_m|)
\end{align*}
with $m=1,2,3$. See {\bf \small (S7)}, \secs{4}. Set $\smallSHSprime(B) = \tfrac{1}{2}\min\{\smallEE(X),\smallEE(X^{\ast}), \smallEE(\widehat{X})\} \in (0,1)$.
The single smallness condition for $\smallFinal$ in this proof is
$\smallFinal < \min \{10^{-3}, \smallRxx(R) \smallSHSprime(B)\}$.
Every application of Proposition \ref{prop:MAINNEW} in this proof uses $\smallSHSprime = \smallSHSprime(B)$.
\vskip 1mm

\noindent \emph{The first smallness condition on $\smallFinalPrime$.} Set
\begin{align}
\notag Y & = \big(R,K,\textstyle\max_{\sigma \in \{-,+\}}\; \|\text{\bf \small DATA}^\sigma\|_{C^{R+2K+6}(\mathcal{Q})}\big)\\
\label{Tchoice} \b{T}(C) & = -1 -\; \max \bigg\{\,1000,\;\; \frac{1}{\smallYxx(Y)\smallSHSprime(B)},\;\;  \bigg(\frac{4\, \bigSobolev(R)\,(\bigEE(X)+1) }{\smallSHSprime(B) \smallRxx(R) \smallYxx(Y)}\bigg)^{\frac{1}{K+1}}\bigg\}
\end{align}
Observe that $\mathcal{C}(a,\mathfrak A,2) = \mathcal{Q}$ where $\mathcal{Q}$ is defined in {\bf \small (S9)}. Thus, $Y$ depends only on $C$, and so does the right hand side of \eqref{Tchoice}. We impose
$\smallFinalPrime < 1/(|\b{T}(C)| + 2)$.
There will be one more smallness condition on $\smallFinalPrime$, later in the proof.
\vskip 1mm

\noindent \emph{Notation:} System \eqref{shsxi} corresponding to $\text{\bf \small DATA}^{\sigma}$ will be denoted $\text{\eqref{shsxi}}^{\sigma}$. We sometimes suppress the superscript and write {\bf \small DATA}. We abbreviate $\FLIP_{\frac{a}{\mathfrak A}}$ by $\FLIP$. We have  $$\FLIP\cdot (\Minkowski_{a,\mathfrak A} + u^{-M}\Psi) = \Minkowski_{a,\mathfrak A} + u^{-M}\FLIP\cdot \Psi
\quad \text{and}\quad \FLIP\cdot \Psi_K^{\sigma} = \Psi_K^{-\sigma}$$
Therefore, if $\Xi$ solves $\text{\eqref{shsxi}}^{\sigma}$, then $\pi^{-1} \FLIP\cdot(\pi \Xi)$ solves $\text{\eqref{shsxi}}^{-\sigma}$.
The permutation $\pi$ is defined in $\text{\bf \small (S3)}$. We abbreviate $\pi^{-1} \FLIP\cdot(\pi \Xi)$ by $\FLIP\cdot \Xi$.
\vskip 1mm

Step 1 through Step 10 below build on each other:
\vskip 1mm

\noindent \emph{Step 1.} The assumptions of Proposition \ref{prop:MAINNEW} up to and including \eqref{prop:MAINNEW:a1} are satisfied for all $T \leq \b{T}(C)$,  if $\text{\bf \small DATA} = \text{\bf \small DATA}^{\sigma}$ for $\sigma = -,\, +$.
\vskip 1mm

\noindent \emph{Step 2.} For all $t_0 < \b{T}(C)$ there is
 $\Xi_{t_0} \in C^{\infty}([t_0,t_1(t_0))\times \R^3,\pi^{-1}\range)$
 with $t_1(t_0) \in (t_0,\b{T}(C)]$ and $\b{M}(\iv,\Xi_{t_0})\Xi_{t_0} = h(\iv,\Xi_{t_0})$ and $\Xi_{t_0}(t_0,\,\cdot\,) \equiv 0$
and $\Xi_{t_0}|_{\iv^3 < 1/2} \equiv 0$,
so that $t_1(t_0) \neq \b{T}(C)$ implies $\text{({\bf\small Break})}_1$ or $\text{({\bf\small Break})}_2$ in Proposition \ref{exexex}.
\vskip 1mm

\noindent \emph{Step 3.} $t_1(t_0) = \b{T}(C)$, for all $t_0 < \b{T}(C)$, and
for all $(\tau,\xi_0) \in [t_0,\b{T}(C)) \times \R^2$:\\
\begin{equation}\label{eq:fkhjkhhgkhkfhgf}
\sqrt{E^R_{\mathcal{O}(\xi_0,2,\tau)}\{\Xi_{t_0}\}(\tau)} \leq  \frac{2\bigEE(X)}{\smallYxx(Y) |\tau|^{K+1}}
 \leq \frac{2 \bigEE(X)}{\smallYxx(Y) |\b{T}(C)|^{K+1}}
 \leq   \frac{\smallRxx(R) \smallSHSprime(B) }{2\bigSobolev(R)}
\end{equation}
For steps 4 and 5, introduce a new field, the restriction of $\Xi_{t_0}$ to $[t_0,\b{T}(C))\times \mathcal{W}$, with $$\mathcal{W} = D_{\frac{5}{2}|\frac{a}{\anisotropic}|}(0) \times (0,1)\;\subset\; \R^3$$
Henceforth, $\Xi_{t_0}$ denotes this new field. We have $\Xi_{t_0} \in C^{\infty}([t_0,\b{T}(C))\times \overline{\mathcal{W}},\pi^{-1}\range)$.
\vskip 1mm

\noindent \emph{Step 4.}  $\Xi_{t_0}$ is a solution to $\text{\eqref{shsxi}}$, for all $t_0 < \b{T}(C)$.
\vskip 1mm

\noindent \emph{Step 5.} For each $t_0 < \b{T}(C)$ and $\sigma \in \{-,+\}$, let $\Xi^{\sigma}_{t_0}$ be the solution to $\text{\eqref{shsxi}}^{\sigma}$. Let
\begin{align*}
\mathcal{Y}(t_0) & = [t_0,\b{T}(C)) \times \big\{\tfrac{2}{5}|\tfrac{a}{\anisotropic}| < |\xi| < \tfrac{5}{2}|\tfrac{a}{\anisotropic}|\big\}\times (0,1)\\
\mathcal{Z}(t_0) & = \textstyle\bigcup_{\tau \in [t_0,\b{T}(C))} \;\;\;\bigcup_{|\frac{a}{\anisotropic}|\leq |\xi_0| < 2|\frac{a}{\anisotropic}|} \;\;\{\tau\}\times \mathcal{O}(\xi_0,1,\tau)\;\;\subset \;\;\mathcal{Y}(t_0)
\end{align*}
Then $\Xi_{t_0}^\sigma$ and $\FLIP\cdot \Xi_{t_0}^{-\sigma}$ are both defined on $\mathcal{Y}(t_0)$, and
coincide on $\mathcal{Z}(t_0) \subset \mathcal{Y}(t_0)$.
\vskip 1mm

\noindent For the remaining steps, introduce a new field in $C^{\infty}([t_0,\b{T}(C)) \times \R^2 \times [0,1],\pi^{-1}\range)$: \begin{align} \label{newfield} \iv & \mapsto \begin{cases}
\Xi_{t_0}^{\sigma}(\iv) & \text {if $|\xi| < 2|\frac{a}{\anisotropic}|$}\\
\FLIP\cdot \Xi_{t_0}^{-\sigma}(\iv) & \text{if $|\xi| > \frac{1}{2}|\frac{a}{\anisotropic}|$}\\
\end{cases} \end{align}
It is well defined by Step 5. Henceforth, $\Xi_{t_0}^{\sigma}$ denotes this new field.
\vskip 1mm

\noindent \emph{Step 6.} For each $t_0 < \b{T}(C)$ and $\sigma \in \{-,+\}$, we have
$\Xi_{t_0}^{\sigma} = \FLIP\cdot \Xi_{t_0}^{-\sigma}$ (when $\xi \neq 0$) and $\Xi_{t_0}^{\sigma}|_{\iv^3 < 1/2} \equiv 0$.
Moreover, $\Xi_{t_0}^{\sigma}$ solves $\text{\eqref{shsxi}}^{\sigma}$ on its entire domain of definition, and
\begin{subequations}
\begin{align*}
\textstyle\sup_{|\xi_0| < 2|\frac{a}{\anisotropic}|}\;\sup_{\tau \in (t_0,\b{T}(C))} \SUP^{(R-2)}_{\mathcal{O}(\xi_0,1,\tau)}\{\Xi_{t_0}^{\sigma}\}(\tau) & \leq \tfrac{1}{2}\, \smallRxx(R)\smallSHSprime(B) \\
\textstyle \sup_{|\xi_0| < 2|\frac{a}{\anisotropic}|}\;\sup_{\tau \in (t_0,\b{T}(C))} |\tau|^{K+1}\,\SUP^{(R-2)}_{\mathcal{O}(\xi_0,1,\tau)}\{\Xi_{t_0}^{\sigma}\}(\tau) & \leq \frac{2\bigEE(X) \bigSobolev(R)}{\smallYxx(Y)}
\end{align*}
\end{subequations}
\noindent \emph{Step 7.} For each $t_0 < \b{T}(C)-1$ and $\sigma \in \{-,+\}$,
\begin{equation*}
\textstyle\sup_{|\xi_0| < 2|\frac{a}{\anisotropic}|}\;\sup_{\tau \in (t_0+1,\b{T}(C))}\; E^0_{\mathcal{O}(\xi_0,1,\tau)}\big\{\big(\Xi_{t_0}^{\sigma}\big)^{\sharp}\big\}(\tau) \; \lesssim_{(Y,J_0)}\; |t_0|^{-1}
\end{equation*}
\noindent \emph{Step 8.} For all $t_1 \leq t_2 < \b{T}(C)$ and $\sigma \in \{-,+\}$,
$$\textstyle\sup_{|\xi_0| < 2|\frac{a}{\anisotropic}|}\;\sup_{\tau \in (t_2,\b{T}(C))} 
 E^0_{\mathcal{O}(\xi_0,1,\tau)}\{\Xi_{t_2}^{\sigma} - \Xi_{t_1}^{\sigma}\}(\tau) \; \lesssim_{(Y,J_0)}\; |t_2|^{-1}$$
\noindent \emph{Step 9.} There is a pair of solutions $\Xi^{\sigma} \in C^{R-3}((-\infty,\b{T}(C))\times \R^2 \times [0,1], \pi^{-1}\range)$ to $\text{\eqref{shsxi}}^{\sigma}$ and $(\Xi^{\sigma})^{\sharp} \equiv 0$, with $\Xi^{\sigma} = \FLIP\cdot \Xi^{-\sigma}$ (when $\xi \neq 0$) and $\Xi^{\sigma}|_{\iv^3 < 1/2} \equiv 0$ and
\begin{subequations}
\begin{align*}
\textstyle\sup_{|\xi_0| < 2|\frac{a}{\anisotropic}|}\;\sup_{\tau \in (-\infty,\b{T}(C))} \SUP^{(R-3)}_{\mathcal{O}(\xi_0,1,\tau)}\{\Xi^{\sigma}\}(\tau) & \,\leq\, \tfrac{1}{2}\, \smallRxx(R)\smallSHSprime(B) \\
\textstyle\sup_{|\xi_0| < 2|\frac{a}{\anisotropic}|}\;\sup_{\tau \in (-\infty,\b{T}(C))} |\tau|^{K+1}\,\SUP^{(R-3)}_{\mathcal{O}(\xi_0,1,\tau)}\{\Xi^{\sigma}\}(\tau) & \,\lesssim_{(Y,J_0)} \,1\\
\textstyle\sup_{\tau \in (-\infty,\b{T}(C))} |\tau|^{K+1}\,\SUP^{(R-3)}_{D_{4|\frac{a}{\anisotropic}|}(0)\times (0,1)} \{\Xi^{\sigma}\}(\tau) & \,\lesssim_{(Y,J_0)} \,1\\
\textstyle\text{$\tfrac{1}{2} \leq \b{B}^0(\iv,\Xi^{\sigma}(\iv)) \leq 2$ for all } \iv \in (-\infty,\b{T}(C)) &\times \R^2 \times (0,1)
\end{align*}
\end{subequations}
\noindent \emph{Step 10.} Suppose $\widetilde{\Xi}^{\sigma} \in C^1((-\infty,t_1] \times \R^2 \times [0,1], \pi^{-1}\range)$ with $t_1 < \b{T}(C)$ 
is a pair of solutions to $\text{\eqref{shsxi}}^{\sigma}$,
with  $\widetilde{\Xi}^{\sigma} = \FLIP\cdot \widetilde{\Xi}^{-\sigma}$ (when $\xi \neq 0$) and $\widetilde{\Xi}^{\sigma}|_{\iv^3 < 1/2} \equiv 0$ and 
\begin{equation}\label{rrft}
\textstyle\lim_{\tau \to -\infty}\, \sup_{|\xi_0| < 2|\frac{a}{\anisotropic}|}\, |\tau|^{J_0}\, \SUP^{(1)}_{\mathcal{O}(\xi_0,1,\tau)}\{\widetilde{\Xi}^{\sigma}\}(\tau) \, = \, 0
\end{equation}
Then $\widetilde{\Xi}^{\sigma}$ coincides (on its domain) with $\Xi^{\sigma}$ in Step 9.
\vskip 2mm

\noindent \emph{Proof of Step 1 through Step 10.}
To prove \emph{Step 2} use Proposition \ref{prop:MAINNEW} Part 1 with $T = \b{T}(C)$.
To prove \emph{Step 3}, introduce for each $\xi_0 \in \R^2$ and $t_0 < \b{T}(C)$ the set
$$\mathcal{J}(\xi_0,t_0) = \big\{t \in [t_0,t_1(t_0))\;\big|\; \textstyle\sup_{\tau \in [t_0,t]}\, E^R_{\mathcal{O}(\xi_0,2,\tau)}\{\Xi_{t_0}\}(\tau) \leq \big(\smallRxx(R)\smallSHSprime(B)\big)^2\big\}$$
It is an interval and closed as a subset of $[t_0,t_1(t_0))$.
Proposition \ref{prop:MAINNEW} Part 2 and \eqref{Tchoice} imply that $[t_0,t^{\ast}] \in \mathcal{J}(\xi_0,t_0)$ for some $t^{\ast} > t_0$. For every such $t^{\ast}$, the assumptions of Proposition \ref{prop:MAINNEW} Part 3 for  {\bf \small (Sys1)} are satisfied with $T = t^{\ast}$. By Conclusion 2 and
$$K+1 \geq J_0
\qquad
\smallSHSprime(B) \leq \smallEE(X)
\qquad
t^{\ast} < \b{T}(C) < -1/\smallSHSprime(B) \leq -1/\smallEE(X)$$
we can apply Proposition \ref{prop:fdsdsfkjhdskhds}. Together with Proposition \ref{prop:MAINNEW} Part 2, we obtain \eqref{eq:fkhjkhhgkhkfhgf} 
for all $\tau \in (t_0,t^{\ast})$. Recall $\bigSobolev(R) > 1$. 
Conclude that $\mathcal{J}(\xi_0,t_0)$ is also open as a subset of $[t_0,t_1(t_0))$, and $\mathcal{J}(\xi_0,t_0) = [t_0,t_1(t_0))$.
See Proposition \ref{prop:MAINNEW} Part 3 Conclusion 1.
Now $\text{({\bf\small Break})}_1$ and $\text{({\bf\small Break})}_2$ in Step 2 are excluded,
and $t_1(t_0) = \b{T}(C)$. 
To prove \emph{Step 5}, use a finite speed of propagation argument.
Namely, use Proposition \ref{prop:MAINNEW} Part 3 for {\bf \small (Sys2)} to show that
if $|\tfrac{a}{\anisotropic}|\leq |\xi_0| < 2|\tfrac{a}{\anisotropic}|$, then
$\Xi_{t_0}^{\sigma} = \FLIP\cdot \Xi_{t_0}^{-\sigma}$ on 
$\bigcup_{\tau \in [t_0,\b{T}(C))} \{\tau\}\times \mathcal{O}(\xi_0,1,\tau) \subset \mathcal{Z}(t_0)$.
To prove \emph{Step 7}, use Proposition  \ref{prop:MAINNEW} Part 3 for {\bf \small (Sys3)}, Conclusions 3 and 2.
To prove \emph{Step 8}, use Proposition  \ref{prop:MAINNEW} Part 3 for {\bf \small (Sys2)}, Conclusion 2.
To prove \emph{Step 9}, introduce
for every $\beta \in (0,1)$ the compact set $$\mathcal{X}_{\beta} = [\b{T}(C) - \beta^{-1},\b{T}(C)-\beta] \times \overline{D_{2|\frac{a}{\anisotropic}|}(0)} \times [0,1]$$
Set $\Xi^{\sigma}|_{\mathcal{X}_{\beta}} = \text{$L^2$-}\lim_{t \to -\infty} \Xi_{t}^{\sigma}$, see Step 8. By Step 6 and the Arzela Ascoli theorem, for each $\beta\in (0,1)$, we have
$\Xi^{\sigma}|_{\mathcal{X}_{\beta}} \in C^{R-3}(\mathcal{X}_{\beta},\pi^{-1}\range)$, and
$\Xi^{\sigma}_{t_n} \to \Xi^{\sigma}$ in $C^{R-3}(\mathcal{X}_{\beta},\pi^{-1}\range)$ for a sequence $t_n \to -\infty$. Recall $R - 3 \geq 1$. By construction,  $\Xi^{\sigma} = \FLIP\cdot \Xi^{-\sigma}$ on  $\mathcal{X}_{\beta}\cap (\FLIP\cdot \mathcal{X}_{\beta})$ for all $\beta \in (0,1)$. 
Hence, there is a unique pair of $\FLIP$-compatible extensions
$\Xi^{\sigma} \in C^{R-3}((-\infty,\b{T}(C))\times \R^2 \times [0,1], \pi^{-1}\range)$.
Step 7 implies $(\Xi^{\sigma})^{\sharp} \equiv 0$. Step 6 implies the first two estimates in Step 9. The third follows from the second, by using $\FLIP$-compatibility.
The fourth estimate in Step 9 follows from $\FLIP$-compatibility and $\text{\bf \small (RE3)}$ in Proposition \ref{prop:fdsdsfkjhdskhds}, see Conclusion 2 of Proposition  \ref{prop:MAINNEW} Part 3 for {\bf \small (Sys2)}.
For \emph{Step 10}, first use Proposition \ref{prop:MAINNEW} Part 3 for {\bf \small (Sys2)}, Conclusion 2, to show that $E^0_{\mathcal{O}(\xi_0,1,\tau)}\{\Xi^{\sigma} - \widetilde{\Xi}^{\sigma}\} (\tau) = 0$ for all $|\xi_0| < 2|\tfrac{a}{\mathfrak A}|$ and all sufficiently negative $\tau < t_1$. Show that this implies \emph{Step 10}.
\vskip 1mm

\noindent
To complete the proof of Theorem \ref{mainenergy}, observe that the $x$-set $\STRIP(1,\smallFinalPrime)$ is contained in the $\iv$-set $(-\infty,\b{T}(C)) \times \R^2 \times (0,1)$. Set $\Psi^{\sigma} = \Psi_K^{\sigma} + \pi\,\Xi^{\sigma}$ with $\Xi^{\sigma}$ as in Step 9. Equations \eqref{zztrt1}, \eqref{zztrt2} follow from the definition of $\Psi_K^{\sigma}$, and $\|\Psi^{\sigma}(k+1)\|_{C^{R+1}(\mathcal{C}(a,\mathfrak A,2))}\lesssim_Y 1$ for $0 \leq k \leq K$, if $\smallFinalPrime$ is sufficiently small depending only on $(Y,J_0)$. 
We have $\tfrac{1}{2} \leq e_3 \leq 2$ on $\STRIP(1,\smallFinalPrime)$ by the last estimate in Step 9, where $e_3$ is a component of $\Phi^{\sigma} = (e,\gamma,w)$.
Equation $L(\PhiOne_1\overline{\PhiOne}_2 - \overline{\PhiOne}_1\PhiOne_2) = -2\PhiTwo_2\,(\PhiOne_1\overline{\PhiOne}_2 - \overline{\PhiOne}_1\PhiOne_2)$
and $\Im \big(\PhiOne_1\overline{\PhiOne}_2\big)|_{\u{u} < 1/2} < 0$ imply  $\Im \big(\PhiOne_1\overline{\PhiOne}_2\big) < 0$ on $\STRIP(1,\smallFinalPrime)$. Existence in Theorem \ref{mainenergy} now follows from Proposition \ref{prop,fdldfkf}. Uniqueness follows from Step 10. \qed
\end{proof}
\section{Conclusions} \label{sec:conclusionsXX}
\begin{proposition}[Asymptotic expansion] \label{wjseXX} Let $\Phi^{\sigma} = \Minkowski_{a,\mathfrak A} + u^{-M}\Psi^{\sigma}$ be the pair 
of vacuum fields of Theorem \ref{mainenergy} for $K = 0$.
For each $L \geq 0$,
$$\textstyle\sup_{|\alpha| \leq R-3}\,\big\|\p^{\alpha}\big(\Psi^{\sigma}(\,\cdot\,,u) - \textstyle\sum_{k=0}^{L}\big(\tfrac{1}{u}\big)^k \Psi^{\sigma}(k)(\,\cdot\,)\big)\,\big\|_{C^0(\mathcal{C}(a,\mathfrak A,1))}  = \mathcal{O}\big(|u|^{-L-1}\big)$$
as $u\to -\infty$, with $\alpha \in \N_0^4$. In other words, the formal power series $[\, \Psi^{\sigma}\,]$ is an asymptotic expansion for $\Psi^{\sigma}$.
\end{proposition}
\begin{proof} The conditions on the data $a$, $\mathfrak A$, $\text{\bf \small DATA}^{\sigma}$, $R$, $\epsilon$ in Theorem \ref{mainenergy} are independent of $K$. Therefore,
Theorem \ref{mainenergy} can be applied with the same data for all $K \geq 0$. \qed
\end{proof}
\begin{remark} \label{fjdhgjdsjgdjsjsdgj}
The smallness assumptions of Theorem \ref{mainenergy}, with $a = \anisotropic$ and $\epsilon = \tfrac{1}{4}$ and $R = 4$ and $K = 0$, are 
$\textstyle\max_{\sigma \in \{-,+\}}\, \|\text{\bf \small DATA}^{\sigma}\|_{C^{8}(D_4(0)\times (0,2))} \leq \smallFinal$
and $0 < |\anisotropic|\leq  \smallFinal$
for a universal constant $\smallFinal \in (0,1)$.
The domain $\STRIP(1,\smallFinalPrime)$ of the vacuum field in Theorem \ref{mainenergy} depends only on $\textstyle\max_{\sigma \in \{-,+\}}\, \|\text{\bf \small DATA}^{\sigma}\|_{C^{10}(D_4(0)\times (0,2))}$.
\end{remark}
\begin{proposition}[Trapped sphere]
Let $\Phi^{\sigma}$ be the pair 
of vacuum fields of Theorem \ref{mainenergy}, in the context of Remark \ref{fjdhgjdsjgdjsjsdgj}. Suppose there is a $\u{u}_1 \in (0,1)$ such that $\Lambda > 0$, where
$$\Lambda = \textstyle\min_{\sigma \in \{-,+\}}\; \inf_{\xi \in D_4(0)} \; \int_0^{\u{u}_1}\dd \u{u}\; |\text{\bf \small DATA}^{\sigma}(\xi,\u{u})|^2$$
If $0 < |\anisotropic| <  \tfrac{1}{4}\sqrt{\mathbf{c}} \min\{1,\Lambda\}$, then
$u_1=-\tfrac{1}{2}\Lambda\,\anisotropic^{-2} \in (-\infty,-\mathbf{c}^{-1})$, and the intersection of $\u{u}=\u{u}_1$ and $u=u_1$  is a trapped sphere.
\end{proposition}
\begin{proof}
Suppress $\sigma\in \{-,+\}$.
By Remark \ref{trappedpenrose}, one must check $\gamma_2$, $\gamma_6 < 0$ on $(\u{u},u)=(\u{u}_1,u_1)$.
We only discuss the more subtle $\gamma_2 < 0$.
Recall $\gamma_2 = \anisotropic^2 (\anisotropic^2 \u{u} - u)^{-1} + u^{-2}\omega_2$.
By
\eqref{zztrt2}, $|\omega_2 - \omega_2(0)| \leq \mathbf{c}^{-1}|u|^{-1}$ when
$(\xi,\u{u},u) \in D_4(0) \times (0,1) \times (-\infty,-\mathbf{c}^{-1})$.
By Proposition \ref{prop:fps}, $\omega_2(0)(\xi,\u{u}) = - \int_0^{\u{u}} \dd \u{u}'\, |\text{\bf \small DATA}(\xi,\u{u}')|^2$. If
$(\xi,\u{u},u) \in D_4(0) \times \{\u{u}_1\} \times (-\infty,-\mathbf{c}^{-1})$
then $\omega_2(0) \leq - \Lambda$ and
$\omega_2 \leq -\Lambda + \mathbf{c}^{-1}|u|^{-1}$
and $\gamma_2 \leq 
|u|^{-1} \anisotropic^2 - |u|^{-2}\Lambda + |u|^{-3} \mathbf{c}^{-1}
$.
If $|u| = \tfrac{1}{2}\Lambda \anisotropic^{-2}$
and 
$0 < |\anisotropic| <  \tfrac{1}{4}\sqrt{\mathbf{c}} \Lambda$ then
$\gamma_2 < 0$.
\qed
\end{proof}
\section{Three Points of View}\label{subsec:tpovXXXXX}
Vacuum fields in Theorem \ref{mainenergy} are, by definition, in the
Regularized Picture R.
We use symmetry transformations (see Section \ref{sec:symmetries})
to introduce two more pictures:
\begin{center}
\hskip-25mm\text{\input{scalings.pstex_t}}
\end{center}
The transformations depend on $a$, $\anisotropic$ from
Theorem \ref{mainenergy}.
They are global conformal transformations
of the Lorentzian metric.  The Minkowski boundary $\mathcal{M}_X$, the $\xi$-disk corresponding to one hemisphere of $S^2$ in $\mathcal{M}_X$, the initial data $\text{\bf \small DATA}^{\sigma}_{X}$, and the domain of the vacuum field
$\Phi_X^{\sigma} = \mathcal{M}_X + u^{-M} \Psi_X^{\sigma}$ are, for $X = \text{R},\text{H},\text{F}$:
\begin{center}
\begin{tabular}{c||c|c|c}
 \rule{0pt}{10pt}$X$ & R & H & F \\
\hline
$\mathcal{M}_X$ \rule{0pt}{10pt}& $\Minkowski_{a,\mathfrak A}$ & $\Minkowski_{1,1}$ & $\Minkowski_{1,1}$\\
 Hemisphere (in $\mathcal{M}_X$) \rule{0pt}{10pt} & $|\xi| < |\frac{a}{\mathfrak A}|$ & $|\xi| < 1$ & $|\xi| < 1$\\
$\text{\bf \small DATA}^{\sigma}_{X}(\xi,\u{u})$
\rule{0pt}{10pt} & 
$\datafunc^{\sigma}\big(\xi,\,\u{u}\big)$ & ${\mathfrak A}^{-2}\, \datafunc^{\sigma}\big(\frac{a}{\mathfrak A}\xi,\, \u{u}\big)$ & ${\mathfrak A}^{-2}\, \datafunc^{\sigma}\big(\frac{a}{\mathfrak A}\xi,\, {\mathfrak A}^{-4} \u{u}\big)$\\
 Domain \rule{0pt}{10pt}& $\STRIP\big(1,\smallFinalPrime\big)$ & $\STRIP\big(1,\smallFinalPrime\,{\mathfrak A}^2\big)$ & $\STRIP\big({\mathfrak A}^4,\smallFinalPrime\,{\mathfrak A}^{-2}\big)$
\end{tabular}
\end{center}
Column $\text{R}$ corresponds to Theorem \ref{mainenergy}. The entry
$\text{\bf \small DATA}^{\sigma}_{\text{R}} = \eta^{\sigma}$ is by definition of $\eta^{\sigma}$. The columns $\text{H}, \text{F}$ are obtained from $\text{R}$, as indicated above, by scaling. Here,
$$\STRIP(\mu,\lambda)
\,=\,   \R^2 \times (0,\mu) \times (-\infty\,,\,-\lambda^{-1}\,)
\qquad 
\mu,\lambda > 0$$
\begin{remark} \label{rem:oooddsssXXXXXXXXXXXX}
From this perspective, Christodoulou \cite{Chr} investigates the case $a = {\mathfrak A}$
 in the picture F. His small $\delta > 0$ is to be identified  with our ${\mathfrak A}^4$.
His ``short pulse ansatz'' corresponds to $\text{\bf \small DATA}^{\sigma}_\text{F}(\xi,\u{u}) =
{\mathfrak A}^{-2} \datafunc^{\sigma}\big(\xi, {\mathfrak A}^{-4} \u{u}\big)$.
\end{remark}
\begin{remark}
In the picture H, the Minkowski boundary is always $\mathcal{M}_{1,1}$, and
the domain of the vacuum field always has $\u{u}\in (0,1)$. Therefore, different triples $(a,\anisotropic,\eta^{\sigma})$ that correspond to the same $\text{\bf \small DATA}^{\sigma}_\text{H}(\xi,\u{u})$ also correspond to the same vacuum field $\Phi_\text{H}^\sigma$ (by uniqueness in  Theorem \ref{mainenergy}). This must be taken into account when interpreting the smallness assumptions in Theorem \ref{mainenergy}.
\end{remark}

%% file: scalings.pstex_t
\begin{picture}(0,0)%
\includegraphics{scalings.pstex}%
\end{picture}%
\setlength{\unitlength}{3039sp}%
\begingroup\makeatletter\ifx\SetFigFont\undefined%
\gdef\SetFigFont#1#2#3#4#5{%
  \reset@font\fontsize{#1}{#2pt}%
  \fontfamily{#3}\fontseries{#4}\fontshape{#5}%
  \selectfont}%
\fi\endgroup%
\begin{picture}(630,2340)(5386,-4330)
\put(6001,-3811){\makebox(0,0)[lb]{\smash{{\SetFigFont{9}{10.8}{\familydefault}{\mddefault}{\updefault}{\color[rgb]{0,0,0}$\text{with scaling constant ${\mathfrak J} = {\mathfrak A}^4$}$}%
}}}}
\put(6001,-2611){\makebox(0,0)[lb]{\smash{{\SetFigFont{9}{10.8}{\familydefault}{\mddefault}{\updefault}{\color[rgb]{0,0,0}$\text{Field transformation ${\mathfrak C}\circ {\mathfrak A}$}$}%
}}}}
\put(6001,-3586){\makebox(0,0)[lb]{\smash{{\SetFigFont{9}{10.8}{\familydefault}{\mddefault}{\updefault}{\color[rgb]{0,0,0}$\text{Isotropic scaling transformation ${\mathfrak J}$}$}%
}}}}
\put(6001,-2836){\makebox(0,0)[lb]{\smash{{\SetFigFont{9}{10.8}{\familydefault}{\mddefault}{\updefault}{\color[rgb]{0,0,0}$\text{where ${\mathfrak C}(\xi) = a\,\xi$}$}%
}}}}
\put(5401,-2161){\makebox(0,0)[lb]{\smash{{\SetFigFont{9}{10.8}{\familydefault}{\mddefault}{\updefault}{\color[rgb]{0,0,0}$\text{{\bf Regularized Picture} (R)}$}%
}}}}
\put(5401,-3211){\makebox(0,0)[lb]{\smash{{\SetFigFont{9}{10.8}{\familydefault}{\mddefault}{\updefault}{\color[rgb]{0,0,0}$\text{{\bf High Amplitude Picture} (H)}$}%
}}}}
\put(5401,-4261){\makebox(0,0)[lb]{\smash{{\SetFigFont{9}{10.8}{\familydefault}{\mddefault}{\updefault}{\color[rgb]{0,0,0}$\text{{\bf Finite Mass Picture} (F)}$}%
}}}}
\end{picture}%